\documentclass{article}
\usepackage{xcolor}
\usepackage{graphicx}
\usepackage[truedimen,margin=30mm]{geometry} 
\usepackage{bm}
\usepackage{natbib}
\usepackage{amsmath,amssymb,amsthm}
\usepackage{mathrsfs}
\usepackage{float}
\usepackage{multirow}
\usepackage{booktabs}
\usepackage{appendix}
\usepackage{hyperref}
\usepackage{subcaption}
\usepackage{authblk}

\theoremstyle{plain}
\newtheorem{proposition}{Proposition}
\newtheorem{theorem}{Theorem}
\newtheorem{lemma}{Lemma}[section]
\newtheorem{corollary}{Corollary}
\newtheorem{assumption}{Assumption}

\theoremstyle{definition}
\newtheorem{definition}{Definition}
\newtheorem{remark}{Remark}

\title{\textbf{
Higher-Order Risk Comparisons of Oracle-Calibrated Global–Local Priors}}
\date{\today}
\author[1]{Dongu Han\thanks{Corresponding author: hdongu94@gmail.com}}

\affil[1]{School of Commerce, Meiji University, Tokyo, Japan}

\begin{document}

\maketitle

\begin{abstract}
First-order minimaxity offers limited guidance for choosing among global–local shrinkage priors, since different local-tail distributions can yield the same leading risk. Comparisons at a common global scale may also be misleading, as the scale interacts with the shape of the local prior. We study the worst-case squared-error risks of posterior means in sparse Gaussian normal-means estimation, allowing each prior its own risk-minimizing global scale. For local-variance priors with polynomial tails and logarithmic corrections, we characterize the calibrations that retain first- and second-order minimaxity and derive the asymptotic form of the unrestricted oracle scale. Even after this optimization, differences in risk remain. The polynomial tail index, logarithmic tail exponent and a scale-invariant higher-order correction distinguish priors at successively finer orders. In particular, the horseshoe has smaller asymptotic oracle risk than the horseshoe+, despite the latter's additional logarithmic tail factor. We also establish the calibration accuracy needed to preserve these differences and sufficient conditions for their stability under high-dimensional Gaussian regression designs.
\end{abstract}

\textbf{keywords: }Bayesian shrinkage; Global--local prior; Higher-order asymptotics; Minimax risk; Normal means; Sparsity.

\section{Introduction}

Sparse Bayesian estimation has developed through both priors that explicitly represent sparse model structure, such as spike-and-slab constructions, and continuous shrinkage priors that approximate sparsity without assigning positive mass to exact zeros \citep{Castillo2012,Rockova2018}. Among the latter, global--local priors combine strong shrinkage of small coefficients with heavy tails that protect large signals, while retaining a convenient continuous parameterization \citep{Carvalho2010,Polson2011}. Constructions including the normal--gamma, generalized beta, horseshoe, Dirichlet--Laplace and horseshoe+ priors exhibit substantial variation in local-tail behaviour and global shrinkage \citep{Griffin2010,Armagan2011,Carvalho2010,Bhattacharya2015,Bhadra2017}. This diversity raises a prior-selection question that leading-order optimality alone cannot resolve: how should one distinguish continuous shrinkage priors that are already equivalent at the sparse minimax order?

Classical sparse-estimation theory distinguishes minimax rates from sharper risk benchmarks. For nearly black normal-mean classes and broader \(\ell_p\)-type spaces, the leading minimax risks were characterized in \citep{Donoho1992,Donoho1994}, while \citet{Johnstone1994} obtained a refined minimax expansion showing that first-order-equivalent procedures can differ at the next asymptotic order. Adaptation to unknown sparsity has also been studied extensively: empirical-Bayes thresholding and shrinkage can recover sharp minimax benchmarks \citep{Johnstone2004,Zhang2005,Jiang2009}, as can procedures based on false-discovery-rate thresholding and sorted \(\ell_1\) regularization \citep{Abramovich2006,Su2016}. These developments establish the relevance of sharper risk comparisons beyond first-order rates, but do not determine how the tail structures of different global--local priors affect their higher-order risks.

For continuous global--local shrinkage, the sparse normal-means model has provided the principal setting for first-order optimality theory. The horseshoe posterior mean attains the minimax squared-error rate under suitable global-scale calibration \citep{VanDerPas2014}, and general conditions on local-scale priors ensure minimax posterior contraction \citep{VanDerPas2016}. For broader classes of regularly varying one-group priors, \citet{Ghosh2017} established asymptotic minimaxity of posterior means, including the exact leading constant for an important horseshoe-type subclass, while \citet{Qin2024} extended exact asymptotic minimax point-estimation results to additional global--local priors and the Dirichlet--Laplace construction. Empirical-Bayes and hierarchical global-scale specifications further permit adaptation to unknown sparsity \citep{VanDerPas2017,Paul2025a}. Collectively, these results provide a substantial first-order theory, but leave unresolved the risk ordering among priors after their leading minimax performance has been matched.

Finer aspects of shrinkage-prior behaviour have received increasing attention. Regular variation characterizes the interaction between concentration near zero and heavy tails \citep{Bhadra2016}, while tail modifications can affect large-signal robustness and testing performance \citep{Bhadra2017,Hamura2026}. Polynomial tail order also influences sharp contraction properties \citep{Song2020}, and tail-adaptive priors permit data-dependent adjustment of tail behaviour \citep{Lee2024}. Global-scale calibration provides another source of regularization, since incorporating sparsity information can substantially change shrinkage even for a fixed local prior \citep{Piironen2017}. Related multiple-testing theory establishes asymptotic Bayes and sharp minimax properties for the horseshoe and broader one-group procedures under testing losses \citep{Bogdan2011,Datta2013,Ghosh2016,Paul2025b}. These results show that prior tails and global scales are statistically consequential, but improvements under large-signal, contraction or testing criteria need not imply an ordering under sparse worst-case squared-error loss. To our knowledge, existing theory does not identify the first separating order of the worst-case risks of different global--local posterior means after each prior has been placed on its own risk-optimal global scale, or characterize the successive tail features governing that separation. The central question is therefore which prior characteristics remain distinguishable after oracle calibration, at what risk resolutions, and with what calibration accuracy.

Extending such comparisons to high-dimensional regression introduces an additional difficulty through design dependence. Restricted-isometry and restricted-eigenvalue conditions control sparse geometry \citep{Candes2005,Bickel2009,Raskutti2010} and underpin minimax regression theory \citep{Raskutti2011}. Bayesian regression results similarly establish concentration, prediction and estimation guarantees under sparse-design conditions for discrete and continuous shrinkage priors \citep{Castillo2015,Armagan2013a,Bhadra2019,Song2023}. Our objective differs from these first-order guarantees: we ask how closely a regression design must approximate orthogonality for a finer, prior-specific risk separation from the Gaussian sequence model to remain visible. Sparse design geometry thus becomes a resolution requirement for transferring higher-order comparisons.

We study posterior-mean estimation under squared-error loss over nearly black normal-mean classes and a broad log-regular family of global--local priors whose local-variance tails combine polynomial decay with logarithmic modification. The analysis gives a four-part answer. First, we derive a sharp calibration phase diagram identifying how the polynomial tail index and global scale must be coupled to attain first- and second-order minimaxity, separating second-order-optimal, first-order-only and suboptimal regimes. Second, optimizing over all positive global scales yields the unrestricted prior-specific oracle calibration, together with a quadratic excess-risk penalty \(4\alpha^2st^2\{1+o(1)\}\) for local multiplicative perturbations \(\tau=\tau_\pi^{\mathrm{or}}e^t\). Third, after oracle calibration removes leading scale differences, the remaining risks are resolved successively by the polynomial tail index \(\alpha\), logarithmic tail exponent \(\gamma\) and scale-invariant next-order coefficient \(\zeta_\pi\), at first separating orders \(s/L\), \(s/L^2\) and \(s/L^3\), respectively. This lexicographic hierarchy implies, in particular, that the horseshoe has smaller oracle worst-case risk than the horseshoe+ at their first asymptotically distinguishable order. Fourth, we transfer the hierarchy to high-dimensional Gaussian linear regression: under sparse near-orthogonality, the regression risk differs from its sequence-model counterpart by \(O[sL^2\{\gamma_n+\delta_n^2\}]\), so a separation of order \(s/L^k\) is preserved under the sufficient design-resolution condition \(\gamma_n+\delta_n^2=o(L^{-k-2})\).

These results identify first-order minimaxity as a screening principle rather than a complete basis for prior choice. Once leading performance has been matched, prior shape determines the first separating risk order, while calibration and design perturbations determine whether that separation remains visible. This resolution principle is specific to posterior means under sparse worst-case squared-error loss and does not imply a universal prior ranking across other inferential criteria.

\section{Sparse normal means and log-regular shrinkage priors}

\subsection{Sparse normal means and risk}
\label{sec:normal-means}

We observe independent coordinates from the Gaussian sequence model
\[
Y_i=\theta_i+\varepsilon_i,\quad \varepsilon_i\sim N(0,1)\quad (i=1,\ldots,p),
\]
where the unknown mean vector $\theta=(\theta_1,\ldots,\theta_p)^T$ is sparse. We consider the nearly black parameter class
\[
\Theta_{p,s}=\{\theta\in\mathbb R^p:\|\theta\|_0\leq s\},\quad s/p\to0,\quad L=\log(p/s)\to\infty.
\]
Thus $s$ bounds the number of nonzero coordinates, while $L$ is the logarithmic sparsity scale used throughout the asymptotic analysis.

For an estimator $T=T(Y)$, let $R_{p,s}(T)$ denote its worst-case squared-error risk over $\Theta_{p,s}$. The corresponding minimax risk satisfies
\begin{equation}
\label{eq:minimax-risk}
R_{p,s}(T)=\sup_{\theta\in\Theta_{p,s}}E_\theta\{\|T-\theta\|_2^2\},\quad R^*_{p,s}=\inf_T R_{p,s}(T)=2sL\{1+o(1)\};
\end{equation}
see \cite{Donoho1992,Johnstone1994}. We study posterior-mean estimators under continuous global--local shrinkage priors and compare their risks with the benchmark in \eqref{eq:minimax-risk}. The leading term $2sL$ determines first-order minimaxity, while the subsequent analysis resolves distinctions that remain after this dominant sparse-estimation difficulty has been matched.

\subsection{Log-regular global--local priors}
\label{sec:log-regular}

We place independent global--local priors on the coordinates of $\theta$. Writing $u_i$ for the squared local scale, the hierarchy is
\begin{equation}
\label{eq:gl-prior}
\theta_i\mid u_i,\tau\sim N(0,\tau^2u_i),\quad u_i\sim\pi\quad (i=1,\ldots,p),
\end{equation}
where $\tau>0$ is the global scale and $\pi$ is a density on $(0,\infty)$. Throughout, tail conditions are stated for the local variance $u_i$, rather than for its square root. For a given $\pi$ and $\tau$, let $\widehat\theta_{\pi,\tau}=E(\theta\mid Y)$ denote the posterior mean under \eqref{eq:gl-prior}, and write $R_{p,s}(\pi,\tau)=R_{p,s}(\widehat\theta_{\pi,\tau})$.

For a scalar observation $y$, write $\widehat\theta_{\pi,\tau}(y)=y\,w(y,\tau)$, where $w(y,\tau)=1-E\{(1+\tau^2u)^{-1}\mid y\}$ is the posterior nonshrinkage weight, with dependence on $\pi$ suppressed in the notation. In the small-$\tau$ regimes considered below, the equation $w(a_\tau,\tau)=1/2$ has a unique positive solution for all sufficiently large $p$; we call $a_\tau$ the intrinsic transition threshold. It marks the point at which the posterior mean retains one half of the observed signal and provides the natural coordinate for the calibration and oracle-risk expansions in Section~\ref{sec:risk-expansion}.

Our analysis distinguishes priors not only through their polynomial tail order but also through logarithmic components that are invisible to first-order regular-variation theory. We use the following log-regular class.

\begin{assumption}[Log-regular local scales]
\label{ass:log-regular}
The density $\pi$ is positive, proper and three times continuously differentiable on $(0,\infty)$. For some $C_\pi>0$, $\alpha>0$ and $\gamma,\beta,\chi\in\mathbb R$,
\begin{equation}
\label{eq:log-regular}
\pi(u)=C_\pi u^{-\alpha-1}(\log u)^\gamma\left\{1+\frac{\beta}{\log u}+\frac{\chi}{(\log u)^2}+\rho(u)\right\}\quad (u\to\infty).
\end{equation}
Writing $D=u\,d/du$, the remainder satisfies
\[
D^j\rho(u)=O\{(\log u)^{-3-j}\}\quad (j=0,1,2,3).
\]
At the origin, $\pi$ is regularly varying with index $\delta_0-1$ for some $\delta_0>0$; equivalently,
\[
\pi(u)=u^{\delta_0-1}\ell_0(1/u)\quad (u\downarrow0),
\]
where $\ell_0$ is eventually positive and slowly varying at infinity. The comparisons below are made over a fixed class for which $\alpha\in[\underline\alpha,\overline\alpha]\subset(0,2)$, $\delta_0\in[\underline\delta,\overline\delta]\subset(0,\infty)$, $|\gamma|+|\beta|+|\chi|$ is uniformly bounded, $C_\pi$ is bounded away from zero and infinity, and the tail-remainder bounds and origin regular variation above hold uniformly.
\end{assumption}

The parameter $\alpha$ determines the polynomial upper-tail order, whereas $\gamma$ records the leading logarithmic modification and $\beta$ and $\chi$ describe finer slowly varying corrections. At the opposite end of the scale, $\delta_0$ determines the regular-variation index at the origin, while $\ell_0$ allows slowly varying modifications there. The upper-tail specification therefore retains distinctions that disappear if the slowly varying component is replaced by an asymptotically constant factor, whereas the origin condition rules out more singular behaviour that could generate competing large-global-scale regimes; the higher-order comparisons below are governed by the upper-tail parameters. The familiar horseshoe prior provides a basic example: its local-variance density is $\pi(u)=\{\pi u^{1/2}(1+u)\}^{-1}$, and hence $\alpha=1/2$, $\gamma=0$ and $\delta_0=1/2$. The horseshoe+ has the same polynomial indices at both ends but an additional logarithmic factor in its upper tail, corresponding to $\alpha=1/2$, $\gamma=1$ and $\delta_0=1/2$ at leading order \cite{Carvalho2010,Bhadra2017}. Thus the two priors are indistinguishable by polynomial tail order alone but are separated by the logarithmic upper-tail structure retained in Assumption~\ref{ass:log-regular}.

To characterize how tail shape interacts with global shrinkage, let $\epsilon=s/p$ and consider the scale family
\begin{equation}
\label{eq:global-calibration}
\tau=c\epsilon^rL^b,
\end{equation}
where $c>0$, $r>0$ and $b\in\mathbb R$ are fixed as $p\to\infty$. The exponent $r$ controls the polynomial calibration to sparsity, $b$ gives its logarithmic refinement and $c$ determines the remaining constant scale. This parameterization is used to characterize the calibration phase diagram in Section~\ref{sec:risk-expansion}; the oracle analysis in Section~\ref{sec:optimal-calibration} subsequently removes this restriction and minimizes the worst-case risk over all $\tau>0$.

\subsection{Orders of optimality and risk comparison}
\label{sec:higher-order}

The leading benchmark in \eqref{eq:minimax-risk} distinguishes procedures only at the dominant sparse-minimax order. We call a prior--scale specification $(\pi,\tau)$ first-order minimax if
\[
R_{p,s}(\pi,\tau)=2sL\{1+o(1)\}.
\]
The classical minimax risk admits the sharper expansion
\begin{equation}
\label{eq:second-order-minimax}
R^*_{p,s}=2sL-2\sqrt{2}\,s\{L\log L\}^{1/2}+o\{s(L\log L)^{1/2}\};
\end{equation}
see \cite{Johnstone1994}. We call $(\pi,\tau)$ second-order minimax if its worst-case risk reproduces this expansion,
\[
R_{p,s}(\pi,\tau)=2sL-2\sqrt{2}\,s\{L\log L\}^{1/2}+o\{s(L\log L)^{1/2}\}.
\]
Thus first-order minimaxity requires agreement only at order $sL$, whereas second-order minimaxity additionally requires agreement at the classical scale $s\{L\log L\}^{1/2}$.

The comparisons of oracle-calibrated priors later in the paper take place on substantially smaller scales. Since two risks can share several terms larger than the first term at which they differ, it is useful to formulate these comparisons directly through their risk difference.

\begin{definition}[Higher-order comparison]
\label{def:higher-order}
Let $(\pi_1,\tau_1)$ and $(\pi_2,\tau_2)$ be first-order minimax specifications, and let $a_{p,s}>0$ satisfy $a_{p,s}=o(sL)$. If, for some nonzero constant $A$,
\[
R_{p,s}(\pi_2,\tau_2)-R_{p,s}(\pi_1,\tau_1)=A\,a_{p,s}+o(a_{p,s}),
\]
we call $a_{p,s}$ their first separating order. If $A>0$, $(\pi_1,\tau_1)$ has smaller risk at that order; if $A<0$, $(\pi_2,\tau_2)$ has smaller risk.
\end{definition}

Definition~\ref{def:higher-order} permits common terms in the two risks to cancel before their first difference appears. In particular, the scale $s\{L\log L\}^{1/2}$ in \eqref{eq:second-order-minimax} is the second-order minimax scale, whereas the smaller scales arising in comparisons between oracle-calibrated priors are referred to as higher-order separating scales. This distinction avoids using ``second order'' for risk differences that occur only after the classical first two minimax terms have already been matched. The next section first determines the calibration regions that preserve the first- and second-order minimax benchmarks, then identifies the unrestricted oracle global scale and resolves the first separating orders among the resulting oracle-calibrated priors. It finally examines when this hierarchy remains stable in high-dimensional Gaussian linear regression.

\section{Higher-order risk, calibration and stability}

\subsection{Sharp calibration phase diagram}
\label{sec:risk-expansion}

The interaction between the local tail and the global scale is summarized by the effective logarithmic index $d=2\alpha b+\gamma$. The polynomial exponent $r$ determines whether the shrinkage transition occurs on the sparse-minimax scale, whereas $d$ governs the finer contribution of the null coordinates once that scale has been matched. The following result identifies exactly when a log-regular prior preserves this benchmark and when its own shrinkage geometry becomes visible at a larger order.

\begin{theorem}[Sharp calibration phase diagram]
\label{thm:risk-expansion}
Suppose Assumption~\ref{ass:log-regular} holds and let the global scale satisfy \eqref{eq:global-calibration}. If $r\neq1/(2\alpha)$, the posterior mean is not first-order minimax. If $r=1/(2\alpha)$, define $d=2\alpha b+\gamma$ and
\[
A_{\pi,c}=C_\pi\Gamma(\alpha+1/2)c^{2\alpha}\alpha^{-\gamma}.
\]
Then the effective transition threshold satisfies
\begin{equation}
\label{eq:critical-threshold}
a_\tau^2=2L+(2\alpha+1-2d)\log L-2\log A_{\pi,c}+o(1),
\end{equation}
and the aggregate null-coordinate risk satisfies
\begin{equation}
\label{eq:null-risk}
(p-s)R_0(\pi,\tau)=\frac{2A_{\pi,c}}{\sqrt{\pi}}\,sL^{d-\alpha}\{1+o(1)\},
\end{equation}
where $R_0(\pi,\tau)=E_0\{\widehat\theta_{\pi,\tau}(Y)^2\}$ is the risk of one null coordinate. Consequently, the following regimes hold:
\[
\begin{array}{ll}
d\leq\alpha+1/2, & R_{p,s}(\pi,\tau)=2sL-2\sqrt{2}\,s\{L\log L\}^{1/2}+o\{s(L\log L)^{1/2}\},\\[3pt]
\alpha+1/2<d<\alpha+1, & R_{p,s}(\pi,\tau)=2sL+O\{sL^{d-\alpha}\}\ \text{and the expansion in \eqref{eq:second-order-minimax} is not attained},\\[3pt]
d=\alpha+1, & R_{p,s}(\pi,\tau)=\left(2+\frac{2A_{\pi,c}}{\sqrt{\pi}}\right)sL\{1+o(1)\},\\[3pt]
d>\alpha+1, & R_{p,s}(\pi,\tau)/(sL)\longrightarrow\infty.
\end{array}
\]
Thus first-order minimaxity holds precisely when $r=1/(2\alpha)$ and $d<\alpha+1$, while second-order minimaxity holds precisely within the sharper region $d\leq\alpha+1/2$.
\end{theorem}

Theorem~\ref{thm:risk-expansion} exhibits two distinct calibration boundaries. The condition $r=1/(2\alpha)$ first matches the polynomial tail probability with the sparsity level; a different value of $r$ either leaves too much aggregate risk among the null coordinates or moves the active-coordinate transition beyond the minimax scale. On the critical polynomial scale, the logarithmic tail exponent $\gamma$ and the logarithmic global-scale exponent $b$ enter only through $d=2\alpha b+\gamma$. The boundary $d=\alpha+1/2$ determines whether the classical second-order term in \eqref{eq:second-order-minimax} is preserved, whereas $d=\alpha+1$ determines whether the prior-induced null risk reaches the leading order $sL$.

The exact null contribution in \eqref{eq:null-risk} is what makes these boundaries sharp rather than consequences of an upper bound. In particular, when $\alpha+1/2<d<\alpha+1$, the posterior mean remains first-order minimax even though the additional null risk dominates the classical second-order scale, while at $d=\alpha+1$ the same contribution changes the leading risk constant. The next subsection asks where the global scale is located when it is allowed to optimize this trade-off rather than being restricted to a prespecified logarithmic calibration.

\subsection{Unrestricted oracle calibration}
\label{sec:optimal-calibration}

The phase diagram in Theorem~\ref{thm:risk-expansion} identifies a range of global scales that preserve first-order, and in part second-order, minimaxity, but it does not determine which scale minimizes the worst-case risk within that range. We therefore remove the parametric restriction in \eqref{eq:global-calibration} and optimize directly over all $\tau>0$. For a fixed log-regular prior $\pi$, let $\tau_\pi^{\mathrm{or}}$ denote a risk-minimizing global scale and write
\[
\tau_\pi^{\mathrm{or}}\in\arg\min_{\tau>0}R_{p,s}(\pi,\tau),
\qquad
R_{p,s}^{\mathrm{or}}(\pi)=R_{p,s}(\pi,\tau_\pi^{\mathrm{or}})=\inf_{\tau>0}R_{p,s}(\pi,\tau).
\]
The next result shows that the unrestricted oracle is asymptotically localized to a unique calibration determined by the tail of $\pi$.

\begin{theorem}[Unrestricted oracle calibration]
\label{thm:optimal-calibration}
Suppose Assumption~\ref{ass:log-regular} holds. Then every risk-minimizing global scale satisfies
\begin{equation}
\label{eq:oracle-calibration}
\frac{\tau_\pi^{\mathrm{or}}}{\tau_\pi^\star}\longrightarrow1,
\qquad
\tau_\pi^\star=
\left\{\frac{\sqrt{\pi}\,\alpha^\gamma}{C_\pi\Gamma(\alpha+1/2)}\right\}^{1/(2\alpha)}
(s/p)^{1/(2\alpha)}L^{(\alpha-\gamma)/(2\alpha)}.
\end{equation}
Equivalently, when $\tau_\pi^\star$ is written in the form $c_\pi^\star(s/p)^{r^\star}L^{b^\star}$, its calibration indices are
\[
r^\star=\frac{1}{2\alpha},\qquad d^\star=2\alpha b^\star+\gamma=\alpha,\qquad b^\star=\frac{\alpha-\gamma}{2\alpha},
\]
with
\[
c_\pi^\star=\left\{\frac{\sqrt{\pi}\,\alpha^\gamma}{C_\pi\Gamma(\alpha+1/2)}\right\}^{1/(2\alpha)}.
\]
Moreover, the effective transition threshold at the oracle scale satisfies
\begin{equation}
\label{eq:oracle-threshold}
a_{\tau_\pi^{\mathrm{or}}}^2=2L+\log L-\log\pi+o(1).
\end{equation}
Thus priors with different polynomial and logarithmic tails generally require different oracle global scales, but their optimally calibrated posterior means share the same asymptotic transition threshold.
\end{theorem}

Theorem~\ref{thm:optimal-calibration} strengthens the phase-diagram result in two ways. First, the calibration in \eqref{eq:oracle-calibration} is not obtained by optimizing within a prespecified family of powers of $s/p$ and $L$; it is the asymptotic form of the unrestricted risk minimizer over $\tau>0$. Second, the oracle lies strictly inside the second-order-optimal region of Theorem~\ref{thm:risk-expansion}, since $d^\star=\alpha<\alpha+1/2$. The local tail and the global scale therefore cannot be optimized independently: a larger logarithmic tail exponent $\gamma$ is offset by a smaller logarithmic exponent $b^\star$, while the normalization $C_\pi$ is absorbed into the constant $c_\pi^\star$.

The common threshold in \eqref{eq:oracle-threshold} provides the natural basis for comparing different prior shapes. Before calibration, two priors can have different risks simply because their shrinkage transitions occur at different locations. Oracle calibration removes this leading scale mismatch by aligning the transitions at $2L+\log L-\log\pi+o(1)$, so that any remaining risk difference reflects finer features of the local-scale distribution rather than an inappropriate choice of the global scale.

\begin{proposition}[Local sensitivity of oracle calibration]
\label{prop:oracle-sensitivity}
Suppose Assumption~\ref{ass:log-regular} holds, and let $\tau_\pi^{\mathrm{or}}$ be a risk-minimizing global scale. If $t=t_{p,s}\to0$, then
\[
R_{p,s}\{\pi,\tau_\pi^{\mathrm{or}}e^t\}-R_{p,s}^{\mathrm{or}}(\pi)=4\alpha^2st^2\{1+o(1)\}.
\]
Consequently, if two oracle-calibrated priors first separate at order $s/L^k$, replacing their oracle scales by $\widetilde\tau_j=\tau_{\pi_j}^{\mathrm{or}}e^{t_j}$ preserves the first separating coefficient whenever $\max_j|t_j|=o(L^{-k/2})$. A perturbation of exact order $L^{-k/2}$, by contrast, contributes to risk at order $s/L^k$ and can therefore alter that coefficient.
\end{proposition}

Proposition~\ref{prop:oracle-sensitivity} quantifies the local resolution of the oracle calibration. For the three separating scales in Theorem~\ref{thm:oracle-ranking}, the required log-scale calibration errors are respectively $o(L^{-1/2})$, $o(L^{-1})$ and $o(L^{-3/2})$. These requirements are substantially sharper than consistency of the global scale itself. In particular, Theorem~\ref{thm:optimal-calibration} establishes only $\tau_\pi^{\mathrm{or}}/\tau_\pi^\star\to1$ and does not by itself imply that the explicit proxy $\tau_\pi^\star$ is accurate enough to preserve a given higher-order separating coefficient.

The logarithmic exponent of the calibration is particularly consequential. If the polynomial exponent is correctly specified but $b=b^\star+\Delta b$, then
\[
\log\frac{c(s/p)^{1/(2\alpha)}L^{b^\star+\Delta b}}{\tau_\pi^{\mathrm{or}}}
=
\Delta b\log L+\log(c/c_\pi^\star)+o(1).
\]
Thus any fixed nonzero $\Delta b$ lies outside every shrinking neighbourhood needed for a higher-order oracle comparison. This does not imply failure of first-order minimaxity: Theorem~\ref{thm:risk-expansion} allows considerably coarser logarithmic calibrations than those required to preserve a finer oracle-risk separation.

\begin{corollary}[Oracle scales for the horseshoe and horseshoe+]
\label{cor:oracle-scales}
For the horseshoe prior, $\alpha=1/2$, $\gamma=0$ and $C_\pi=1/\pi$, so $\tau_{\rm H}^\star=\pi^{3/2}(s/p)L^{1/2}$. For the horseshoe+ prior, $\alpha=1/2$, $\gamma=1$ and $C_\pi=1/\pi^2$, so $\tau_{\rm H+}^\star=\pi^{5/2}(s/p)L^{-1/2}/2$. Their unrestricted risk minimizers satisfy $\tau_{\rm H}^{\mathrm{or}}/\tau_{\rm H}^\star\to1$ and $\tau_{\rm H+}^{\mathrm{or}}/\tau_{\rm H+}^\star\to1$, and both yield the common oracle threshold in \eqref{eq:oracle-threshold}.
\end{corollary}

Corollary~\ref{cor:oracle-scales} illustrates why comparisons made under a common numerical value of $\tau$ can obscure intrinsic differences between priors. The additional logarithmic factor in the horseshoe+ tail is compensated by a factor of order $L^{-1}$ in its oracle global scale relative to the horseshoe. Once this reoptimization has aligned their transition thresholds, the remaining comparison is no longer driven by global-scale misspecification. The next subsection determines the orders at which the residual effects of prior shape separate the oracle risks.

\subsection{Higher-order oracle ranking}
\label{sec:prior-ranking}

The unrestricted calibration in Theorem~\ref{thm:optimal-calibration} removes differences attributable to a mismatched global scale and aligns the leading shrinkage transition across priors. We may therefore compare the residual contribution of prior shape itself. When two priors have the same polynomial and logarithmic tail indices, the next comparison involves the scale-invariant combination
\[
\zeta_\pi=\beta+\frac{\gamma}{\alpha}\log C_\pi.
\]
The quantity $\zeta_\pi$ is unchanged by a fixed rescaling of the local variance and is therefore the relevant next-order descriptor after the oracle scale has absorbed the leading tail normalization.

\begin{theorem}[Higher-order oracle ranking]
\label{thm:oracle-ranking}
Let $\pi_1$ and $\pi_2$ satisfy Assumption~\ref{ass:log-regular}, with tail parameters $(\alpha_j,\gamma_j,\beta_j,C_{\pi_j})$, and let $\zeta_j=\beta_j+(\gamma_j/\alpha_j)\log C_{\pi_j}$ for $j=1,2$. Assume that the priors belong to the fixed log-regular class under consideration and that the parameter differences appearing below do not depend on $p$ or $s$. If $\alpha_1\neq\alpha_2$, then
\[
R_{p,s}^{\mathrm{or}}(\pi_2)-R_{p,s}^{\mathrm{or}}(\pi_1)=-4(\alpha_2-\alpha_1)\frac{s}{L}+o(s/L).
\]
If $\alpha_1=\alpha_2=\alpha$ but $\gamma_1\neq\gamma_2$, then
\[
R_{p,s}^{\mathrm{or}}(\pi_2)-R_{p,s}^{\mathrm{or}}(\pi_1)=4\alpha(\gamma_2-\gamma_1)\frac{s}{L^2}+o(s/L^2).
\]
If $\alpha_1=\alpha_2=\alpha$ and $\gamma_1=\gamma_2=\gamma$ but $\zeta_1\neq\zeta_2$, then
\[
R_{p,s}^{\mathrm{or}}(\pi_2)-R_{p,s}^{\mathrm{or}}(\pi_1)=-4\alpha^2(\zeta_2-\zeta_1)\frac{s}{L^3}+o(s/L^3).
\]
Hence, within the fixed log-regular class, oracle-calibrated priors are resolved successively by $\alpha$, $\gamma$ and $\zeta$: larger $\alpha$ is favoured at order $s/L$; conditional on $\alpha$, smaller $\gamma$ is favoured at order $s/L^2$; and conditional on $(\alpha,\gamma)$, larger $\zeta$ is favoured at order $s/L^3$.
\end{theorem}

Theorem~\ref{thm:oracle-ranking} gives a lexicographic resolution of priors that first-order minimaxity regards as equivalent:
\[
\alpha\ \succ\ \gamma\ \succ\ \zeta,
\qquad
s/L\ \succ\ s/L^2\ \succ\ s/L^3.
\]
The hierarchy concerns fixed differences within the log-regular class and should not be interpreted as recommending that the polynomial tail index be increased without bound. Its implication is instead comparative: once each prior has been calibrated according to Theorem~\ref{thm:optimal-calibration}, a difference in $\alpha$ dominates any fixed difference in $\gamma$ or $\zeta$, while the finer descriptors become visible only after the coarser tail characteristics have been matched. The normalization $C_\pi$ does not constitute a separate stage of this hierarchy because its leading effect is absorbed by the oracle constant $c_\pi^\star$; when $\alpha$ and $\gamma$ agree, it enters only through the invariant combination $\zeta_\pi$.

The first stage of the hierarchy can also be read from the matched-threshold expansion. At the common oracle threshold $z_0=2L+\log L-\log\pi$,
\[
\kappa_\pi(z_0)=\frac{2\alpha+1}{z_0}+O(L^{-2}),
\]
so a fixed change in $\alpha$ perturbs the transition-shrinkage coefficient already at order $L^{-1}$. The matched-threshold risk comparison therefore separates such priors at order $s/L$, before the logarithmic and next-order tail descriptors become visible. This preference for larger $\alpha$ is specific to oracle-calibrated sparse worst-case squared-error risk. The restriction $\alpha<2$ delimits the class over which the present expansion is established and should not be interpreted as identifying $\alpha=2$ as an optimal tail boundary.

\begin{corollary}[Horseshoe versus horseshoe+]
\label{cor:horseshoe-comparison}
Let $\pi_{\rm H}$ and $\pi_{\rm H+}$ denote the horseshoe and horseshoe+ local-variance densities, respectively. Under their unrestricted oracle calibrations in Corollary~\ref{cor:oracle-scales},
\[
R_{p,s}^{\mathrm{or}}(\pi_{\rm H+})-R_{p,s}^{\mathrm{or}}(\pi_{\rm H})=\frac{2s}{L^2}+o(s/L^2).
\]
Consequently, the horseshoe and horseshoe+ have the same first-order minimax risk and the same asymptotic oracle transition threshold, but the horseshoe has the smaller worst-case squared-error risk at the first order at which their oracle risks differ.
\end{corollary}

Corollary~\ref{cor:horseshoe-comparison} follows from the second part of Theorem~\ref{thm:oracle-ranking}, since both priors have $\alpha=1/2$ while $\gamma_{\rm H}=0$ and $\gamma_{\rm H+}=1$. The additional logarithmic tail factor of the horseshoe+ is therefore invisible at the $s/L$ scale but produces a positive risk difference of size $2s/L^2$ after each prior has been separately calibrated. This comparison is one instance of the more general principle that first-order equivalence need not persist at the first prior-sensitive order.

\begin{remark}[Large-signal robustness and sparse minimax risk]
\label{rem:large-signal}
The ordering in Corollary~\ref{cor:horseshoe-comparison} does not conflict with results showing advantages of the horseshoe+ for isolated large signals or under other inferential criteria \cite{Bhadra2017}. Such comparisons emphasize regions in which additional tail mass reduces shrinkage of sufficiently large observations. Worst-case risk over $\Theta_{p,s}$ also accumulates the contribution of the $p-s$ null coordinates and of signals near the shrinkage transition, where faster unshrinking can instead increase risk. The two criteria therefore need not induce the same ordering, and Theorem~\ref{thm:oracle-ranking} should be interpreted as a criterion-specific ranking for oracle-calibrated posterior means under sparse worst-case squared-error loss.
\end{remark}

The normal-means results now separate three distinct questions. Theorem~\ref{thm:risk-expansion} determines which prior--scale specifications survive first- and second-order minimax screening, Theorem~\ref{thm:optimal-calibration} identifies and aligns the unrestricted oracle scales, Proposition~\ref{prop:oracle-sensitivity} quantifies the accuracy required to retain a given oracle-risk resolution, and Theorem~\ref{thm:oracle-ranking} identifies the tail characteristics governing the remaining separation. We next ask when this hierarchy remains visible after the orthogonal sequence model is replaced by a high-dimensional linear regression design.

\subsection{Stability under high-dimensional linear regression}
\label{sec:regression-stability}

We now examine when the higher-order risk hierarchy derived in the Gaussian sequence model remains visible under high-dimensional regression. Consider
\[
\boldsymbol Y=\boldsymbol X\boldsymbol\beta^0+\boldsymbol\varepsilon,\qquad
\boldsymbol\varepsilon\sim N_n(\boldsymbol0,\boldsymbol I_n),
\]
where $\boldsymbol X\in\mathbb R^{n\times p}$ has columns normalized so that $\|\boldsymbol X_j\|_2^2=n$. Put $\boldsymbol\theta=\sqrt n\,\boldsymbol\beta$ and assign $\beta_j\mid u_j,\tau\sim N(0,\tau^2u_j/n)$, $u_j\sim\pi$, so that the induced prior on $\theta_j$ agrees with \eqref{eq:gl-prior}.

For fixed $C_\theta>0$, define
\[
\Theta_{p,s}^{\mathrm{env}}(C_\theta)
=
\left\{
\boldsymbol\theta\in\mathbb R^p:
\|\boldsymbol\theta\|_0\leq s,\ 
\max_{j\in\operatorname{supp}(\boldsymbol\theta)}
\log(1+|\theta_j|)
\leq C_\theta L
\right\}.
\]
Let $\widehat{\boldsymbol\beta}_{\pi,\tau}$ denote the posterior mean and define the normalized worst-case regression risk
\[
R_{n,p,s}^{X}(\pi,\tau)
=
\sup_{\sqrt n\,\boldsymbol\beta^0\in\Theta_{p,s}^{\mathrm{env}}(C_\theta)}
E_{\boldsymbol\beta^0}
\left\{
n\|\widehat{\boldsymbol\beta}_{\pi,\tau}-\boldsymbol\beta^0\|_2^2
\right\}.
\]
The dependence on $C_\theta$ is suppressed. For the sequence-model oracle scale $\tau_\pi^{\mathrm{or}}$ in Theorem~\ref{thm:optimal-calibration}, write $R_{n,p,s}^{X,\mathrm{seq}}(\pi)=R_{n,p,s}^{X}(\pi,\tau_\pi^{\mathrm{or}})$.

Let $\boldsymbol G=\boldsymbol X^T\boldsymbol X/n$. We measure departures from orthogonality through the restricted-isometry quantity
\[
\Delta_{n,m}
=
\sup_{\substack{\boldsymbol v\neq\boldsymbol0\\ \|\boldsymbol v\|_0\leq m}}
\left|
\frac{\boldsymbol v^T\boldsymbol G\boldsymbol v}{\|\boldsymbol v\|_2^2}
-1
\right|
\]
and the cumulative active-set coherence
\[
\Gamma_{n,s}
=
\sup_{|S|\leq s}
\|(\boldsymbol G-\boldsymbol I_p)_{SS}\|_{\infty\to\infty}.
\]
For a sequence $m_n$, write $\delta_n=\Delta_{n,2m_n}$ and $\gamma_n=\Gamma_{n,s}$. These quantities control sparse quadratic geometry and cumulative dependence within active sets, respectively.

\begin{theorem}[Stability under sparse designs]
\label{thm:regression-stability}
Suppose Assumption~\ref{ass:log-regular} holds with $0<\alpha<1$, fix $C_\theta>0$, and let $\tau_\pi^{\mathrm{or}}$ be the sequence-model oracle scale. Suppose that, for some $m_n$,
\[
sL=o(m_n),
\qquad
\delta_nL\to0,
\qquad
\gamma_nL\to0.
\]
Then
\[
R_{n,p,s}^{X,\mathrm{seq}}(\pi)
=
R_{p,s}^{\mathrm{or}}(\pi)
+
O\left[
sL^2\{\gamma_n+\delta_n^2\}
\right].
\]
\end{theorem}

The signal envelope ensures a uniform marginal-likelihood bound in the posterior-localization argument without excluding the sequence-model least-favourable configurations, whose active signals are of order $L^{1/2}$, compared with the envelope $\exp(C_\theta L)$. The two design quantities enter differently: $\gamma_n$ controls active--active dependence at linear order, while cross--null perturbations have vanishing first variation at the orthogonal reference experiment and contribute at quadratic order through $\delta_n^2$. The proof localizes the effective activity of nominally null coordinates rather than requiring the posterior mean error to be sparse. Consequently, full-matrix near-orthogonality is unnecessary, and the result remains compatible with $p>n$.

The restriction $0<\alpha<1$ ensures that the oracle null activity satisfies $p q_\pi(\tau_\pi^{\mathrm{or}})=O(sL^\alpha)=o(sL)$, where $q_\pi(\tau)=E_0\{w(Y,\tau)\}$. Thus $sL=o(m_n)$ provides sufficient sparse-design coverage for the active coordinates and the effective activity of nominally null coordinates.

For a normalized Gaussian design with independent standard normal entries, standard sparse random-matrix concentration \cite{Vershynin2018} and a Gaussian maximal bound give, upon choosing $m_n=sLh_n$ with $h_n\to\infty$ arbitrarily slowly and $\log(Lh_n)=o(L)$,
\[
\gamma_n+\delta_n^2
=
O_P\left\{
s\sqrt{\frac{L}{n}}
+
\frac{s h_nL^2}{n}
\right\},
\]
provided $m_n\log(ep/m_n)=o(n)$. Hence a separation of order $s/L^k$ is preserved, for example, if $s^2L^{2k+5}=o(n)$ and $s h_nL^{k+4}=o(n)$. These conditions are sufficient rather than sharp.

Theorem~\ref{thm:regression-stability} therefore preserves a sequence-model separation of order $s/L^k$ whenever $\gamma_n+\delta_n^2=o(L^{-k-2})$. Applied to Theorem~\ref{thm:oracle-ranking}, this requires design resolution $o(L^{-3})$, $o(L^{-4})$ and $o(L^{-5})$ for the $\alpha$-, $\gamma$- and $\zeta$-level comparisons, respectively. In particular, under $\gamma_n+\delta_n^2=o(L^{-4})$, Corollary~\ref{cor:horseshoe-comparison} transfers to regression:
\[
R_{n,p,s}^{X,\mathrm{seq}}(\pi_{\rm H+})
-
R_{n,p,s}^{X,\mathrm{seq}}(\pi_{\rm H})
=
\frac{2s}{L^2}
+
o(s/L^2).
\]
Thus the sequence-model ordering remains informative beyond exact orthogonality, with increasingly fine prior distinctions requiring correspondingly stronger control of design perturbations.

\section{Consequences for global--local prior choice}

\subsection{Oracle calibration of representative priors}
\label{sec:canonical-priors}

Table~\ref{tab:canonical-priors} applies Theorem~\ref{thm:optimal-calibration} to representative global--local priors, retaining logarithmic tail modifications beyond the usual regular-variation characterization \cite{Bhadra2016}. The horseshoe is the half-Cauchy member of the half-\(t\) family \cite{Gelman2006,Carvalho2010}, the horseshoe+ introduces an additional logarithmic tail factor \cite{Bhadra2017}, and beta-prime local variances provide a family with polynomial tails \cite{Armagan2011,Bai2021}. Tail constants affect \(c_\pi^\star\) in \eqref{eq:oracle-calibration}, but not the exponents \(r^\star\) and \(b^\star\).

\begin{table}[!ht]
\begin{tabular}{lcccc}
Prior & Tail of \(\pi(u)\) & \(\alpha\) & \(\gamma\) & \((r^\star,b^\star)\)\\
\hline
Horseshoe & \(u^{-3/2}\) & \(1/2\) & \(0\) & \((1,1/2)\)\\
Horseshoe+ & \(u^{-3/2}\log u\) & \(1/2\) & \(1\) & \((1,-1/2)\)\\
Half-\(t_\nu\) & \(u^{-\nu/2-1}\) & \(\nu/2\) & \(0\) & \((1/\nu,1/2)\)\\
Beta-prime\((a_0,b_0)\) & \(u^{-b_0-1}\) & \(b_0\) & \(0\) & \((1/(2b_0),1/2)\)\\
\hline
\end{tabular}
\caption{Tail indices and oracle scale exponents for representative global--local priors. Tail expressions omit multiplicative constants. The exponents refer to \(\tau_\pi^\star=c_\pi^\star(s/p)^{r^\star}L^{b^\star}\). The condition \(0<\alpha<2\) in Theorem~\ref{thm:optimal-calibration} requires \(0<\nu<4\) for half-\(t_\nu\) and \(0<b_0<2\) for beta-prime\((a_0,b_0)\).}
\label{tab:canonical-priors}
\end{table}

For polynomial tails, \(\gamma=0\) gives \(b^\star=1/2\), whereas \(r^\star=1/(2\alpha)\) varies with the tail index. The horseshoe and horseshoe+ share \(\alpha=1/2\) but differ in \(\gamma\), so their oracle scales differ by a factor of order \(L^{-1}\). This compensates for the logarithmic tail modification and aligns their leading transitions at \eqref{eq:oracle-threshold}; comparison at a common numerical value of \(\tau\) would instead confound prior shape with calibration. Table~\ref{tab:canonical-priors} thus illustrates the distinct roles of polynomial and logarithmic tail behaviour in Theorem~\ref{thm:oracle-ranking}.

\subsection{A resolution hierarchy for prior choice}
\label{sec:prior-choice}

Theorem~\ref{thm:oracle-ranking} orders oracle-calibrated priors through \(\alpha\succ\gamma\succ\zeta_\pi\), with first separating orders \(s/L\), \(s/L^2\) and \(s/L^3\). Proposition~\ref{prop:oracle-sensitivity} and Theorem~\ref{thm:regression-stability} further determine the calibration and design accuracies needed to preserve each separation. Table~\ref{tab:resolution-hierarchy} summarizes these joint resolution requirements.

\begin{table}[!ht]
\begin{tabular}{lccc}
Prior feature & Separating order & Calibration accuracy & Design accuracy\\
\hline
\(\alpha\) & \(s/L\) & \(o(L^{-1/2})\) & \(o(L^{-3})\)\\
\(\gamma\) & \(s/L^2\) & \(o(L^{-1})\) & \(o(L^{-4})\)\\
\(\zeta\) & \(s/L^3\) & \(o(L^{-3/2})\) & \(o(L^{-5})\)\\
\hline
\end{tabular}
\caption{Resolution requirements for higher-order prior comparison. Calibration accuracy refers to \(|\log(\widetilde\tau/\tau_\pi^{\rm or})|\). Design accuracy refers to \(\gamma_n+\delta_n^2\) under the signal-envelope and sparse-coverage conditions of Theorem~\ref{thm:regression-stability}; these conditions do not restrict the normal-means comparisons.}
\label{tab:resolution-hierarchy}
\end{table}

Within the fixed log-regular class, larger \(\alpha\), then smaller \(\gamma\), then larger \(\zeta_\pi\) are favoured, although the preference for larger \(\alpha\) should not be extrapolated towards arbitrarily light tails. In particular, the horseshoe--horseshoe+ comparison occurs at order \(s/L^2\), requiring calibration error \(o(L^{-1})\) and, for regression transfer, design distortion \(o(L^{-4})\). Thus first-order minimaxity provides an initial screening criterion, while higher-order prior choice depends on whether the relevant risk separation can be resolved under the available calibration and design accuracy.

\section{Numerical study}

\subsection{Finite-sparsity phase diagram}
\label{sec:numerical-phase}

We first examine the calibration boundaries in Theorem~\ref{thm:risk-expansion} at finite sparsity levels. Posterior means and risks are evaluated by deterministic numerical quadrature, with the active-coordinate risk maximized numerically over the signal magnitude. Computational details are given in the Supplementary Material. For the horseshoe prior, $\alpha=1/2$ and $\gamma=0$, we fix the critical polynomial exponent $r=1$ and set $\tau=(s/p)L^d$ for $d=0.5,1,1.25,1.5,1.75$, varying $L=\log(p/s)$ from 3 to 200. Theorem~\ref{thm:risk-expansion} predicts first-order minimaxity for $d<3/2$ and preservation of the classical second-order expansion for $d\leq1$.

\begin{figure}[!ht]
\centering
\includegraphics[width=\textwidth]{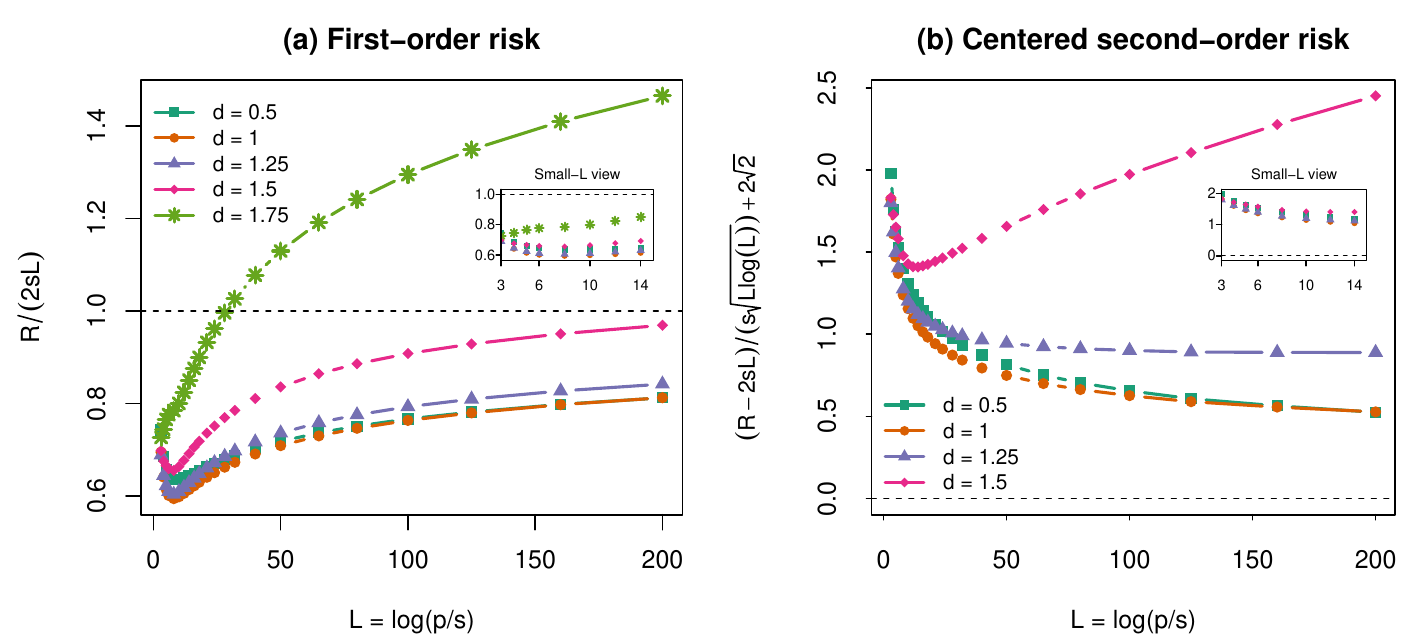}
\caption{Finite-sparsity phase diagram for horseshoe calibrations $\tau=(s/p)L^d$. Panel (a) shows $R/(2sL)$ for $d=0.5,1,1.25,1.5,1.75$, with the dashed line at one. Panel (b) shows $\{R-2sL\}/[s\{L\log L\}^{1/2}]+2\sqrt{2}$ for $d\leq1.5$, with the dashed line at zero. Insets enlarge the range $3\leq L\leq14$.}
\label{fig:phase-diagram}
\end{figure}

Figure~\ref{fig:phase-diagram}(a) shows that the first-order boundaries emerge gradually. The $d=1.75$ calibration diverges from the others and eventually exceeds the leading minimax benchmark, whereas the boundary case $d=1.5$ remains below it at moderate $L$ because its positive null-risk contribution is offset by the common negative higher-order term. In Fig.~\ref{fig:phase-diagram}(b), the centered risks for $d=0.5$ and $d=1$ approach zero together, while the $d=1.25$ curve eventually departs from them and the $d=1.5$ curve separates more rapidly. Thus failure of second-order optimality becomes apparent only after the leading two minimax terms have been removed.

\subsection{Oracle calibration}
\label{sec:numerical-calibration}

We next assess the unrestricted oracle calibration in Theorem~\ref{thm:optimal-calibration} for the horseshoe, horseshoe+, half-$t_{1.5}$ and beta-prime$(1/2,1)$ priors over $3\leq L\leq50$. Their polynomial tail indices are $1/2$, $1/2$, $3/4$ and $1$, respectively, while the horseshoe and horseshoe+ differ in their logarithmic tail exponents. For each prior, the finite-$L$ oracle scale $\widehat\tau^{\rm or}$ is obtained by minimizing the worst-case risk over a broad log-scale interval, independently of a neighbourhood of $\tau_\pi^\star$. All minimizers are interior; search bounds and tolerances are reported in the Supplementary Material. We also compute the effective threshold $\widehat a_{\rm eff}$ at which the posterior mean equals half the observation. Theorem~\ref{thm:optimal-calibration} predicts $\widehat\tau^{\rm or}/\tau_\pi^\star\to1$ and a common squared threshold $a_\star^2=2L+\log L-\log\pi+o(1)$.

\begin{figure}[!ht]
\centering
\includegraphics[width=\textwidth]{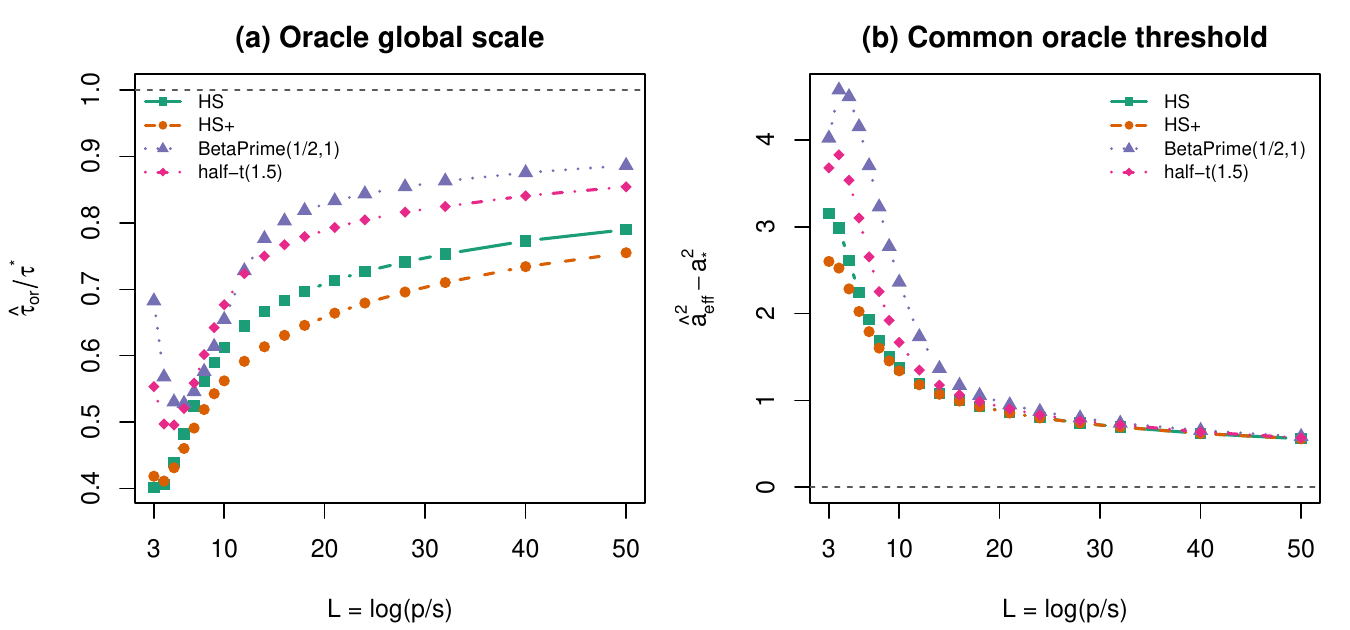}
\caption{Numerical oracle calibration for four global--local priors. Panel (a) shows $\widehat\tau^{\rm or}/\tau_\pi^\star$, with the dashed line at one. Panel (b) shows $\widehat a_{\rm eff}^2-a_\star^2$, where $a_\star^2=2L+\log L-\log\pi$, with the dashed line at zero.}
\label{fig:oracle-calibration}
\end{figure}

Figure~\ref{fig:oracle-calibration}(a) shows convergence towards the prior-specific oracle scales, although appreciable finite-$L$ discrepancies remain even at $L=50$. In Fig.~\ref{fig:oracle-calibration}(b), the initially heterogeneous threshold discrepancies contract as $L$ increases, supporting the common-transition principle. Both panels show that oracle localization is gradual: the asymptotic calibration captures the leading structure without providing a uniformly accurate finite-sparsity approximation.

\subsection{Higher-order prior ranking}
\label{sec:numerical-ranking}

We examine the oracle-risk differences over $3\leq L\leq50$, optimizing each global scale as in Section~\ref{sec:numerical-calibration}. Figure~\ref{fig:higher-order-ranking} compares the horseshoe with horseshoe+ at the logarithmic-tail level and with half-$t_{1.5}$ and beta-prime$(1/2,1)$ priors at the polynomial-tail level.

\begin{figure}[!ht]
\centering
\includegraphics[width=\textwidth]{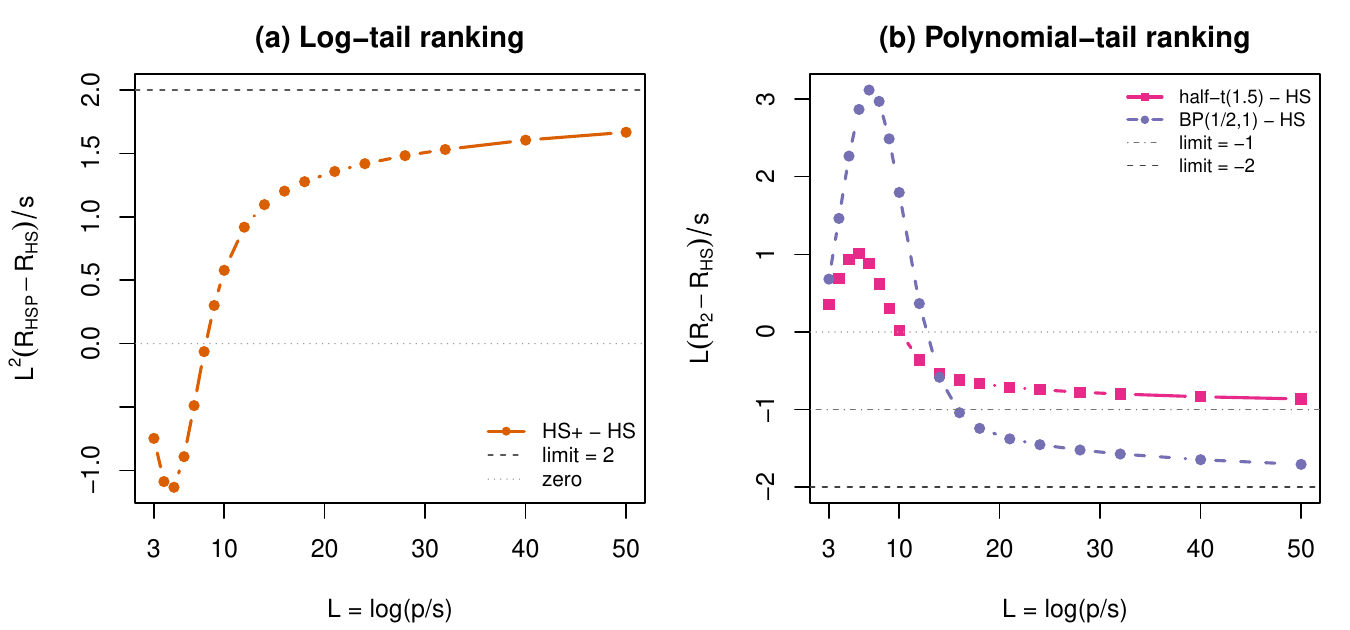}
\caption{Higher-order risk differences under numerically optimized global scales. Panel (a) shows $L^2\{R^{\rm or}_{\rm H+}-R^{\rm or}_{\rm H}\}/s$; the horizontal lines mark zero and the limiting value 2. Panel (b) shows $L\{R^{\rm or}_2-R^{\rm or}_{\rm H}\}/s$ for half-$t_{1.5}$ and beta-prime$(1/2,1)$ priors, with limiting values $-1$ and $-2$, respectively.}
\label{fig:higher-order-ranking}
\end{figure}

Figure~\ref{fig:higher-order-ranking}(a) reveals a finite-sparsity reversal of the asymptotic horseshoe--horseshoe+ ordering. The normalized oracle-risk difference is initially negative, favouring horseshoe+, before changing sign near $L=9$ and gradually approaching the limiting value 2. Thus the horseshoe advantage in Corollary~\ref{cor:horseshoe-comparison} emerges only when the $s/L^2$ separation dominates the remaining finite-$L$ terms.

In Fig.~\ref{fig:higher-order-ranking}(b), the polynomial-tail comparisons also exhibit positive differences at small $L$ before crossing zero and approaching their respective negative limits. Even the coarser $s/L$ ordering therefore need not hold at moderate sparsity. Together, the panels support the hierarchy in Theorem~\ref{thm:oracle-ranking}, while showing that its limiting risk ordering should not be interpreted as a universal finite-sparsity recommendation.

\subsection{Finite-sample stability under regression designs}
\label{sec:regression-numerics}

We finally examine the finite-sample resolution of the regression transfer in Theorem~\ref{thm:regression-stability}. This experiment is a diagnostic of the higher-order mechanism rather than a verification of the asymptotic design conditions. We fix $s=3$ and $p=445$, giving $L_0=\log(p/s)\approx5$, and consider $n=80,120,180,260,360$, with $p>n$ throughout. Each design has independent standard normal entries and columns normalized to squared norm $n$. The design, active set, signal signs and noise are shared across priors to obtain paired comparisons.

Each prior uses its numerically optimized finite-$L_0$ oracle scale. Nonzero coefficients are set to $0.9$, $1$ or $1.1$ times the corresponding sequence-model least-favourable magnitude, and the largest mean loss across these configurations defines the least-favourable-profile risk $R_\pi^X$. This diagnostic is distinct from the exact supremum in $R_{n,p,s}^X$. Posterior means are computed by Markov chain Monte Carlo over 200 replications per $n$. We compare half-$t_{1.5}$ with the horseshoe at order $s/L$, and horseshoe+ with the horseshoe at order $s/L^2$. Figure~\ref{fig:regression-stability} reports the normalized differences with $95\%$ paired bootstrap intervals based on 1000 resamples; computational details are given in the Supplementary Material.

\begin{figure}[!ht]
\centering
\includegraphics[width=\textwidth]{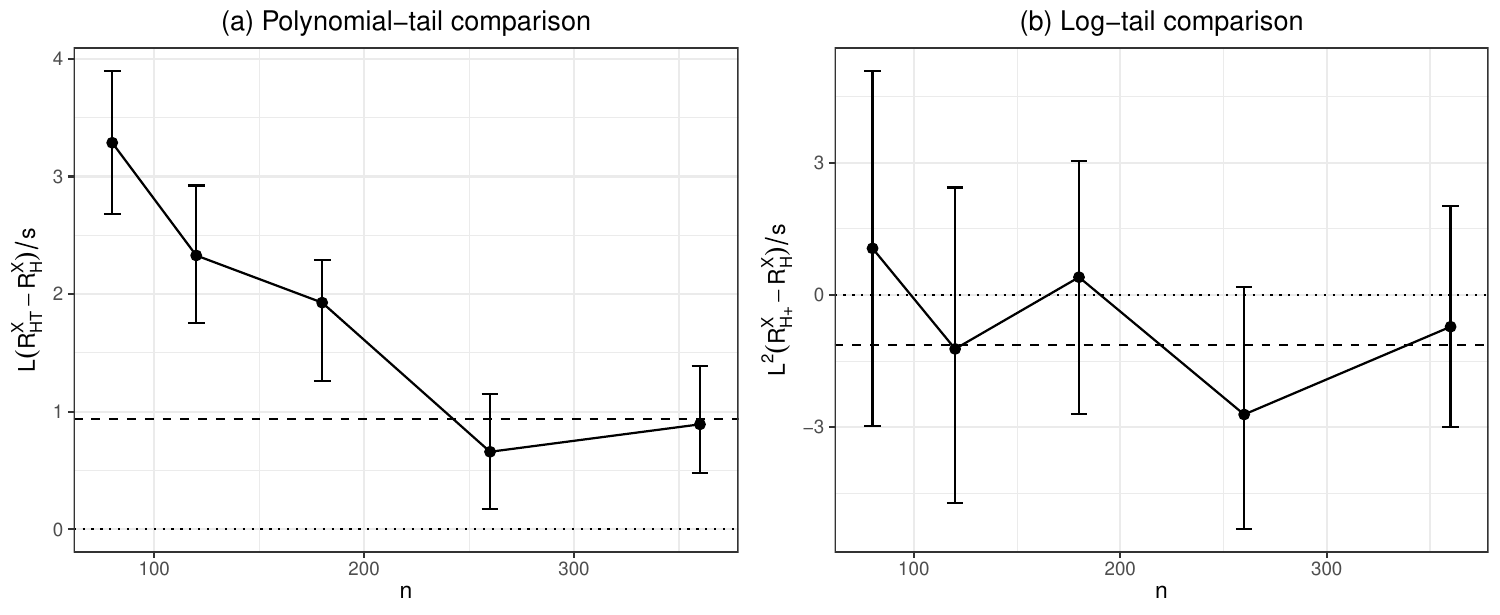}
\caption{Finite-sample stability of higher-order regression risk differences for $s=3$, $p=445$ and $L_0=5$. Panel (a) shows $L_0(R_{\mathrm{HT}}^X-R_{\mathrm H}^X)/s$, and panel (b) shows $L_0^2(R_{\mathrm{H+}}^X-R_{\mathrm H}^X)/s$. Dashed lines are finite-$L_0$ sequence-model benchmarks, dotted lines mark zero, and error bars are $95\%$ paired bootstrap intervals from 200 replications.}
\label{fig:regression-stability}
\end{figure}

Figure~\ref{fig:regression-stability}(a) shows that the polynomial-tail risk difference moves towards its finite-$L_0$ sequence-model benchmark as $n$ increases, with the largest-sample estimate lying close to that benchmark despite some nonmonotonicity. In contrast, Fig.~\ref{fig:regression-stability}(b) shows substantial uncertainty in the logarithmic-tail comparison: the intervals include both zero and the sequence-model benchmark throughout. The finer $s/L^2$ separation is therefore not reliably resolved even when the coarser $s/L$ comparison approaches its benchmark. This contrast is consistent with the stronger design-resolution requirement for finer comparisons in Theorem~\ref{thm:regression-stability}, but does not establish that its asymptotic sufficient conditions hold at these sample sizes.

\section{Discussion}

The regression extension should be interpreted as a higher-order transfer result rather than as a new sparse-design condition. Theorem~\ref{thm:regression-stability} quantifies how small the design perturbation must be for a prior-specific sequence-model risk difference to remain visible in high-dimensional regression. This also delineates the present scope of the result: stronger predictor dependence may alter the effective detection boundary and accumulation of estimation error, while non-Gaussian likelihoods introduce additional curvature terms that can occur at the same or larger orders than the prior-sensitive effects studied here. Extending the theory beyond sparse near-orthogonality will therefore require separating higher-order contributions generated by prior shape from those induced by design dependence and likelihood geometry.

A second issue is adaptation of the global scale. Empirical or hierarchical treatments can recover adaptation to unknown sparsity at the first-order level for global--local priors \citep{VanDerPas2017,Paul2025a}, but the accuracy required for higher-order comparison is substantially finer. If two oracle-calibrated priors first separate at order \(s/L^k\), Proposition~\ref{prop:oracle-sensitivity} shows that a perturbation \(\widehat{\tau}=\tau_{\pi}^{\mathrm{or}}\exp(t)\) is negligible at that resolution only when \(t=o(L^{-k/2})\). In particular, preserving the horseshoe--horseshoe+ comparison at order \(s/L^2\) requires \(\lvert\log(\widehat{\tau}/\tau_{\pi}^{\mathrm{or}})\rvert=o(L^{-1})\). Thus first-order adaptation does not imply the calibration precision needed for higher-order prior comparison. It remains to determine whether empirical-Bayes or fully Bayesian global-scale procedures attain these prior-specific accuracy levels uniformly over sparse regimes, and how their calibration error interacts with design perturbation in regression.

More broadly, the ordering derived here is specific to sparse worst-case squared-error risk for posterior-mean estimation. Other criteria can favour different tail behaviour. Additional tail mass or logarithmic tail modification can improve large-signal robustness, contraction or multiple-testing performance \citep{Song2020,Bhadra2017,Hamura2026,Paul2025b}; such results are not in conflict with the present oracle-calibrated minimax ordering because they evaluate different aspects of the procedure. The finite-\(L\) calculations provide a second qualification: the horseshoe+ can have smaller oracle risk than the horseshoe before the asymptotic \(s/L^2\) separation becomes dominant, so the limiting hierarchy should not be interpreted as a universal recommendation at moderate sparsity. Nor need a worst-case minimax ranking agree with an average-risk comparison under a stochastic model for the nonzero signals, where the distribution of signal magnitudes can place different value on tail robustness. The broader conclusion is therefore not that one global--local prior is intrinsically best, but that prior characteristics become distinguishable only at the statistical resolution permitted jointly by the loss, sparsity regime, global-scale calibration and design geometry.

\appendix

\section{Proofs and technical results}

\subsection{Posterior representation and risk decomposition}
\label{app:posterior-geometry}

We begin with exact one-dimensional identities for the posterior mean under (2). For \(j\geq0\), define
\[
M_j(y,\tau)
=
\int_0^\infty
(1+\tau^2u)^{-j-1/2}
\exp\left\{-\frac{y^2}{2(1+\tau^2u)}\right\}
\pi(u)\,du,
\]
and write \(\kappa=\kappa(u,\tau)=(1+\tau^2u)^{-1}\). The marginal density of \(Y\) is \((2\pi)^{-1/2}M_0(y,\tau)\), while \(E(\kappa^j\mid y)=M_j(y,\tau)/M_0(y,\tau)\). Since \(E(\theta\mid y,u)=y\{1-\kappa(u,\tau)\}\), the posterior mean is
\begin{equation}
\label{eq:app-postmean}
\widehat\theta_{\pi,\tau}(y)
=
y\,w(y,\tau),
\qquad
w(y,\tau)
=
1-\frac{M_1(y,\tau)}{M_0(y,\tau)}.
\end{equation}
Thus \(w(y,\tau)\) is the posterior nonshrinkage weight.

\begin{lemma}[Posterior identities and monotonicity]
\label{lem:posterior-identities}
For every \(\tau>0\) and \(j\geq0\),
\[
\frac{\partial}{\partial y}M_j(y,\tau)
=
-yM_{j+1}(y,\tau),
\qquad
\frac{\partial}{\partial(y^2)}M_j(y,\tau)
=
-\frac12M_{j+1}(y,\tau).
\]
Consequently,
\[
\frac{\partial}{\partial(y^2)}
\frac{M_1(y,\tau)}{M_0(y,\tau)}
=
-\frac12
\left\{
\frac{M_2(y,\tau)}{M_0(y,\tau)}
-
\left(
\frac{M_1(y,\tau)}{M_0(y,\tau)}
\right)^2
\right\}
=
-\frac12\operatorname{Var}(\kappa\mid y)
<0.
\]
Hence \(w(y,\tau)\) is strictly increasing as a function of \(y^2\). In particular, \(w\) is even and \(\widehat\theta_{\pi,\tau}\) is odd.
\end{lemma}

\begin{proof}
Differentiation under the integral sign is justified because multiplication by \((1+\tau^2u)^{-1}\) preserves an integrable envelope proportional to \(\pi(u)\). Differentiating the defining integral for \(M_j\) gives the first two identities, and the quotient rule gives the displayed posterior-variance representation. Since \(\pi\) is positive on \((0,\infty)\), the posterior distribution of \(\kappa\) is nondegenerate and the variance is strictly positive. The remaining assertions follow from \eqref{eq:app-postmean}.
\end{proof}

By Lemma~\ref{lem:posterior-identities}, any positive solution of
\begin{equation}
\label{eq:app-transition-threshold}
w(a_\tau,\tau)=\frac12
\end{equation}
is unique. We call \(a_\tau\) the intrinsic transition threshold whenever this solution exists. Section~\ref{app:uniform-expansions} establishes its existence in the asymptotic regimes used below and determines its location.

For later risk calculations, define
\[
R_\theta(\pi,\tau)
=
E_\theta\!\left[
\{\widehat\theta_{\pi,\tau}(Y)-\theta\}^2
\right],
\qquad
R_0(\pi,\tau)=R_{\theta=0}(\pi,\tau),
\qquad
R_\star(\pi,\tau)
=
\sup_{\theta\in\mathbb R}R_\theta(\pi,\tau).
\]
The oddness of \(\widehat\theta_{\pi,\tau}\) implies \(R_\theta(\pi,\tau)=R_{-\theta}(\pi,\tau)\). Since the estimator acts coordinatewise and \(\Theta_{p,s}\) contains at most \(s\) nonzero coordinates,
\begin{equation}
\label{eq:app-risk-decomposition}
R_{p,s}(\pi,\tau)
=
(p-s)R_0(\pi,\tau)
+
sR_\star(\pi,\tau).
\end{equation}
Indeed, coordinatewise additivity gives the upper bound, while the supremum is approached by assigning \(s\) coordinates signal values whose one-dimensional risks approach \(R_\star(\pi,\tau)\) and setting the remaining \(p-s\) coordinates to zero.

\subsection{Uniform central--tail and transition expansions}
\label{app:uniform-expansions}

We next develop the posterior-integral expansions used throughout the risk analysis. Let \(B(y)=e^{-y^2/2}, H(y,\tau)=\log\frac{y^2}{2\tau^2}, U_\tau=\tau^{-1}\),
and, for \(j=0,1\), decompose \(M_j(y,\tau)=C_j(y,\tau)+T_j(y,\tau)\),
where
\[
C_j(y,\tau)=\int_0^{U_\tau}(1+\tau^2u)^{-j-1/2}
\exp\left\{-\frac{y^2}{2(1+\tau^2u)}\right\}\pi(u)\,du,
\]
\[
T_j(y,\tau)=\int_{U_\tau}^\infty(1+\tau^2u)^{-j-1/2}
\exp\left\{-\frac{y^2}{2(1+\tau^2u)}\right\}\pi(u)\,du.
\]
The origin condition in Assumption~1 is not needed in this subsection.

We first record a coarse central--tail balance that does not require \(y^2\) and \(H(y,\tau)\) to have the same order. Suppose \(\tau\to0\), \(y^2\to\infty\), \(H(y,\tau)\to\infty\), and \(y^2\tau\to0\). On \(0<u\leq U_\tau\),
\[
(1+\tau^2u)^{-j-1/2}
\exp\left\{\frac{y^2\tau^2u}{2(1+\tau^2u)}\right\}
=
1+O\{(1+y^2)\tau\},
\]
uniformly, while \(\int_{U_\tau}^\infty\pi(u)\,du=o(1)\). Hence \(C_j(y,\tau)=B(y)\{1+o(1)\}\).

For the tail component use \(t=\frac{y^2}{2(1+\tau^2u)}, u=\frac{y^2-2t}{2\tau^2t}\).
Since \(u>U_\tau\) corresponds to \(0<t<y^2/\{2(1+\tau)\}\),
\[
T_j(y,\tau)
=
\frac{2^{j-1/2}y^{1-2j}}{\tau^2}
\int_0^{y^2/\{2(1+\tau)\}}
t^{j-3/2}e^{-t}
\pi\left(\frac{y^2-2t}{2\tau^2t}\right)dt.
\]
Put \(q_j=\alpha+j+1/2, K_j=C_\pi2^{\alpha+j+1/2}\Gamma(q_j)\).
Retaining the leading terms of the tail expansion gives
\[
T_j(y,\tau)
=
K_j\tau^{2\alpha}y^{-2\alpha-2j-1}H(y,\tau)^\gamma
\{1+O(H^{-1}+y^{-2})\},
\qquad j=0,1.
\]

Suppose now that an intrinsic threshold \(a_\tau\) exists and satisfies \(a_\tau^2\to\infty\), \(H(a_\tau,\tau)\to\infty\), and \(a_\tau^2\tau\to0\). Since \(T_1/T_0=O(a_\tau^{-2})\), the defining relation \(w(a_\tau,\tau)=1/2\) and the preceding central approximation imply \(T_0(a_\tau,\tau)/B(a_\tau)=1+o(1)\). Writing \(\lambda=\log(1/\tau)\) and taking logarithms therefore gives
\[
\frac{a_\tau^2}{2}
-2\alpha\lambda
-\left(\alpha+\frac12\right)\log a_\tau^2
+\gamma\log H(a_\tau,\tau)
+\log K_0
+o(1)
=
0.
\]
Since \(H(a_\tau,\tau)=2\lambda+\log(a_\tau^2/2)\),
\begin{equation}
\label{eq:app-coarse-threshold-balance}
a_\tau^2
=
4\alpha\lambda
+
O\{\log a_\tau+\log(1+\lambda)\}.
\end{equation}

We now turn to the sharper sparse-transition regime. Unless stated otherwise, throughout the remainder of this subsection \(\tau\to0\), \(y^2\asymp L\), and \(H(y,\tau)\asymp L\), with all bounds uniform over the fixed log-regular class in Assumption~1. These conditions imply that \(\tau\) is exponentially small in \(L\), and hence, for every fixed \(m>0\), \(C_j(y,\tau)=B(y)\{1+o(L^{-m})\}, j=0,1\).

\begin{proposition}[Uniform tail expansion]
\label{prop:uniform-tail}
Suppose Assumption~1 holds. Uniformly whenever \(y^2\asymp L\) and \(H=H(y,\tau)\asymp L\),
\begin{equation}
\label{eq:app-tail-expansion}
T_j(y,\tau)
=
K_j\tau^{2\alpha}y^{-2\alpha-2j-1}H^\gamma
\left\{
1+\frac{a_j}{H}
+\frac{b_j}{y^2}
+\frac{c_j}{H^2}
+\frac{d_j}{Hy^2}
+\frac{e_j}{y^4}
+\mathcal E_j(y,\tau)
\right\},
\end{equation}
for \(j=0,1\), where
\[
a_j=\beta-\gamma\psi(q_j),
\qquad
b_j=2(\alpha+1)q_j,
\]
\[
c_j
=
\chi-\beta(\gamma-1)\psi(q_j)
+\frac{\gamma(\gamma-1)}{2}
\{\psi_1(q_j)+\psi(q_j)^2\},
\]
\[
d_j
=
2q_j
\{(\alpha+1)\beta-(\alpha+1)\gamma\psi(q_j+1)-\gamma\},
\qquad
e_j
=
2(\alpha+1)(\alpha+2)q_j(q_j+1).
\]
Here
\[
\mathcal E_j(y,\tau)
=
O\left(
H^{-3}
+
H^{-2}y^{-2}
+
H^{-1}y^{-4}
+
y^{-6}
\right).
\]
Moreover, writing \(z=y^2\), the remainder may be differentiated termwise up to order two and satisfies
\[
\partial_z^k\mathcal E_j(y,\tau)=O(L^{-3-k}),
\qquad
k=1,2.
\]
\end{proposition}

\begin{proof}
Starting from the transformed representation above, fix \(A>0\) sufficiently large and first restrict the integral to \(L^{-A}\leq t\leq A\log L\). On this range,
\[
u^{-\alpha-1}
=
\left(\frac{2\tau^2t}{y^2}\right)^{\alpha+1}
\left(1-\frac{2t}{y^2}\right)^{-\alpha-1},
\qquad
\log u
=
H-\log t+\log\left(1-\frac{2t}{y^2}\right),
\]
and
\[
\left(1-\frac{2t}{y^2}\right)^{-\alpha-1}
=
1+\frac{2(\alpha+1)t}{y^2}
+\frac{2(\alpha+1)(\alpha+2)t^2}{y^4}
+O\left(\frac{t^3}{y^6}\right).
\]
The logarithmic factor in Assumption~1 satisfies
\[
\begin{split}
\frac{(\log u)^\gamma}{H^\gamma}
\left\{
1+\frac{\beta}{\log u}
+\frac{\chi}{(\log u)^2}
+\rho(u)
\right\}
={}&
1+\frac{\beta-\gamma\log t}{H}
-\frac{2\gamma t}{Hy^2}\\
&+
\frac{
\chi-\beta(\gamma-1)\log t
+\gamma(\gamma-1)(\log t)^2/2
}{H^2}\\
&+
O\left(
H^{-3}
+
H^{-2}\frac{t}{y^2}
+
H^{-1}\frac{t^2}{y^4}
\right).
\end{split}
\]
Multiplying the two expansions reduces the coefficients in \eqref{eq:app-tail-expansion} to moments of a \(\mathrm{Gamma}(q_j,1)\) variable \(T\): \(E(\log T)=\psi(q_j), E\{(\log T)^2\}=\psi_1(q_j)+\psi(q_j)^2\), \(E(T)=q_j, E(T\log T)=q_j\psi(q_j+1), E(T^2)=q_j(q_j+1)\).
These identities give \(a_j,\ldots,e_j\).

For \(0<t<L^{-A}\), the factor \(t^{q_j-1}\) is uniformly integrable because \(\alpha\) is bounded away from zero, and increasing \(A\) makes this part smaller than any prescribed power of \(L^{-1}\). For \(t>A\log L\), the factor \(e^{-t}\) gives the corresponding bound. On the principal range, \(\frac{\partial}{\partial z}\log u = \frac{1}{z-2t} = O(L^{-1}), z=y^2\),
so each differentiation with respect to \(z\) introduces an additional factor of order \(L^{-1}\). The bounds on \(D^k\rho\) in Assumption~1 control these derivatives through the required order, and the two complementary \(t\)-ranges remain negligible after up to two differentiations. This proves \eqref{eq:app-tail-expansion} and the stated derivative bounds.
\end{proof}

The posterior-mean expansion depends on the ratio of the two tail components. Let \(q=\alpha+1/2\) and \(\kappa_T(y,\tau)=T_1(y,\tau)/T_0(y,\tau)\). Proposition~\ref{prop:uniform-tail}, together with \(\psi(q+1)=\psi(q)+1/q\) and \(\psi_1(q+1)=\psi_1(q)-1/q^2\), gives
\begin{equation}
\label{eq:app-kappa-expansion}
\begin{split}
\kappa_T(y,\tau)
=
\frac{2q}{y^2}
\Bigg[
1
&-\frac{\gamma}{qH}
+\frac{2q+1}{y^2}
+\frac{\beta-\gamma\psi(q)}{qH^2}\\
&-\frac{\gamma(4q+1)}{qHy^2}
+\frac{(2q+1)(4q+3)}{y^4}
+O\left(
H^{-3}
+
H^{-2}y^{-2}
+
H^{-1}y^{-4}
+
y^{-6}
\right)
\Bigg].
\end{split}
\end{equation}
In particular,
\begin{equation}
\label{eq:app-kappa-leading}
\kappa_T(y,\tau)
=
\frac{2\alpha+1}{y^2}
-\frac{2\gamma}{y^2H}
+O(L^{-2})
=
O(L^{-1}).
\end{equation}
The \(H^{-2}\) contribution involving \(\chi\) is common to \(T_0\) and \(T_1\) and therefore drops out of the displayed ratio expansion. The same cancellation applies to the leading remainder contribution, while the derivative bounds in Assumption~1 control its variation across the two Gamma weights.

Define the tail-to-central odds \(Q(y,\tau)=T_0(y,\tau)/B(y)\). The preceding central expansion and \eqref{eq:app-postmean} give
\begin{equation}
\label{eq:app-wQ}
w(y,\tau)
=
\frac{Q(y,\tau)}{1+Q(y,\tau)}
\{1-\kappa_T(y,\tau)\}
+
o(L^{-m})
\end{equation}
for every fixed \(m>0\), uniformly whenever \(Q(y,\tau)\) is bounded above and below by fixed polynomial powers of \(L\). Since the cutoff \(U_\tau\) does not depend on \(y\), \(\frac{\partial}{\partial(y^2)}T_0(y,\tau) = -1/2T_1(y,\tau)\),
and hence
\begin{equation}
\label{eq:app-Q-derivative}
\frac{\partial}{\partial(y^2)}
\log Q(y,\tau)
=
\frac12\{1-\kappa_T(y,\tau)\}.
\end{equation}

We next locate the intrinsic threshold under the deterministic calibration (4). Write \(d=2\alpha b+\gamma\) and define
\[
A_{\pi,c,r}
=
C_\pi\Gamma(\alpha+1/2)c^{2\alpha}
(2r)^\gamma
(2\alpha r)^{-\alpha-1/2}.
\]
For \(y^2=CL\), with fixed \(C>0\), Proposition~\ref{prop:uniform-tail} gives \(\log Q(y,\tau) = \left(\frac C2-2\alpha r\right)L+O(\log L)\).
Hence, for every fixed \(\eta>0\), \(w(y,\tau)\to0\) at \(y^2=(4\alpha r-\eta)L\) and \(w(y,\tau)\to1\) at \(y^2=(4\alpha r+\eta)L\). Lemma~\ref{lem:posterior-identities} therefore implies that the intrinsic threshold exists uniquely for all sufficiently large \(L\) and \(\frac{a_\tau^2}{L}\longrightarrow4\alpha r, H(a_\tau,\tau)=2rL\{1+o(1)\}\).

Let \(\kappa_\tau=\kappa_T(a_\tau,\tau)\). Equation~\eqref{eq:app-wQ} and \(w(a_\tau,\tau)=1/2\) give
\begin{equation}
\label{eq:app-threshold-Q}
Q(a_\tau,\tau)
=
\frac{1}{1-2\kappa_\tau}
\{1+o(L^{-m})\}
\end{equation}
for every fixed \(m>0\), and hence
\begin{equation}
\label{eq:app-threshold-logQ}
\log Q(a_\tau,\tau)
=
2\kappa_\tau+O(L^{-2})
=
O(L^{-1}).
\end{equation}
Substituting the leading tail expansion into \eqref{eq:app-threshold-logQ} yields
\begin{equation}
\label{eq:app-general-threshold}
a_\tau^2
=
4\alpha rL
+
(2\alpha+1-2d)\log L
-
2\log A_{\pi,c,r}
+
o(1).
\end{equation}
On the critical polynomial scale \(r=1/(2\alpha)\), \(A_{\pi,c,r} = C_\pi\Gamma(\alpha+1/2)c^{2\alpha}\alpha^{-\gamma} = A_{\pi,c}\),
so \eqref{eq:app-general-threshold} reduces to (6) in Theorem~1.

The same expansions determine the local shape of the posterior transition. For \(x>-a_\tau^2/2\), let \(y_x=(a_\tau^2+2x)^{1/2}\), so that \(y_x^2-a_\tau^2=2x\).

\begin{proposition}[Uniform transition profile]
\label{prop:transition-profile}
Suppose \(a_\tau^2\asymp L\) and \(H(a_\tau,\tau)\asymp L\), and let \(M_L\to\infty\) satisfy \(M_L^2/L\to0\). Uniformly for \(|x|\leq M_L\),
\[
\kappa_T(y_x,\tau)
=
\kappa_\tau
+
O\left(\frac{|x|}{L^2}\right),\qquad 
\log Q(y_x,\tau)
=
x+\kappa_\tau(2-x)
+
O\left(\frac{1+x^2}{L^2}\right).
\]
Consequently,
\begin{equation}
\label{eq:app-local-w}
w(y_x,\tau)
=
\ell(x)
+
\kappa_\tau h(x)
+
O\left(\frac{1+x^2}{L^2}\right),
\end{equation}
where
\[
\ell(x)=\frac{1}{1+e^{-x}},
\qquad
h(x)=\ell(x)\{1-\ell(x)\}(2-x)-\ell(x).
\]
In particular,
\[
\sup_{|x|\leq M_L}|w(y_x,\tau)-\ell(x)|
\longrightarrow0.
\]
\end{proposition}

\begin{proof}
Proposition~\ref{prop:uniform-tail} and \eqref{eq:app-kappa-expansion} give \(\frac{\partial}{\partial(y^2)} \kappa_T(y,\tau) = -\frac{2q}{y^4}+O(L^{-3}) = O(L^{-2})\)
uniformly in the stated region. Since \(y_x^2-a_\tau^2=2x\), \(\kappa_T(y_x,\tau) = \kappa_\tau + O\left(\frac{|x|}{L^2}\right)\).
Integrating \eqref{eq:app-Q-derivative} from \(a_\tau^2\) to \(y_x^2\), and using \eqref{eq:app-threshold-logQ}, gives
\[
\log Q(y_x,\tau)
=
-\log(1-2\kappa_\tau)
+
\int_0^x\{1-\kappa_T(y_z,\tau)\}\,dz
+
o(L^{-2})
=
x+\kappa_\tau(2-x)
+
O\left(\frac{1+x^2}{L^2}\right).
\]
Expanding \(Q/(1+Q)\) around \(e^x/(1+e^x)=\ell(x)\) in \eqref{eq:app-wQ}, together with the preceding expansion for \(\kappa_T(y_x,\tau)\), yields \eqref{eq:app-local-w}. Since \(M_L^2/L\to0\) and \(\kappa_\tau=O(L^{-1})\), both the first-order correction and the remainder vanish uniformly over \(|x|\leq M_L\).
\end{proof}

The following identities will be used in the coordinatewise and pairwise risk calculations. With all integrals over \(\mathbb R\),
\[
\int e^{-x}\ell(x)^2\,dx=1,
\qquad
\int\left[\{1-\ell(x)\}^2-1(x<0)\right]dx=-1,
\]
\[
2\int e^{-x}\ell(x)h(x)\,dx=1,
\qquad
\int\{1-\ell(x)\}h(x)\,dx=\frac12.
\]
They follow directly from the substitution \(v=e^x\).

\subsection{Coordinatewise risk expansions}
\label{app:coordinate-risk}

We now convert the transition geometry in Section~\ref{app:uniform-expansions} into sharp one-coordinate risk expansions. The null-coordinate risk is generated by observations in a neighbourhood of the intrinsic threshold, whereas the worst active-coordinate risk is attained by signals lying slightly below that threshold. These two contributions determine the phase diagram in Theorem~1 and the local risk trade-off used in the oracle analysis.

\begin{proposition}[Sharp null-coordinate risk]
\label{prop:null-risk}
Suppose Assumption~1 holds. If \(\tau\to0\), \(a_\tau^2\asymp L\) and \(H(a_\tau,\tau)\asymp L\), then
\begin{equation}
\label{eq:app-null-leading}
R_0(\pi,\tau)
=
2a_\tau\phi(a_\tau)\{1+o(1)\},
\end{equation}
uniformly over the fixed log-regular class. Equivalently,
\begin{equation}
\label{eq:app-null-tail}
R_0(\pi,\tau)
=
\frac{2^{\alpha+1}}{\sqrt{\pi}}\,
C_\pi\Gamma(\alpha+1/2)
\tau^{2\alpha}
a_\tau^{-2\alpha}
H(a_\tau,\tau)^\gamma
\{1+o(1)\}.
\end{equation}
Under the deterministic calibration (4),
\begin{equation}
\label{eq:app-null-general}
(p-s)R_0(\pi,\tau)
=
2\sqrt{\frac{2\alpha r}{\pi}}\,
A_{\pi,c,r}\,
s(s/p)^{2\alpha r-1}
L^{d-\alpha}
\{1+o(1)\},
\end{equation}
where \(d=2\alpha b+\gamma\) and \(A_{\pi,c,r}\) is defined in Section~\ref{app:uniform-expansions}. In particular, if \(r=1/(2\alpha)\), then
\[
(p-s)R_0(\pi,\tau)
=
\frac{2A_{\pi,c}}{\sqrt{\pi}}\,
sL^{d-\alpha}\{1+o(1)\},
\]
which is (7).
\end{proposition}

\begin{proof}
By symmetry and \eqref{eq:app-postmean},
\[
R_0(\pi,\tau)
=
2\int_0^\infty y^2w(y,\tau)^2\phi(y)\,dy.
\]
Introduce the exact transition coordinate \(x=(y^2-a_\tau^2)/2\) and \(y_x=(a_\tau^2+2x)^{1/2}\). Since \(dy=dx/y_x\) and \(\phi(y_x)=\phi(a_\tau)e^{-x}\),
\begin{equation}
\label{eq:app-null-coordinate}
R_0(\pi,\tau)
=
2a_\tau\phi(a_\tau)
\int_{-a_\tau^2/2}^\infty
\left(1+\frac{2x}{a_\tau^2}\right)^{1/2}
e^{-x}w(y_x,\tau)^2\,dx.
\end{equation}

At the intrinsic threshold, the central expansions in Section~\ref{app:uniform-expansions}, the relation \(w(a_\tau,\tau)=1/2\) and \eqref{eq:app-kappa-leading} give \(\frac{T_0(a_\tau,\tau)}{B(a_\tau)} = 1+o(1)\).
Proposition~\ref{prop:uniform-tail} with \(j=0\) therefore yields
\[
e^{-a_\tau^2/2}
=
C_\pi2^{\alpha+1/2}\Gamma(\alpha+1/2)
\tau^{2\alpha}
a_\tau^{-2\alpha-1}
H(a_\tau,\tau)^\gamma
\{1+o(1)\}.
\]
In particular, \(a_\tau\phi(a_\tau) \asymp \tau^{2\alpha} a_\tau^{-2\alpha} H(a_\tau,\tau)^\gamma\).

Let \(M_L\to\infty\) sufficiently slowly that \(M_L^2/L\to0\). Proposition~\ref{prop:transition-profile} gives \(w(y_x,\tau)=\ell(x)+o(1)\) uniformly for \(|x|\leq M_L\), while the square-root factor in \eqref{eq:app-null-coordinate} is \(1+o(1)\). Using the first logistic identity recorded in Section~\ref{app:uniform-expansions}, the transition region therefore contributes \(2a_\tau\phi(a_\tau) \left\{ \int_{-M_L}^{M_L}e^{-x}\ell(x)^2\,dx+o(1) \right\} = 2a_\tau\phi(a_\tau)\{1+o(1)\}\).

We next control the regions outside the transition window. Write \(M_0-M_1 = (C_0-C_1)+(T_0-T_1)\).
For \(y^2\leq a_\tau^2\), the central component satisfies \(\frac{C_0(y,\tau)-C_1(y,\tau)}{B(y)} \lesssim \tau^2\int_0^{U_\tau}u\pi(u)\,du\).
The upper-tail condition in Assumption~1 gives, for some fixed \(K>0\),
\[
\tau^2\int_0^{U_\tau}u\pi(u)\,du
=
\begin{cases}
O(\tau^{1+\alpha}L^K), & 0<\alpha<1,\\
O(\tau^2L^K), & \alpha\geq1,
\end{cases}
\]
where the logarithmic factor at \(\alpha=1\) is absorbed into \(L^K\). Since \(M_0(y,\tau)\geq B(y)\{1+o(1)\}\) on this range, the corresponding contribution to the null risk is \(O(\tau^{2+2\alpha}L^K) \quad\text{if }0<\alpha<1, O(\tau^4L^K) \quad\text{if }\alpha\geq1\).
Relative to \(a_\tau\phi(a_\tau)\), these bounds contain respectively the exponentially vanishing factors \(\tau^2\) and \(\tau^{4-2\alpha}\). The restriction \(\overline\alpha<2\) therefore makes the central leakage uniformly negligible.

It remains to control the tail component below the transition. Choose a fixed \(0<\eta<\min\{1,1/\overline\alpha\}\).
For \(\eta a_\tau^2\leq y^2\leq a_\tau^2-2M_L\), we have \(y^2\asymp L\) and \(H(y,\tau)\asymp L\). Since \(Q(a_\tau,\tau)=1+o(1)\), \eqref{eq:app-Q-derivative} and \eqref{eq:app-kappa-leading} imply \(Q(y_x,\tau)\lesssim e^x, -\frac{(1-\eta)a_\tau^2}{2}\leq x\leq-M_L\).
Moreover, \(M_0\geq B\{1+o(1)\}\), so \((T_0-T_1)/M_0\leq Q\{1+o(1)\}\). The resulting tail contribution to \eqref{eq:app-null-coordinate} is therefore \(O\{a_\tau\phi(a_\tau)e^{-M_L}\}\).

For the remaining region \(y^2\leq\eta a_\tau^2\), let
\[
D_\tau
=
\int_{U_\tau}^\infty
\frac{\tau^2u}{(1+\tau^2u)^{3/2}}
\pi(u)\,du.
\]
Splitting the integral at \(u=\tau^{-2}\) and using the upper-tail bounds gives
\[
D_\tau
=
\begin{cases}
O(\tau^{2\alpha}L^K), & 0<\alpha\leq1,\\
O(\tau^{\alpha+1}L^K), & \alpha>1.
\end{cases}
\]
Since \(M_0(y,\tau)\geq B(y)\{1+o(1)\}\), \(\frac{T_0(y,\tau)-T_1(y,\tau)}{M_0(y,\tau)} \lesssim D_\tau e^{y^2/2}\).
Consequently, the deep-tail leakage is bounded by a constant multiple of \(D_\tau^2 \int_0^{\sqrt{\eta}a_\tau} y^2e^{y^2/2}\,dy \lesssim D_\tau^2 a_\tau e^{\eta a_\tau^2/2}\).
Using the threshold relation above, its ratio to \(a_\tau\phi(a_\tau)\) is, up to powers of \(L\), \(\tau^{2\alpha(1-\eta)} \quad\text{if }0<\alpha\leq1, \tau^{2(1-\alpha\eta)} \quad\text{if }\alpha>1\),
both of which vanish uniformly by the choice of \(\eta\). Above the transition, \(0\leq w\leq1\), so the contribution from \(x>M_L\) is \(O\{a_\tau\phi(a_\tau)e^{-M_L}\}\). Letting \(M_L\to\infty\) in \eqref{eq:app-null-coordinate} proves \eqref{eq:app-null-leading}.

The threshold relation displayed above and \eqref{eq:app-null-leading} give \eqref{eq:app-null-tail}. Finally, under (4), equation \eqref{eq:app-general-threshold} gives \(a_\tau^2 = 4\alpha rL\{1+o(1)\}, H(a_\tau,\tau) = 2rL\{1+o(1)\}\).
Substituting these relations and \(\tau^{2\alpha}=c^{2\alpha}(s/p)^{2\alpha r}L^{2\alpha b}\) into \eqref{eq:app-null-tail}, and using \((p-s)/p\to1\), gives \eqref{eq:app-null-general}. The critical-scale expression follows by setting \(r=1/(2\alpha)\).
\end{proof}

The restriction \(\overline\alpha<2\) in Proposition~\ref{prop:null-risk} is needed only to make the central leakage negligible relative to the transition contribution when \(\alpha>1\). In that case the former has order \(\tau^4\), whereas the latter has order \(\tau^{2\alpha}\) up to logarithmic factors. At \(\alpha=2\) the two mechanisms occur on the same exponential scale.

We next turn to the active coordinates. For \(a>0\), consider the one-sided hard-threshold rule \(\delta_a^{\rm H}(y)=y1(y>a)\) and parameterize a positive signal as \(\theta=a-t\). Its risk is
\[
H_a(t)
=
E_{a-t}\left[\{\delta_a^{\rm H}(Y)-(a-t)\}^2\right]
=
(a-t)^2\Phi(t)+t\phi(t)+\bar\Phi(t),
\]
where \(\bar\Phi=1-\Phi\).

\begin{proposition}[Least-favourable active-coordinate risk]
\label{prop:active-risk}
Suppose Assumption~1 holds, \(\tau\to0\), \(a_\tau^2\asymp L\) and \(H(a_\tau,\tau)\asymp L\). Let \(t_a\) denote a maximizer of \(H_a(t)\) over \(t<a\). Then
\begin{equation}
\label{eq:app-active-location}
t_a^2
=
2\log a-\log(8\pi)+o(1),
\qquad
a\phi(t_a)
=
2\{1+o(1)\}.
\end{equation}
For the posterior mean under (2),
\begin{equation}
\label{eq:app-active-hard}
R_\star(\pi,\tau)
=
H_{a_\tau}(t_{a_\tau})-2+o(1),
\end{equation}
and consequently
\begin{equation}
\label{eq:app-active-leading}
R_\star(\pi,\tau)
=
a_\tau^2
-
2a_\tau\{2\log a_\tau\}^{1/2}
+
o\{a_\tau(\log a_\tau)^{1/2}\}.
\end{equation}
If two global-scale sequences have squared intrinsic thresholds \(z\) and \(z+h\), with \(z=2L+o(L)\), \(z+h=2L+o(L)\) and \(|h|=o(L)\), then
\[
R_\star(z+h)-R_\star(z)
=
h\{1+o(1)\}+o(1).
\]
If, in addition, \(z=2L+O(\log L)\) and \(|h|\leq C\log L\) for fixed \(C>0\), then
\begin{equation}
\label{eq:app-active-slope}
R_\star(z+h)-R_\star(z)
=
h+o(1)
\end{equation}
uniformly over the stated range.
\end{proposition}

\begin{proof}
Differentiating \(H_a(t)\) gives \(H_a'(t) = -2(a-t)\Phi(t)+a(a-2t)\phi(t)\).
Hence an interior maximizer satisfies \(a(a-2t_a)\phi(t_a) = 2(a-t_a)\Phi(t_a)\).
The maximizing sequence has \(t_a\to\infty\) and \(t_a/a\to0\). A bounded \(t_a\) is inconsistent with the preceding stationary equation as \(a\to\infty\), while \(t_a\) of order \(a\) gives a risk strictly below that obtained for \(t=o(a)\). Thus \(\Phi(t_a)\to1\) and \((a-2t_a)/(a-t_a)\to1\), so \(a\phi(t_a)=2\{1+o(1)\}\). The Gaussian density then gives the first relation in \eqref{eq:app-active-location}.

At \(t=t_a\), Mills' ratio and \eqref{eq:app-active-location} yield \(H_a(t_a) = a^2-2at_a + O\left(t_a^2+\frac{a}{t_a}\right) = a^2-2a(2\log a)^{1/2} + o\{a(\log a)^{1/2}\}\).

It remains to transfer this expansion to the posterior-mean risk and to verify that the global supremum is attained in the same least-favourable neighbourhood. The transition and tail bounds from Section~\ref{app:uniform-expansions}, together with the leakage bounds used in Proposition~\ref{prop:null-risk}, imply the rough comparison \(\sup_{\theta\geq0} \left| R_\theta(\pi,\tau) - H_a(a-\theta) \right| = O(a), a=a_\tau\).
Indeed, the region \(|Y-a|\lesssim a^{-1}\) has probability \(O(a^{-1})\) uniformly in the signal and contributes at most \(O(a)\) to the risk difference, whereas below and above this transition neighbourhood the bounds on \(w\) and \(\kappa_T\) make the discrepancy smaller.

The stationary equation shows that \(t_a\) is the unique maximizer for all sufficiently large \(a\), and \(H_a''(t_a) = -2at_a\{1+o(1)\}\).
Let \(\delta_a=t_a^{-1/4}\). Taylor expansion near \(t_a\), together with the signs of \(H_a'\) away from its maximizer, gives \(H_a(t_a) - \sup_{|t-t_a|\geq\delta_a}H_a(t) \gg a\).
The preceding rough comparison therefore implies that any maximizer of the posterior-mean risk satisfies \(|t-t_a|\leq\delta_a=o(1)\).

Uniformly on this neighbourhood, put \(\theta=a-t\) and use the local coordinate \(y_x=(a^2+2x)^{1/2}\). Proposition~\ref{prop:transition-profile} and the second logistic identity in Section~\ref{app:uniform-expansions} give \(R_{a-t}(\pi,\tau) = H_a(t) - a\phi(t) + o(1)\).
The negative crossing \(Y<-a\) is uniformly negligible in this neighbourhood. The derivative of the transition correction is \(at\phi(t)=O(t_a)\), whereas \(H_a''(t_a)=-2at_a\{1+o(1)\}\). Hence the maximizing displacement from \(t_a\) is \(O(a^{-1})\), and the resulting change in the optimized risk is \(o(1)\). Since \(a\phi(t_a)=2\{1+o(1)\}\), this proves \eqref{eq:app-active-hard}. Combining it with the hard-threshold expansion above gives \eqref{eq:app-active-leading}.

For the threshold-perturbation statements, let \(F(z)=H_{\sqrt z}(t_{\sqrt z})\). The envelope theorem gives \(F'(z) = \frac{\sqrt z-t_{\sqrt z}}{\sqrt z} \Phi(t_{\sqrt z}) = 1+ O\left\{ \left(\frac{\log L}{L}\right)^{1/2} \right\}\)
uniformly for \(z=2L+o(L)\). The argument leading to \eqref{eq:app-active-hard} is uniform over the same neighbourhood, so \(R_\star(z) = F(z)-2+o(1)\).
If \(z\) and \(z+h\) are both \(2L+o(L)\) and \(|h|=o(L)\), integration of the preceding derivative gives \(R_\star(z+h)-R_\star(z) = h\{1+o(1)\}+o(1)\).
If \(z=2L+O(\log L)\) and \(|h|\leq C\log L\), the accumulated derivative error is \(O\left\{ \log L \left(\frac{\log L}{L}\right)^{1/2} \right\} = o(1)\),
which proves \eqref{eq:app-active-slope}.
\end{proof}

On the critical polynomial scale \(r=1/(2\alpha)\), equation \eqref{eq:app-general-threshold} gives \(a_\tau^2=2L+O(\log L)\). Hence Proposition~\ref{prop:active-risk} yields
\begin{equation}
\label{eq:app-active-critical}
R_\star(\pi,\tau)
=
2L
-
2\sqrt{2}\{L\log L\}^{1/2}
+
o\{(L\log L)^{1/2}\}.
\end{equation}
Together, \eqref{eq:app-null-general} and \eqref{eq:app-active-critical} provide the two coordinatewise contributions used in the proof of Theorem~1.

\subsection{Proof of Theorem 1}
\label{app:proof-theorem1}

\begin{proof}[of Theorem~1]
Let \(\tau=c(s/p)^rL^b\) as in (4), and write \(d=2\alpha b+\gamma\). By \eqref{eq:app-risk-decomposition}, \(R_{p,s}(\pi,\tau) = (p-s)R_0(\pi,\tau)+sR_\star(\pi,\tau)\).

Suppose first that \(r<1/(2\alpha)\). Proposition~\ref{prop:null-risk} gives \((p-s)R_0(\pi,\tau) = 2\sqrt{\frac{2\alpha r}{\pi}}\, A_{\pi,c,r}\, s(s/p)^{2\alpha r-1} L^{d-\alpha} \{1+o(1)\}\).
Since \((s/p)^{2\alpha r-1} = \exp\{(1-2\alpha r)L\}\)
diverges exponentially, \(R_{p,s}(\pi,\tau)/(sL)\to\infty\). Hence the posterior mean is not first-order minimax.

If \(r>1/(2\alpha)\), equation \eqref{eq:app-general-threshold} gives \(a_\tau^2 = 4\alpha rL\{1+o(1)\}\).
Proposition~\ref{prop:active-risk} then yields \(R_\star(\pi,\tau) = 4\alpha rL\{1+o(1)\}\),
because the second term in \eqref{eq:app-active-leading} is \(o(L)\). Consequently, \(\liminf_{L\to\infty} \frac{R_{p,s}(\pi,\tau)}{2sL} \geq 2\alpha r > 1\),
so first-order minimaxity again fails.

It remains to consider the critical polynomial calibration \(r=1/(2\alpha)\). Then \(A_{\pi,c,r}=A_{\pi,c}\), and \eqref{eq:app-general-threshold} becomes \(a_\tau^2 = 2L+(2\alpha+1-2d)\log L -2\log A_{\pi,c} +o(1)\),
which is (6). Proposition~\ref{prop:null-risk} similarly gives \((p-s)R_0(\pi,\tau) = \frac{2A_{\pi,c}}{\sqrt{\pi}}\, sL^{d-\alpha}\{1+o(1)\}\),
which is (7). Since \(a_\tau^2=2L+O(\log L)\), Proposition~\ref{prop:active-risk} gives \(sR_\star(\pi,\tau) = 2sL - 2\sqrt{2}\,s(L\log L)^{1/2} + o\{s(L\log L)^{1/2}\}\).

If \(d\leq\alpha+1/2\), then \(sL^{d-\alpha} = O(sL^{1/2}) = o\{s(L\log L)^{1/2}\}\),
and therefore \(R_{p,s}(\pi,\tau) = 2sL - 2\sqrt{2}\,s(L\log L)^{1/2} + o\{s(L\log L)^{1/2}\}\).
Thus the posterior mean is second-order minimax.

If \(\alpha+1/2<d<\alpha+1\), then \(s(L\log L)^{1/2} = o(sL^{d-\alpha}), sL^{d-\alpha} = o(sL)\).
Hence \(R_{p,s}(\pi,\tau) = 2sL + \frac{2A_{\pi,c}}{\sqrt{\pi}}\, sL^{d-\alpha}\{1+o(1)\}\).
The posterior mean therefore remains first-order minimax, but the positive null-coordinate contribution dominates the classical second-order scale, so second-order minimaxity fails.

If \(d=\alpha+1\), the null contribution satisfies \((p-s)R_0(\pi,\tau) = \frac{2A_{\pi,c}}{\sqrt{\pi}}\, sL\{1+o(1)\}\),
while \(sR_\star(\pi,\tau)=2sL\{1+o(1)\}\). Thus \(R_{p,s}(\pi,\tau) = \left( 2+\frac{2A_{\pi,c}}{\sqrt{\pi}} \right) sL\{1+o(1)\}\),
and the leading minimax constant is no longer attained.

Finally, if \(d>\alpha+1\), Proposition~\ref{prop:null-risk} alone gives
\[
\frac{R_{p,s}(\pi,\tau)}{sL}
\geq
\frac{(p-s)R_0(\pi,\tau)}{sL}
=
\frac{2A_{\pi,c}}{\sqrt{\pi}}\,
L^{d-\alpha-1}\{1+o(1)\}
\longrightarrow\infty.
\]

Hence first-order minimaxity holds precisely when \(r=\frac{1}{2\alpha}, d<\alpha+1\),
whereas second-order minimaxity holds precisely when \(r=\frac{1}{2\alpha}, d\leq\alpha+1/2\).
\end{proof}

\subsection{Unrestricted oracle calibration}
\label{app:oracle-calibration}

We now remove the deterministic restriction in (4) and minimize the worst-case risk over all \(\tau>0\). The origin regularity in Assumption~1 is used only to exclude nonvanishing global scales.

\begin{lemma}[Boundary behaviour of the risk]
\label{lem:nonvanishing-scales}
Suppose Assumption~1 holds. For every fixed \(\varepsilon>0\) there exists \(c_\varepsilon>0\), depending on \(\pi\) but not on \(p\) or \(s\), such that
\[
\inf_{\tau\geq\varepsilon}R_0(\pi,\tau)\geq c_\varepsilon.
\]
Moreover,
\[
R_\star(\pi,\tau)\longrightarrow\infty
\qquad
(\tau\downarrow0),
\]
and \(R_{p,s}(\pi,\tau)\) is continuous in \(\tau\) on \((0,\infty)\). Consequently, any sequence satisfying \(R_{p,s}(\pi,\tau)=O(sL)\) must have \(\tau\to0\).
\end{lemma}

\begin{proof}
On every compact interval \([\varepsilon,M]\subset(0,\infty)\), \(R_0(\pi,\tau)\) is continuous and strictly positive. Indeed, \(0<w(y,\tau)<1\) for \(y\neq0\), so
\[
R_0(\pi,\tau)
=
2\int_0^\infty y^2w(y,\tau)^2\phi(y)\,dy
>
0.
\]
Compactness therefore gives \(\inf_{\varepsilon\leq\tau\leq M}R_0(\pi,\tau)>0\).

It remains to consider \(\tau\to\infty\). Write the origin behaviour in Assumption~1 as \(\pi(u) = u^{\delta_0-1}\ell_0(1/u)\{1+o(1)\}, u\downarrow0\),
with \(\delta_0>0\). Under the transformation \(v=\tau^2u\),
\[
M_j(y,\tau)
=
\tau^{-2}
\int_0^\infty
(1+v)^{-j-1/2}
\exp\left\{-\frac{y^2}{2(1+v)}\right\}
\pi(v/\tau^2)\,dv.
\]
If \(0<\delta_0<1/2\), uniform convergence for regularly varying functions and Potter bounds give, for fixed \(y\), \(M_j(y,\tau) = \tau^{-2\delta_0}\ell_0(\tau^2) \{I_j(y)+o(1)\}\),
where
\[
I_j(y)
=
\int_0^\infty
v^{\delta_0-1}(1+v)^{-j-1/2}
\exp\left\{-\frac{y^2}{2(1+v)}\right\}dv.
\]
At \(y=0\), \(\frac{I_1(0)}{I_0(0)} = \frac{B(\delta_0,3/2-\delta_0)} {B(\delta_0,1/2-\delta_0)} = 1-2\delta_0\),
so \(w(0,\tau)\to2\delta_0>0\). By continuity, \(w(y,\tau)\) is bounded away from zero on a fixed interval of positive \(y\) for all sufficiently large \(\tau\), giving a positive lower bound for \(R_0(\pi,\tau)\).

If \(\delta_0>1/2\), the negative half moment \(\int_0^\infty u^{-1/2}\pi(u)\,du\) is finite and
\[
\tau M_0(y,\tau)
\longrightarrow
\int_0^\infty u^{-1/2}\pi(u)\,du,
\]
whereas the origin regular variation gives \(M_1(y,\tau)=o(\tau^{-1})\). Thus \(M_1/M_0\to0\) and \(w(y,\tau)\to1\) for every fixed \(y\). When \(\delta_0=1/2\), split the integrals at \(u=\tau^{-2}\). If the negative half moment is finite the preceding argument applies directly; otherwise the slowly varying truncated half moment dominates the corresponding numerator and again \(M_1/M_0\to0\). Hence \(\liminf_{\tau\to\infty}R_0(\pi,\tau)>0\), proving the first assertion.

For the small-scale boundary, fix \(\theta>1\). The posterior-integral representation and dominated convergence give \(w(y,\tau)\to0\) uniformly for \(y\) in compact sets as \(\tau\downarrow0\). Hence, on \(|Y-\theta|\leq1\), \(\widehat\theta_{\pi,\tau}(Y)=o(1)\)
uniformly, and therefore \(\liminf_{\tau\downarrow0}R_\star(\pi,\tau) \geq P(|Z|\leq1)(\theta-1)^2\).
Letting \(\theta\to\infty\) proves \(R_\star(\pi,\tau)\to\infty\).

Finally, the posterior-integral representation gives joint continuity of \(R_\theta(\pi,\tau)\) in \((\theta,\tau)\). On compact \(\tau\)-intervals, the upper-tail condition in Assumption~1 gives \(\widehat\theta_{\pi,\tau}(y)=y+O(y^{-1})\) uniformly as \(|y|\to\infty\), so \(R_\theta(\pi,\tau)\to1\) uniformly as \(|\theta|\to\infty\). The supremum defining \(R_\star\) may therefore be restricted to a compact signal interval locally in \(\tau\), which proves continuity of \(R_\star\), and hence of \(R_{p,s}\).

If \(\tau\) stays bounded away from zero, the first assertion and \eqref{eq:app-risk-decomposition} give \(R_{p,s}(\pi,\tau)\gtrsim p\). Since \(\frac{p}{sL} = \frac{e^L}{L} \longrightarrow\infty\),
such a sequence cannot have risk \(O(sL)\).
\end{proof}

\begin{proof}[of Theorem~2]
Define
\[
\tau_\pi^\star
=
\left\{
\frac{\sqrt{\pi}\,\alpha^\gamma}
{C_\pi\Gamma(\alpha+1/2)}
\right\}^{1/(2\alpha)}
(s/p)^{1/(2\alpha)}
L^{(\alpha-\gamma)/(2\alpha)}.
\]
This is a member of the deterministic family (4) with \(r^\star=\frac{1}{2\alpha}, b^\star=\frac{\alpha-\gamma}{2\alpha}, d^\star=\alpha\),
and its constant \(c_\pi^\star\) satisfies \(C_\pi\Gamma(\alpha+1/2) (c_\pi^\star)^{2\alpha}\alpha^{-\gamma} = \sqrt{\pi}\).
Equation~\eqref{eq:app-general-threshold} therefore gives \(a_{\tau_\pi^\star}^2 = 2L+\log L-\log\pi+o(1)\).
Propositions~\ref{prop:null-risk} and \ref{prop:active-risk} then give \((p-s)R_0(\pi,\tau_\pi^\star) = 2s\{1+o(1)\}\),
and \(sR_\star(\pi,\tau_\pi^\star) = 2sL - 2\sqrt{2}\,s(L\log L)^{1/2} + o\{s(L\log L)^{1/2}\}\).
In particular, \(R_{p,s}(\pi,\tau_\pi^\star)=O(sL)\).

Lemma~\ref{lem:nonvanishing-scales}, the preceding candidate bound and the continuity of the risk imply that, for all sufficiently large \(L\), the infimum over \(\tau>0\) is attained. Let \(\tau_\pi^{\mathrm{or}}\) be any minimizer. Since its risk cannot exceed the candidate risk, Lemma~\ref{lem:nonvanishing-scales} also gives \(\tau_\pi^{\mathrm{or}}\longrightarrow0\).
For brevity write \(\tau=\tau_\pi^{\mathrm{or}}, a=a_\tau, \lambda=\log(1/\tau)\).
For every fixed \(y\), dominated convergence gives \(w(y,\tau)\to0\), and therefore \(a\to\infty\).

We first localize the intrinsic threshold. Take \(\theta=a-1\). On \(0\leq Y-\theta\leq1\), we have \(a-1\leq Y\leq a\), so Lemma~\ref{lem:posterior-identities} gives \(w(Y,\tau)\leq1/2\). Hence \(R_\star(\pi,\tau) \geq P(0\leq Z\leq1) \left(\frac{a}{2}-1\right)^2 \gtrsim a^2\),
and the candidate bound implies \(a^2=O(L)\). Conversely, \(w(y,\tau)\geq1/2\) for \(y\geq a\), so
\[
R_0(\pi,\tau)
\geq
\frac12\int_a^\infty y^2\phi(y)\,dy
\asymp
a\phi(a).
\]
Since \((p-s)R_0(\pi,\tau)=O(sL)\), \(pa\phi(a)\lesssim sL\),
and hence \(a^2\geq2L-O(\log L)\).
Thus \(a^2\asymp L\).

Before applying the coarse threshold balance, we verify that the oracle scale is exponentially small. For \(1\leq y\leq2\),
\[
M_0(y,\tau)-M_1(y,\tau)
=
\int_0^\infty
\frac{\tau^2u}{(1+\tau^2u)^{3/2}}
\exp\left\{-\frac{y^2}{2(1+\tau^2u)}\right\}
\pi(u)\,du.
\]
Since \(\pi\) is positive on \((0,\infty)\), integration over \(1\leq u\leq2\) gives \(M_0(y,\tau)-M_1(y,\tau)\geq c_\pi\tau^2\) for all sufficiently small \(\tau\), uniformly over \(1\leq y\leq2\). As \(M_0(y,\tau)\leq1\), \(w(y,\tau)\geq c_\pi\tau^2, 1\leq y\leq2\),
and consequently \(R_0(\pi,\tau)\geq c_\pi'\tau^4\). The oracle risk bound therefore gives \(\tau^4\lesssim\frac{sL}{p}, \lambda \geq \frac L4-O(\log L)\).
Together with \(a^2=O(L)\), this implies \(a^2\tau\to0\) and \(H(a,\tau)\to\infty\). Equation~\eqref{eq:app-coarse-threshold-balance} is therefore applicable and gives \(a^2 = 4\alpha\lambda + O\{\log a+\log(1+\lambda)\}\).
Since \(a^2=O(L)\) and \(\lambda\gtrsim L\), it follows that \(\lambda\asymp L, H(a,\tau)\asymp L\).
The sharp expansions in Sections~\ref{app:uniform-expansions} and \ref{app:coordinate-risk} thus apply to the oracle sequence.

We next sharpen the threshold location. If, along a subsequence, \(a^2\geq(2+\eta)L\) for some fixed \(\eta>0\), Proposition~\ref{prop:active-risk} gives \(sR_\star(\pi,\tau) \geq (2+\eta/2)sL\)
for all sufficiently large \(L\), whereas the candidate risk is \(2sL+o(sL)\). This is impossible for a minimizer. Together with the lower bound above, \(a^2=2L+o(L)\).

Set \(z_0=2L+\log L-\log\pi, h=a^2-z_0\).
Then \(h=o(L)\). Proposition~\ref{prop:null-risk} gives \((p-s)R_0(\pi,\tau) = 2se^{-h/2}\{1+o(1)\}\),
because \(a/\sqrt{2L}\to1\) and \((p-s)/p\to1\). The coarse perturbation result in Proposition~\ref{prop:active-risk} gives \(R_\star(a^2)-R_\star(z_0) = h\{1+o(1)\}+o(1)\)
for \(|h|=o(L)\), while \eqref{eq:app-active-slope} gives \(R_\star(a^2)-R_\star(z_0) = h+o(1)\)
for bounded \(h\). The candidate \(\tau_\pi^\star\) has squared threshold \(z_0+o(1)\).

If \(h\to+\infty\), the active contribution increases by \(sh\{1+o(1)\}\), whereas the largest possible reduction of the null contribution relative to the candidate is \(2s\{1+o(1)\}\). If \(h\to-\infty\), the null contribution increases as \(2se^{|h|/2}\), whereas the active-coordinate reduction is only \(O(s|h|)\). Neither case can occur along a minimizing sequence, so \(h=O(1)\).

For bounded \(h\), the preceding expansions and their counterparts at \(\tau_\pi^\star\) give
\begin{equation}
\label{eq:app-oracle-local-objective}
\frac{
R_{p,s}(\pi,\tau)
-
R_{p,s}(\pi,\tau_\pi^\star)
}{s}
=
h+2e^{-h/2}-2+o(1),
\end{equation}
uniformly on compact sets. The limiting function \(g(h)=h+2e^{-h/2}-2\) satisfies \(g'(h)=1-e^{-h/2}, g''(h)=1/2e^{-h/2}>0\),
and has the unique global minimum at \(h=0\). Since every minimizing sequence has bounded \(h\), \eqref{eq:app-oracle-local-objective} implies \(a_{\tau_\pi^{\mathrm{or}}}^2 = 2L+\log L-\log\pi+o(1)\),
which proves (9).

The derivative-controlled expansions in Proposition~\ref{prop:uniform-tail}, Proposition~\ref{prop:null-risk} and the envelope calculation in Proposition~\ref{prop:active-risk} also give, along the local small-\(\tau\) branch indexed by the squared threshold \(z=z_0+h\), \(\frac{\partial^2}{\partial z^2} R_{p,s}\{\pi,\tau(z)\} = \frac{s}{2}\{1+o(1)\}\) uniformly for \(h=o(1)\). This local curvature will be used in the sensitivity result below and in the oracle-risk comparison. This local curvature will be used in the oracle-risk comparison below.

It remains to recover the global scale. At \(a=a_{\tau_\pi^{\mathrm{or}}}\), equations \eqref{eq:app-threshold-logQ} and \eqref{eq:app-tail-expansion} give
\[
0
=
\frac{a^2}{2}
+
2\alpha\log\tau
-
(2\alpha+1)\log a
+
\gamma\log H(a,\tau)
+
\log\{C_\pi2^{\alpha+1/2}\Gamma(\alpha+1/2)\}
+
o(1).
\]
Equation~\eqref{eq:app-coarse-threshold-balance} and the localized threshold imply \(-\log\tau = \frac{L}{2\alpha}+O(\log L), H(a,\tau) = \frac{L}{\alpha}\{1+o(1)\}\).
Thus \(\log H(a,\tau) = \log L-\log\alpha+o(1), \log a = 1/2\log(2L)+o(1)\).
Substitution gives
\[
2\alpha\log\tau_\pi^{\mathrm{or}}
=
-L
+
(\alpha-\gamma)\log L
+
\frac12\log\pi
+
\gamma\log\alpha
-
\log C_\pi
-
\log\Gamma(\alpha+1/2)
+
o(1).
\]
Exponentiating,
\[
\tau_\pi^{\mathrm{or}}
=
\left\{
\frac{\sqrt{\pi}\,\alpha^\gamma}
{C_\pi\Gamma(\alpha+1/2)}
\right\}^{1/(2\alpha)}
(s/p)^{1/(2\alpha)}
L^{(\alpha-\gamma)/(2\alpha)}
\{1+o(1)\}.
\]
Therefore \(\frac{\tau_\pi^{\mathrm{or}}}{\tau_\pi^\star} \longrightarrow1\),
which proves (8). The expressions for \(r^\star\), \(b^\star\), \(d^\star\) and \(c_\pi^\star\) follow directly from the definition of \(\tau_\pi^\star\).
\end{proof}

\begin{proof}[of Proposition~1]
Let \(z(\tau)=a_\tau^2\) and write \(z_{\rm or}=z(\tau_\pi^{\rm or})\). Theorem~2 gives \(z_{\rm or}=2L+\log L-\log\pi+o(1)\). We first determine how the intrinsic threshold changes under a local perturbation of the global scale.

Along the local small-\(\tau\) branch, the threshold relation obtained from \eqref{eq:app-threshold-logQ} and \eqref{eq:app-tail-expansion} may be written as
\[
0=\frac{z}{2}+2\alpha\log\tau-\frac{2\alpha+1}{2}\log z+\gamma\log H(z,\tau)+C_\pi^\dagger+o(1),
\qquad
H(z,\tau)=\log\frac{z}{2\tau^2},
\]
where \(C_\pi^\dagger\) does not depend on \(z\) or \(\tau\), and the remainder may be differentiated uniformly for \(z=z_{\rm or}+o(1)\). Differentiating with respect to \(\log\tau\), and using \(z\asymp H(z,\tau)\asymp L\), gives
\[
\frac12\frac{dz}{d\log\tau}+2\alpha+O(L^{-1})\left\{1+\left|\frac{dz}{d\log\tau}\right|\right\}=o(1).
\]
Hence
\[
\frac{dz}{d\log\tau}=-4\alpha+o(1)
\]
uniformly for \(z=z_{\rm or}+o(1)\). Therefore, if \(t=o(1)\),
\[
z(\tau_\pi^{\rm or}e^t)-z_{\rm or}=-4\alpha t\{1+o(1)\}.
\]

Write \(\mathcal R(z)=R_{p,s}\{\pi,\tau(z)\}\) on this local branch. Since \(\tau_\pi^{\rm or}\) is an interior minimizer and \(dz/d\log\tau=-4\alpha+o(1)\) is bounded away from zero, \(\mathcal R'(z_{\rm or})=0\). The local curvature calculation in the proof of Theorem~2 gives
\[
\mathcal R''(z)=\frac{s}{2}\{1+o(1)\}
\]
uniformly for \(z=z_{\rm or}+o(1)\). Taylor expansion between \(z_{\rm or}\) and \(z(\tau_\pi^{\rm or}e^t)\) consequently yields
\[
\begin{split}
R_{p,s}\{\pi,\tau_\pi^{\rm or}e^t\}-R_{p,s}^{\rm or}(\pi)
&=\frac12\mathcal R''(\widetilde z)\{z(\tau_\pi^{\rm or}e^t)-z_{\rm or}\}^2\\
&=4\alpha^2st^2\{1+o(1)\},
\end{split}
\]
where \(\widetilde z\) lies between the two squared thresholds.

If two oracle risks first separate at order \(s/L^k\), the excess risk generated by replacing \(\tau_{\pi_j}^{\rm or}\) by \(\tau_{\pi_j}^{\rm or}e^{t_j}\) is therefore \(o(s/L^k)\) whenever \(\max_j|t_j|=o(L^{-k/2})\). If instead \(t_j=c_jL^{-k/2}\) for a fixed nonzero \(c_j\), the induced excess risk is \(4\alpha_j^2c_j^2s/L^k\{1+o(1)\}\), and hence occurs at the same order as the first separating term. This proves the proposition.
\end{proof}

\begin{proof}[of Corollary~1]
For the horseshoe, \(\alpha=1/2\), \(\gamma=0\), \(C_\pi=1/\pi\) and \(\Gamma(\alpha+1/2)=1\). Substitution into (8) gives \(\tau_{\rm H}^\star = \pi^{3/2}(s/p)L^{1/2}\).
For the horseshoe+, \(\alpha=1/2\), \(\gamma=1\) and \(C_\pi=1/\pi^2\), so \(\tau_{\rm H+}^\star = \frac{\pi^{5/2}}{2}(s/p)L^{-1/2}\).
The convergence of the unrestricted minimizers follows from Theorem~2, and (9) gives the common oracle transition threshold.
\end{proof}

\subsection{Matched-threshold risk comparison}
\label{app:matched-threshold}

We now compare two priors after aligning their intrinsic transition thresholds. This removes the leading discrepancy caused by global-scale calibration and isolates the prior-dependent perturbation of the transition profile. The comparison is first carried out at a common threshold and then transferred to the separately optimized oracle scales.

Let \(z_0=2L+\log L-\log\pi\).
For a prior \(\pi_j\), let \(\tau_j(z)\) denote the solution on the small-\(\tau\) branch satisfying \(w_j(\sqrt z,\tau_j(z))=1/2\),
and define
\[
H_j(z)=\log\frac{z}{2\tau_j(z)^2},
\qquad
\kappa_j(z)
=
\frac{T_{1,j}(\sqrt z,\tau_j(z))}
{T_{0,j}(\sqrt z,\tau_j(z))},
\qquad
\Delta\kappa(z)=\kappa_2(z)-\kappa_1(z).
\]
Local existence and uniqueness follow from the threshold expansion in Section~\ref{app:uniform-expansions}. Differentiating its leading balance gives
\[
\frac{dz}{d\log\tau_j}
=
-4\alpha_j+o(1),
\qquad
\frac{dH_j(z)}{dz}
=
\frac{1}{2\alpha_j}+O(L^{-1})
\]
uniformly for \(z=z_0+o(1)\).

For the pair \((\pi_1,\pi_2)\), define
\[
\delta_L
=
\begin{cases}
L^{-1}, & \alpha_1\neq\alpha_2,\\
L^{-2}, & \alpha_1=\alpha_2,\ \gamma_1\neq\gamma_2,\\
L^{-3}, & \alpha_1=\alpha_2,\ \gamma_1=\gamma_2.
\end{cases}
\]
The last case includes the possibility that the two priors agree at order \(L^{-3}\); the exact coefficient is identified in Section~\ref{app:higher-order}.

The full expansion in \eqref{eq:app-kappa-expansion} may be written, with \(q_j=\alpha_j+1/2\), as
\[
\kappa_j(z)
=
\frac{2q_j}{z}
-\frac{2\gamma_j}{zH_j}
+\frac{2q_j(2q_j+1)}{z^2}
+\frac{2\{\beta_j-\gamma_j\psi(q_j)\}}{zH_j^2}-\frac{2\gamma_j(4q_j+1)}{z^2H_j}
+\frac{2q_j(2q_j+1)(4q_j+3)}{z^3}
+O(L^{-4}).
\]
The matched-threshold balance implies \(H_j\asymp L\); if \(\alpha_1=\alpha_2\), then \(H_2-H_1=O(\log L)\), and if both \(\alpha\) and \(\gamma\) agree, then \(H_2-H_1=O(1)\). Consequently, \(\Delta\kappa(z)=O(\delta_L), \frac{d}{dz}\Delta\kappa(z) = O(\delta_L/L)\)
uniformly for \(z=z_0+o(1)\).

\begin{lemma}[Pairwise transition perturbation]
\label{lem:pairwise-transition}
Let \(\pi_1\) and \(\pi_2\) satisfy Assumption~1, let \(z=z_0+o(1)\), and write \(a=\sqrt z\). Fix \(K>4\), set \(M_L=K\log L\), and for \(|x|\leq M_L\) define
\[
y_x=(z+2x)^{1/2}.
\]
Then
\begin{equation}
\label{eq:app-pairwise-transition}
w_2\{y_x,\tau_2(z)\}
-
w_1\{y_x,\tau_1(z)\}
=
\Delta\kappa(z)h(x)+r_L(x),
\end{equation}
where \(h(x)=\ell(x)\{1-\ell(x)\}(2-x)-\ell(x)\) and
\[
\sup_{|x|\leq M_L}|r_L(x)|=o(\delta_L).
\]
Moreover, the remainder satisfies the weighted bounds
\[
\int_{-M_L}^{M_L}
e^{-x}\ell(x)|r_L(x)|\,dx
=
o(\delta_L),
\qquad
\int_{-M_L}^{M_L}
\{1-\ell(x)\}|r_L(x)|\,dx
=
o(\delta_L).
\]
The same statements hold after one differentiation with respect to \(z\), with \(\delta_L\) replaced by \(\delta_L/L\).
\end{lemma}

\begin{proof}
For brevity write \(Q_j(x) = Q_j\{y_x,\tau_j(z)\}, \kappa_{j,x} = \kappa_{T,j}\{y_x,\tau_j(z)\}\).
The derivative bound above gives \(\kappa_{2,x}-\kappa_{1,x} = \Delta\kappa(z) + O\left(\frac{|x|\delta_L}{L}\right)\)
uniformly for \(|x|\leq M_L\).

At the common threshold, \eqref{eq:app-threshold-Q} gives \(Q_j(0) = \frac{1}{1-2\kappa_j(z)} \{1+o(L^{-m})\}\)
for every fixed \(m>0\). Since \eqref{eq:app-Q-derivative} and \(y_x^2=z+2x\) imply \(\frac{d}{dx}\log Q_j(x) = 1-\kappa_{j,x}\),
subtraction between the two priors yields
\[
\begin{split}
\log Q_2(x)-\log Q_1(x)
={}&
-\log\{1-2\kappa_2(z)\}
+\log\{1-2\kappa_1(z)\}\\
&-
\int_0^x
\{\kappa_{2,u}-\kappa_{1,u}\}\,du
+
o(L^{-m}).
\end{split}
\]
Because \(\kappa_j(z)=O(L^{-1})\), \(-\log\{1-2\kappa_2(z)\} + \log\{1-2\kappa_1(z)\} = 2\Delta\kappa(z) + O(\delta_L/L)\),
while \(\int_0^x \{\kappa_{2,u}-\kappa_{1,u}\}\,du = x\Delta\kappa(z) + O\left(\frac{x^2\delta_L}{L}\right)\).
Since \(M_L^2/L\to0\), \(\log Q_2(x)-\log Q_1(x) = (2-x)\Delta\kappa(z)+o(\delta_L)\)
uniformly on the transition window.

For each \(j\), Proposition~\ref{prop:transition-profile} gives \(\log Q_j(x) = x+O\{(1+|x|)/L\}\).
A first-order Taylor expansion of \(q\mapsto e^q/(1+e^q)\) therefore gives \(\frac{Q_2(x)}{1+Q_2(x)} - \frac{Q_1(x)}{1+Q_1(x)} = \ell(x)\{1-\ell(x)\}(2-x)\Delta\kappa(z) + o(\delta_L)\).
Combining this with \eqref{eq:app-wQ} and \(\kappa_{2,x}-\kappa_{1,x} = \Delta\kappa(z)+o(\delta_L)\)
proves \eqref{eq:app-pairwise-transition}.

The Taylor remainders inherit a factor \(\ell(x)\{1-\ell(x)\}\) from the logistic map, while the direct shrinkage term inherits a factor \(\ell(x)\). Together with \(M_L=K\log L\), these bounds are integrable against the two weights displayed in the statement and give the asserted \(o(\delta_L)\) integrated remainders. Differentiating the preceding identities with respect to \(z\), using Proposition~\ref{prop:uniform-tail} and \(d\Delta\kappa/dz=O(\delta_L/L)\), gives the final assertion.
\end{proof}

\begin{proposition}[Matched-threshold and oracle risk comparison]
\label{prop:matched-comparison}
Let \(\pi_1\) and \(\pi_2\) satisfy Assumption~1. Uniformly for matched thresholds \(z=z_0+o(1)\),
\begin{equation}
\label{eq:app-matched-risk}
R_{p,s}\{\pi_2,\tau_2(z)\}
-
R_{p,s}\{\pi_1,\tau_1(z)\}
=
-4s\Delta\kappa(z)
+
o(s\delta_L).
\end{equation}
Moreover,
\[
\frac{d}{dz}
\left[
R_{p,s}\{\pi_2,\tau_2(z)\}
-
R_{p,s}\{\pi_1,\tau_1(z)\}
\right]
=
-4s\Delta\kappa'(z)
+
o(s\delta_L/L).
\]
If \(z_j^{\rm or}=a_{\tau_{\pi_j}^{\rm or}}^2\), then $z_2^{\rm or}-z_1^{\rm or}=O(\delta_L/L)$, and separate oracle reoptimization changes the matched-threshold comparison by $O\{s(\delta_L/L)^2\}=O(s/L^4)$. Consequently,
\begin{equation}
\label{eq:app-oracle-kappa-comparison}
R_{p,s}^{\rm or}(\pi_2)
-
R_{p,s}^{\rm or}(\pi_1)
=
-4s\{\kappa_2(z_0)-\kappa_1(z_0)\}
+
o(s\delta_L).
\end{equation}
\end{proposition}

\begin{proof}
We first compare the null-coordinate risks. By \eqref{eq:app-null-coordinate},
\[
R_{0,j}
=
2a\phi(a)
\int_{-a^2/2}^{\infty}
\left(1+\frac{2x}{a^2}\right)^{1/2}
e^{-x}
w_j\{y_x,\tau_j(z)\}^2\,dx.
\]
On \(|x|\leq M_L\), Lemma~\ref{lem:pairwise-transition} and Proposition~\ref{prop:transition-profile} give \(w_2^2-w_1^2 = 2\ell(x)h(x)\Delta\kappa(z) + r_{0,L}(x)\),
where \(\int_{-M_L}^{M_L} e^{-x}|r_{0,L}(x)|\,dx = o(\delta_L)\).
Using the logistic identity recorded in Section~\ref{app:uniform-expansions}, \(2\int_{\mathbb R}e^{-x}\ell(x)h(x)\,dx=1\),
we obtain \(R_{0,2}-R_{0,1} = 2a\phi(a)\Delta\kappa(z) + o\{a\phi(a)\delta_L\}\).

The parts outside the transition window are negligible at the same scale. For \(x>M_L\), \(0\leq w_j\leq1\) gives an \(O\{a\phi(a)e^{-M_L}\}\) contribution. On the moderate lower region, the \(Q\)-slope bound used in Proposition~\ref{prop:null-risk} gives \(w_j\lesssim e^x\), and the contribution is again \(O\{a\phi(a)e^{-M_L}\}\). The deep central and tail leakages are exponentially smaller than \(a\phi(a)L^{-m}\) for every fixed \(m\). Since \(K>4\), all three terms are \(o\{a\phi(a)\delta_L\}\). Because \(z=z_0+o(1)\), \((p-s)a\phi(a) = s\{1+o(1)\}\). Hence \((p-s)(R_{0,2}-R_{0,1}) = 2s\Delta\kappa(z) + o(s\delta_L)\).

We next compare the active-coordinate risks. Let \(t_a\) be the least-favourable displacement in Proposition~\ref{prop:active-risk}, so that \(a\phi(t_a)=2\{1+o(1)\}\). The first prior-dependent contribution comes from smoothing the transition. At the signal \(\theta=a-t_a\), Lemma~\ref{lem:pairwise-transition} gives the first variation \(-2a^2\{1-\ell(x)\}h(x)\Delta\kappa(z)\) of the leading squared-error term. After multiplication by the local Gaussian density and transition width, the second logistic identity in Section~\ref{app:uniform-expansions} gives \(R_{\star,2}^{\rm tr} - R_{\star,1}^{\rm tr} = -2\Delta\kappa(z)+o(\delta_L)\).

There is a second contribution above the transition. Write \(Y=a-t_a+Z\) and fix a sufficiently large constant \(K_0\). On \(t_a<Z\leq K_0t_a\), we have \(Y^2-z = O(at_a) = O\{(L\log L)^{1/2}\} = o(L)\). The derivative bound for the pairwise tail-shrinkage coefficient therefore gives \(\kappa_2(Y)-\kappa_1(Y) = \Delta\kappa(z)+o(\delta_L)\) uniformly on this region. Proposition~\ref{prop:uniform-tail} also gives \(\widehat\theta_{\pi_j,\tau_j(z)}(Y) = Y-Y\kappa_j(Y)+o\{Y\kappa_j(Y)\}\). The linear tail correction is consequently \(-2\Delta\kappa(z) E\{ZY1(Z>t_a)\} + o(\delta_L)\). Since \(E\{ZY1(Z>t_a)\} = a\phi(t_a)+\bar\Phi(t_a) = 2+o(1)\), this contribution equals \(-4\Delta\kappa(z)+o(\delta_L)\). The Gaussian tail beyond \(K_0t_a\) is \(o(L^{-m})\) for any required fixed \(m\) when \(K_0\) is chosen sufficiently large.

The quadratic shrinkage term is smaller. Indeed, \(\kappa_2(Y)^2-\kappa_1(Y)^2 = O(L^{-1}\delta_L)\) throughout the relevant exceedance region, whereas \(E\{Y^2 1(Z>t_a)\} = O(a/t_a)\). Its pairwise contribution is therefore \(o(\delta_L)\).

The perturbation of the least-favourable signal is also negligible. Proposition~\ref{prop:active-risk} gives \(H_a''(t_a) = -2at_a\{1+o(1)\}\), whereas the derivative of the prior-dependent perturbation is \(O(t_a\delta_L)\). The maximizing displacement is therefore \(O(\delta_L/a)\), and its effect on the optimized risk is \(o(\delta_L)\). Combining the transition and above-threshold contributions gives \(R_{\star,2}-R_{\star,1} = -6\Delta\kappa(z)+o(\delta_L)\).

Together with the null-coordinate comparison and \eqref{eq:app-risk-decomposition}, \(2s\Delta\kappa(z)-6s\Delta\kappa(z) = -4s\Delta\kappa(z)\), which proves \eqref{eq:app-matched-risk}. Differentiating the pairwise calculations above with respect to \(z\), using the derivative assertion in Lemma~\ref{lem:pairwise-transition}, gives
\[
\frac{d}{dz}
\left[
R_{p,s}\{\pi_2,\tau_2(z)\}
-
R_{p,s}\{\pi_1,\tau_1(z)\}
\right]
=
-4s\Delta\kappa'(z)
+
o(s\delta_L/L).
\]

It remains to account for separate oracle reoptimization. Write \(R_j(z) = R_{p,s}\{\pi_j,\tau_j(z)\}\). By Theorem~2, \(z_j^{\rm or} = z_0+o(1), j=1,2\), and Section~\ref{app:oracle-calibration} gives \(R_j''(z) = \frac{s}{2}\{1+o(1)\}\) uniformly in an \(o(1)\) neighbourhood of \(z_0\).

Since \(R_1'(z_1^{\rm or})=0\), the derivative comparison above gives \(R_2'(z_1^{\rm or}) = O(s\delta_L/L)\). The mean-value theorem and the local curvature therefore imply \(z_2^{\rm or}-z_1^{\rm or} = O(\delta_L/L)\). Taylor expansion around \(z_1^{\rm or}\) then gives \(R_2(z_2^{\rm or}) - R_2(z_1^{\rm or}) = O\{s(\delta_L/L)^2\}\). Consequently, \(R_{p,s}^{\rm or}(\pi_2) - R_{p,s}^{\rm or}(\pi_1) = -4s\Delta\kappa(z_1^{\rm or}) + o(s\delta_L)\). Finally, \(\Delta\kappa(z_1^{\rm or}) - \Delta\kappa(z_0) = O(\delta_L/L)\,o(1) = o(\delta_L)\), because \(z_1^{\rm or}=z_0+o(1)\). This proves \eqref{eq:app-oracle-kappa-comparison}.
\end{proof}

\subsection{Higher-order oracle ranking}
\label{app:higher-order}

We now identify the tail characteristics that determine the first separating orders among oracle-calibrated priors. Section~\ref{app:matched-threshold} reduces the comparison to the difference between the tail-shrinkage coefficients at the common oracle threshold \(z_0=2L+\log L-\log\pi\).
The remaining task is therefore algebraic: we expand \(\kappa_j(z_0)\) and identify the first prior-dependent term that survives after coarser tail characteristics have been matched.

\begin{lemma}[Scale invariance of the third descriptor]
\label{lem:zeta-invariance}
For \(k>0\), let $\pi^{(k)}(u)=k^{-1}\pi(u/k)$, corresponding to the fixed rescaling \(u\mapsto ku\). If \(\pi\) has tail parameters \((\alpha,\gamma,\beta,C_\pi)\), then
\[
\alpha^{(k)}=\alpha,
\qquad
\gamma^{(k)}=\gamma,
\qquad
C_{\pi^{(k)}}=C_\pi k^\alpha,
\qquad
\beta^{(k)}=\beta-\gamma\log k.
\]
Consequently,
\[
\beta^{(k)}
+
\frac{\gamma^{(k)}}{\alpha^{(k)}}
\log C_{\pi^{(k)}}
=
\beta+\frac{\gamma}{\alpha}\log C_\pi.
\]
Thus $\zeta_\pi=\beta+\gamma\log C_\pi/\alpha$ is invariant under fixed rescaling of the local variance.
\end{lemma}

\begin{proof}
Substituting the tail expansion in Assumption~1,
\[
\begin{split}
\pi^{(k)}(u)
&=
k^{-1}\pi(u/k)\\
&=
C_\pi k^\alpha
u^{-\alpha-1}
(\log u-\log k)^\gamma
\left\{
1+\frac{\beta}{\log u-\log k}
+O\{(\log u)^{-2}\}
\right\}.
\end{split}
\]
Since \((\log u-\log k)^\gamma = (\log u)^\gamma \left\{ 1-\frac{\gamma\log k}{\log u} +O\{(\log u)^{-2}\} \right\}\),
the transformed parameters follow immediately. Substitution into the definition of \(\zeta_\pi\) proves the invariance.
\end{proof}

\begin{proof}[of Theorem~3]
For \(j=1,2\), let \(\tau_j=\tau_j(z_0)\) be the matched-scale branch from Section~\ref{app:matched-threshold}, and write \(H_j = \log\frac{z_0}{2\tau_j^2}, \kappa_j = \kappa_j(z_0)\).
At the matched threshold, equations \eqref{eq:app-threshold-Q} and \eqref{eq:app-tail-expansion} give
\[
0
=
\frac{z_0}{2}
+
2\alpha_j\log\tau_j
-
(2\alpha_j+1)\log(z_0^{1/2})
+
\gamma_j\log H_j
+
\log\{C_{\pi_j}2^{\alpha_j+1/2}\Gamma(\alpha_j+1/2)\}
+
o(1).
\]
Since \(2\alpha_j\log\tau_j = \alpha_j\log(z_0/2)-\alpha_jH_j\),
this becomes \(\alpha_jH_j = \frac{z_0}{2} - 1/2\log(z_0/2) + \gamma_j\log H_j + \log C_{\pi_j} + \log\Gamma(\alpha_j+1/2) + o(1)\).
Now \(\frac{z_0}{2} - 1/2\log(z_0/2) = L-1/2\log\pi+o(1)\),
while the coarse balance in Section~\ref{app:uniform-expansions} gives \(H_j=L/\alpha_j+O(\log L)\). Hence \(\log H_j = \log L-\log\alpha_j+o(1)\),
and therefore
\begin{equation}
\label{eq:app-H-third-order}
H_j
=
\frac{1}{\alpha_j}
\left\{
L+\gamma_j\log L+\log C_{\pi_j}
+
K_{\alpha_j,\gamma_j}
\right\}
+
o(1),
\end{equation}
where \(K_{\alpha,\gamma} = \log\Gamma(\alpha+1/2) - \gamma\log\alpha - 1/2\log\pi\).

Suppose first that \(\alpha_1\neq\alpha_2\). Equation~\eqref{eq:app-kappa-expansion} gives \(\kappa_j = \frac{2\alpha_j+1}{z_0} + O(L^{-2}\log L)\),
so \(\kappa_2-\kappa_1 = \frac{2(\alpha_2-\alpha_1)}{z_0} + o(L^{-1}) = (\alpha_2-\alpha_1)L^{-1} + o(L^{-1})\).
Equation~\eqref{eq:app-oracle-kappa-comparison} therefore yields \(R_{p,s}^{\mathrm{or}}(\pi_2) - R_{p,s}^{\mathrm{or}}(\pi_1) = -4(\alpha_2-\alpha_1)\frac{s}{L} + o(s/L)\).

Now suppose \(\alpha_1=\alpha_2=\alpha\) but \(\gamma_1\neq\gamma_2\). Put \(q=\alpha+1/2\). The leading \(2q/z_0\) term and the pure \(z_0^{-2}\) term in \eqref{eq:app-kappa-expansion} are common to the pair. From \eqref{eq:app-H-third-order}, \(H_j = \frac{L}{\alpha}\{1+o(1)\}\),
and hence \(\frac{\gamma_j}{H_j} = \frac{\alpha\gamma_j}{L} + o(L^{-1})\).
Therefore
\[
\begin{split}
\kappa_2-\kappa_1
&=
-\frac{2}{z_0}
\left(
\frac{\gamma_2}{H_2}
-
\frac{\gamma_1}{H_1}
\right)
+
o(L^{-2})\\
&=
-\alpha(\gamma_2-\gamma_1)L^{-2}
+
o(L^{-2}).
\end{split}
\]
Another application of \eqref{eq:app-oracle-kappa-comparison} gives \(R_{p,s}^{\mathrm{or}}(\pi_2) - R_{p,s}^{\mathrm{or}}(\pi_1) = 4\alpha(\gamma_2-\gamma_1)\frac{s}{L^2} + o(s/L^2)\).

Finally suppose \(\alpha_1=\alpha_2=\alpha\) and \(\gamma_1=\gamma_2=\gamma\). Equation~\eqref{eq:app-H-third-order} gives \(H_2-H_1 = \frac{1}{\alpha} \log\frac{C_{\pi_2}}{C_{\pi_1}} + o(1)\),
and therefore \(\frac1{H_2}-\frac1{H_1} = -\frac{\alpha}{L^2} \log\frac{C_{\pi_2}}{C_{\pi_1}} + o(L^{-2})\).
At order \(L^{-3}\), the pair-specific terms in \eqref{eq:app-kappa-expansion} are the direct \(H^{-2}\) contribution involving \(\beta_j\) and the effect of the shift in \(H_j\) on the \(-\gamma/H_j\) term. Since \(z_0=2L\{1+o(1)\}\) and \(H_j=L/\alpha\{1+o(1)\}\), \(\frac{2(\beta_2-\beta_1)}{z_0H_j^2} = \alpha^2(\beta_2-\beta_1)L^{-3} + o(L^{-3})\),
while
\[
-\frac{2\gamma}{z_0}
\left(
\frac1{H_2}-\frac1{H_1}
\right)
=
\alpha\gamma
\log\frac{C_{\pi_2}}{C_{\pi_1}}
L^{-3}
+
o(L^{-3}).
\]
All remaining displayed terms in \eqref{eq:app-kappa-expansion} are common to the pair or differ by \(o(L^{-3})\). Proposition~\ref{prop:uniform-tail} also shows that the \(\chi\) and remainder contributions enter only at smaller pairwise order. Hence
\[
\begin{split}
\kappa_2-\kappa_1
&=
\alpha^2
\left\{
\beta_2-\beta_1
+
\frac{\gamma}{\alpha}
\log\frac{C_{\pi_2}}{C_{\pi_1}}
\right\}
L^{-3}
+
o(L^{-3})\\
&=
\alpha^2(\zeta_2-\zeta_1)L^{-3}
+
o(L^{-3}).
\end{split}
\]
Equation~\eqref{eq:app-oracle-kappa-comparison} therefore gives \(R_{p,s}^{\mathrm{or}}(\pi_2) - R_{p,s}^{\mathrm{or}}(\pi_1) = -4\alpha^2(\zeta_2-\zeta_1)\frac{s}{L^3} + o(s/L^3)\).
The three cases prove Theorem~3.
\end{proof}

\begin{proof}[of Corollary~2]
The horseshoe and horseshoe+ have the same polynomial tail index \(\alpha=1/2\), while \(\gamma_{\rm H}=0, \gamma_{\rm H+}=1\).
The second case of Theorem~3 therefore gives
\[
\begin{split}
R_{p,s}^{\mathrm{or}}(\pi_{\rm H+})
-
R_{p,s}^{\mathrm{or}}(\pi_{\rm H})
&=
4\left(\frac12\right)(1-0)\frac{s}{L^2}
+
o(s/L^2)\\
&=
\frac{2s}{L^2}
+
o(s/L^2),
\end{split}
\]
which proves the result.
\end{proof}

\section{Proofs for regression model}

\subsection{Regression reduction and two-stage interpolation}
\label{app:regression-reduction}

Fix a log-regular prior \(\pi\) and set \(\tau=\tau_\pi^{\mathrm{or}}\) throughout Appendix~B. Write \(\boldsymbol\theta=\sqrt n\,\boldsymbol\beta\), \(\boldsymbol G=\boldsymbol X^T\boldsymbol X/n\), and \(\boldsymbol E=\boldsymbol G-\boldsymbol I_p\). For a fixed \(C_\theta>0\), the regression truths considered throughout Appendix~B satisfy
\begin{equation}
\label{eq:regression-signal-envelope}
\|\boldsymbol\theta^0\|_0\leq s,\qquad \max_{j\in S}\log(1+|\theta_j^0|)\leq C_\theta L,\qquad S=\operatorname{supp}(\boldsymbol\theta^0).
\end{equation}
This envelope is much wider than the least-favourable sequence-model neighbourhood, where the active signals are of order \(L^{1/2}\), and is used only to obtain a uniform lower bound for the regression marginal likelihood. For \(\boldsymbol h=\boldsymbol\theta-\boldsymbol\theta^0\), define the centered score
\[
\boldsymbol W
=
\frac{\boldsymbol X^T\boldsymbol Y}{\sqrt n}
-
\boldsymbol G\boldsymbol\theta^0
=
\frac{\boldsymbol X^T\boldsymbol\varepsilon}{\sqrt n},
\qquad
\boldsymbol W\sim N_p(\boldsymbol0,\boldsymbol G).
\]
Up to a term independent of \(\boldsymbol\theta\), the Gaussian regression log-likelihood ratio is
\[
\ell(\boldsymbol\theta)-\ell(\boldsymbol\theta^0)
=
\boldsymbol h^T\boldsymbol W
-
\frac12\boldsymbol h^T\boldsymbol G\boldsymbol h.
\]
Thus the regression posterior depends on the data only through the Gaussian score experiment \((\boldsymbol W,\boldsymbol G)\). When \(\boldsymbol G=\boldsymbol I_p\), the same expression is the likelihood ratio of the Gaussian sequence experiment with observation \(\boldsymbol\theta^0+\boldsymbol W\).

For a symmetric nonnegative definite matrix \(\boldsymbol H\), let \(\boldsymbol W_H\sim N_p(\boldsymbol0,\boldsymbol H)\) and define the corresponding posterior by
\[
\Pi_H(d\boldsymbol\theta\mid\boldsymbol W_H)
\propto
\exp\left\{
\boldsymbol h^T\boldsymbol W_H
-
\frac12\boldsymbol h^T\boldsymbol H\boldsymbol h
\right\}
\Pi_\tau(d\boldsymbol\theta),
\]
where \(\Pi_\tau\) is the product marginal prior induced by the hierarchy in \((2)\) at the fixed scale \(\tau\). Write
\[
\boldsymbol\mu_H
=
E_H(\boldsymbol\theta\mid\boldsymbol W_H),
\qquad
\boldsymbol b_H
=
\boldsymbol\mu_H-\boldsymbol\theta^0,
\qquad
\boldsymbol C_H
=
\operatorname{Var}_H(\boldsymbol\theta\mid\boldsymbol W_H),
\]
and let
\[
\mathcal R_H(\boldsymbol\theta^0)
=
E_{\boldsymbol\theta^0}
\|\boldsymbol b_H\|_2^2.
\]
Then \(\mathcal R_{\boldsymbol I_p}(\boldsymbol\theta^0)\) is the coordinatewise Gaussian sequence risk, whereas \(\mathcal R_{\boldsymbol G}(\boldsymbol\theta^0)\) is the coefficient risk in the regression experiment. Consequently, \(R_{p,s}(\pi,\tau)=\sup_{\|\boldsymbol\theta^0\|_0\le s}\mathcal R_{\boldsymbol I_p}(\boldsymbol\theta^0)\), whereas the regression risk in Theorem~4 is \(R_{n,p,s}^{X}(\pi,\tau)=\sup_{\boldsymbol\theta^0\text{ satisfying }\eqref{eq:regression-signal-envelope}}\mathcal R_{\boldsymbol G}(\boldsymbol\theta^0)\).

We compare these experiments in two stages. For a fixed truth, let \(S=\operatorname{supp}(\boldsymbol\theta^0)\) and \(N=S^c\), and partition
\[
\boldsymbol E
=
\begin{pmatrix}
\boldsymbol E_{SS} & \boldsymbol E_{SN}\\
\boldsymbol E_{NS} & \boldsymbol E_{NN}
\end{pmatrix}.
\]
Separate the active-block perturbation from the remaining interactions by
\[
\boldsymbol E^{A}
=
\begin{pmatrix}
\boldsymbol E_{SS} & \boldsymbol0\\
\boldsymbol0 & \boldsymbol0
\end{pmatrix},
\qquad
\boldsymbol F
=
\boldsymbol E-\boldsymbol E^{A}
=
\begin{pmatrix}
\boldsymbol0 & \boldsymbol E_{SN}\\
\boldsymbol E_{NS} & \boldsymbol E_{NN}
\end{pmatrix}.
\]
The first interpolation is
\[
\boldsymbol H_t^{A}
=
\boldsymbol I_p+t\boldsymbol E^{A},
\qquad
0\le t\le1,
\]
and connects the orthogonal sequence experiment to the experiment containing only the active--active design dependence. The second interpolation is
\[
\boldsymbol H_\lambda^{C}
=
\boldsymbol I_p+\boldsymbol E^{A}+\lambda\boldsymbol F
=
(1-\lambda)(\boldsymbol I_p+\boldsymbol E^{A})
+
\lambda\boldsymbol G,
\qquad
0\le\lambda\le1,
\]
and adds the active--null and null--null interactions. Denote the corresponding risks by
\[
\mathcal R_t^{A}
=
\mathcal R_{\boldsymbol H_t^{A}}(\boldsymbol\theta^0),
\qquad
\widetilde{\mathcal R}_\lambda
=
\mathcal R_{\boldsymbol H_\lambda^{C}}(\boldsymbol\theta^0).
\]
Since \(\mathcal R_0^{A}=\mathcal R_{\boldsymbol I_p}\), \(\widetilde{\mathcal R}_0=\mathcal R_1^{A}\), and \(\widetilde{\mathcal R}_1=\mathcal R_{\boldsymbol G}\),
\[
\mathcal R_{\boldsymbol G}(\boldsymbol\theta^0)
-
\mathcal R_{\boldsymbol I_p}(\boldsymbol\theta^0)
=
\{\mathcal R_1^{A}-\mathcal R_0^{A}\}
+
\{\widetilde{\mathcal R}_1-\widetilde{\mathcal R}_0\}.
\]
The two terms are treated separately below. The first is controlled by cumulative dependence within \(S\), whereas the second starts at second order because the cross--null perturbation has no active--active block.

The following differentiation identity will be used for both interpolations. For a path \(\boldsymbol H_r=\boldsymbol H_0+r\boldsymbol A\), let \(\boldsymbol b_r\), \(\boldsymbol C_r\), and \(\boldsymbol T_r\) denote the posterior mean error, covariance matrix, and third central cumulant tensor. For \(\boldsymbol v\in\mathbb R^p\), write \(\boldsymbol T_r[\boldsymbol v]\) for the matrix with entries
\[
\{\boldsymbol T_r[\boldsymbol v]\}_{jk}
=
\sum_{\ell=1}^p
T_{r,jk\ell}v_\ell.
\]

\begin{lemma}[Gaussian interpolation identity]
\label{lem:gaussian-interpolation}
Suppose that \(\boldsymbol H_r\) is positive definite on the interior of the interpolation interval. Then
\[
\frac{d}{dr}\mathcal R_{\boldsymbol H_r}(\boldsymbol\theta^0)
=
E_{\boldsymbol\theta^0}
\operatorname{tr}(\boldsymbol A\boldsymbol C_r^2)
-
2E_{\boldsymbol\theta^0}
\left(
\boldsymbol b_r^T
\boldsymbol C_r
\boldsymbol A
\boldsymbol b_r
\right).
\]
Moreover, if
\[
\mathscr D_r f
=
\frac{\partial f}{\partial r}
+
\frac12
\sum_{j,k=1}^p
A_{jk}
\frac{\partial^2f}
{\partial w_j\partial w_k},
\]
then
\[
\mathscr D_r\boldsymbol b_r
=
-\boldsymbol C_r\boldsymbol A\boldsymbol b_r,
\qquad
\mathscr D_r\boldsymbol C_r
=
-\boldsymbol C_r\boldsymbol A\boldsymbol C_r
-
\boldsymbol T_r[\boldsymbol A\boldsymbol b_r].
\]
\end{lemma}

\begin{proof}
For fixed \(\boldsymbol w\), differentiation of the tilted posterior gives
\[
\frac{\partial\boldsymbol b_r}{\partial r}
=
-\boldsymbol C_r\boldsymbol A\boldsymbol b_r
-
\frac12\boldsymbol T_r[\boldsymbol A],
\]
where \(\boldsymbol T_r[\boldsymbol A]\) denotes the vector with \(j\)th component \(\sum_{k,\ell}T_{r,jk\ell}A_{k\ell}\). Posterior differentiation with respect to the Gaussian score gives
\[
\frac{\partial\boldsymbol b_r}{\partial\boldsymbol w}
=
\boldsymbol C_r,
\qquad
\frac12
\sum_{k,\ell}
A_{k\ell}
\frac{\partial^2\boldsymbol b_r}
{\partial w_k\partial w_\ell}
=
\frac12\boldsymbol T_r[\boldsymbol A].
\]
The two third-cumulant terms cancel, yielding \(\mathscr D_r\boldsymbol b_r=-\boldsymbol C_r\boldsymbol A\boldsymbol b_r\). The same calculation applied to the posterior covariance gives
\[
\mathscr D_r\boldsymbol C_r
=
-\boldsymbol C_r\boldsymbol A\boldsymbol C_r
-
\boldsymbol T_r[\boldsymbol A\boldsymbol b_r].
\]
Finally, Gaussian differentiation implies \(dE f(\boldsymbol W_r,r)/dr=E\{\mathscr D_r f(\boldsymbol W_r,r)\}\). Applying this identity to \(f=\|\boldsymbol b_r\|_2^2\), and using \(\partial\boldsymbol b_r/\partial\boldsymbol w=\boldsymbol C_r\), gives
\[
\mathscr D_r\|\boldsymbol b_r\|_2^2
=
\operatorname{tr}(\boldsymbol A\boldsymbol C_r^2)
-
2\boldsymbol b_r^T
\boldsymbol C_r
\boldsymbol A
\boldsymbol b_r,
\]
which proves the result.
\end{proof}

When \(p>n\), \(\boldsymbol G\) may be singular. The second interpolation is positive definite for \(0\le\lambda<1\) because \(\boldsymbol H_\lambda^{C}=(1-\lambda)(\boldsymbol I_p+\boldsymbol E^{A})+\lambda\boldsymbol G\), and the endpoint \(\lambda=1\) is obtained by continuity. Hence Lemma~\ref{lem:gaussian-interpolation} applies along both interpolation paths under the design conditions used below.

\subsection{Sparse geometry and scalar response bounds}
\label{app:regression-preliminaries}

We collect the design and one-dimensional posterior bounds used in the interpolation arguments. Write
\[
\delta_n=\Delta_{n,2m_n},
\qquad
\gamma_n
=
\sup_{|S|\leq s}
\max_{i\in S}
\sum_{\substack{j\in S\\j\neq i}}
|G_{ij}|,
\]
and recall that \(\boldsymbol E=\boldsymbol G-\boldsymbol I_p\). Throughout Appendix~B we assume
\[
0<\alpha<1,
\qquad
sL=o(m_n),
\qquad
\delta_nL\to0,
\qquad
\gamma_nL\to0.
\]
The restricted-isometry quantity \(\delta_n\) controls sparse quadratic geometry, whereas \(\gamma_n\) controls cumulative dependence within the true active set.

\begin{lemma}[Sparse design geometry]
\label{lem:sparse-design-geometry}
Suppose \(\Delta_{n,2m}\leq\delta\). If \(\boldsymbol x\) and \(\boldsymbol y\) have disjoint supports of size at most \(m\), then
\[
|\boldsymbol x^T\boldsymbol E\boldsymbol y|
\leq
\delta\|\boldsymbol x\|_2\|\boldsymbol y\|_2.
\]
For arbitrary \(\boldsymbol x,\boldsymbol y\in\mathbb R^p\),
\[
|\boldsymbol x^T\boldsymbol E\boldsymbol y|
\lesssim
\delta
\left(
\|\boldsymbol x\|_2+\frac{\|\boldsymbol x\|_1}{\sqrt m}
\right)
\left(
\|\boldsymbol y\|_2+\frac{\|\boldsymbol y\|_1}{\sqrt m}
\right).
\]
Moreover, if \(S\subset\{1,\ldots,p\}\), \(|S|\leq m\), and \(N=S^c\), then
\[
\|\boldsymbol E_{SS}\|_{\rm op}\leq\delta,
\qquad
\|\boldsymbol E_{NS}\|_{\rm F}^2
\lesssim
\delta^2\frac{|S|p}{m},
\qquad
\|\boldsymbol E_{NN}\|_{\rm F}^2
\lesssim
\delta^2\frac{p^2}{m}.
\]
If, in addition, \(|S|\leq s\), then
\[
\|\boldsymbol E_{SS}\|_{\rm op}
\leq
\|\boldsymbol E_{SS}\|_{\infty\to\infty}
\leq
\gamma_n.
\]
\end{lemma}

\begin{proof}
The first inequality follows by polarization, since the union of the two supports has size at most \(2m\). For the second, partition the coordinates of \(\boldsymbol x\) and \(\boldsymbol y\) into blocks of size at most \(m\), ordered by decreasing absolute magnitude. Restricted orthogonality applied blockwise, together with
\[
\sum_{k\geq2}\|\boldsymbol x_{I_k}\|_2
\leq
\|\boldsymbol x\|_2+
\frac{\|\boldsymbol x\|_1}{\sqrt m},
\]
and the analogous inequality for \(\boldsymbol y\), gives the result.

For the Frobenius bounds, partition \(N\) into blocks \(I_k\) of size at most \(m\). Since \(|S|\leq m\), restricted orthogonality gives \(\|\boldsymbol E_{I_kS}\|_{\rm op}\leq\delta\), and hence \(\|\boldsymbol E_{I_kS}\|_{\rm F}^2 \leq |S|\delta^2.\) Summing over \(O(p/m)\) blocks yields \(\|\boldsymbol E_{NS}\|_{\rm F}^2 \lesssim \delta^2\frac{|S|p}{m}.\) Similarly, \(\|\boldsymbol E_{I_kI_\ell}\|_{\rm F}^2\leq m\delta^2\) for every pair of blocks, including \(k=\ell\), and summing over \(O(p^2/m^2)\) pairs gives \(\|\boldsymbol E_{NN}\|_{\rm F}^2 \lesssim \delta^2\frac{p^2}{m}.\) The final assertion follows from symmetry of \(\boldsymbol E_{SS}\) and the definition of \(\gamma_n\).
\end{proof}

For the scalar posterior, let \(w(y,\tau)\) be the nonshrinkage weight in \eqref{eq:app-postmean} and define
\[
q_\pi(\tau)
=
E_0\{w(Y,\tau)\}.
\]
The quantity \(pq_\pi(\tau)\) measures aggregate null-scale activity. At the sequence-model oracle it is smaller than the sparse-design complexity used below.

\begin{lemma}[Oracle null activity]
\label{lem:oracle-null-activity}
Suppose Assumption~1 holds with \(0<\alpha<1\), and let \(\tau=\tau_\pi^{\mathrm{or}}\). Then
\[
a_\tau^2
=
2L+O(\log L),
\qquad
pR_0(\pi,\tau)
=
O(s),
\]
and
\[
pq_\pi(\tau)
=
O(sL^\alpha)
=
o(sL).
\]
\end{lemma}

\begin{proof}
The threshold relation follows from Theorem~3.1, while Proposition~\ref{prop:null-risk} gives \((p-s)R_0(\pi,\tau)=O(s).\) Since \(s/p\to0\), this also gives \(pR_0(\pi,\tau)=O(s)\).

To control \(q_\pi(\tau)\), split the integral defining \(E_0w(Y,\tau)\) into the deep-central, moderate-central, transition and upper-tail regions as in the proof of Proposition~\ref{prop:null-risk}. On the moderate-central region, the central--tail comparison gives
\[
w(y,\tau)\phi(y)
\lesssim
\tau^{2\alpha}
|y|^{-2\alpha-1}
H(y,\tau)^\gamma,
\qquad
1\leq |y|\leq a_\tau.
\]
Since \(H(y,\tau)\asymp L\) throughout this range and \(\int_1^{a_\tau}y^{-2\alpha-1}\,dy \lesssim1,\) its contribution is \(O(\tau^{2\alpha}L^\gamma)\). The deep-central leakage has the same or smaller order by the bounds used in Proposition~\ref{prop:null-risk}, whereas the transition and upper-tail regions contribute only \(O\{\phi(a_\tau)/a_\tau\}\). Hence \(q_\pi(\tau) \lesssim \tau^{2\alpha}L^\gamma.\) At the oracle calibration, \(\tau^{2\alpha} \asymp \frac{s}{p}L^{\alpha-\gamma},\) so \(pq_\pi(\tau) = O(sL^\alpha).\) Since \(0<\alpha<1\), \(sL^\alpha=o(sL)\).
\end{proof}

We next record the scalar posterior responses needed in the interpolation arguments. Put
\[
m_\tau(y)
=
\widehat\theta_{\pi,\tau}(y),
\qquad
b_\vartheta(y)
=
m_\tau(y)-\vartheta,
\]
and define
\[
V_\tau(y)
=
m_\tau'(y),
\qquad
T_\tau(y)
=
m_\tau''(y),
\qquad
U_\tau(y)
=
m_\tau'''(y).
\]
Under the Gaussian sequence experiment these are, respectively, the posterior variance, third cumulant and fourth cumulant of \(\theta\). Let \(a=a_\tau\) and define
\[
\Lambda_a(\vartheta)
=
1+
a\{\phi(a-\vartheta)+\phi(a+\vartheta)\}.
\]

\begin{lemma}[Scalar response bounds]
\label{lem:scalar-response}
Suppose Assumption~1 holds with \(0<\alpha<1\), and let \(\tau=\tau_\pi^{\mathrm{or}}\). Then
\[
v_\tau
:=
\sup_{y\in\mathbb R}V_\tau(y)
\lesssim
a^2.
\]
Uniformly over \(\vartheta\in\mathbb R\),
\[
E_\vartheta\{b_\vartheta(Y)^2\}
\lesssim
a^2,
\]
\[
E_\vartheta
\left[
V_\tau(Y)^2+
|b_\vartheta(Y)T_\tau(Y)|
\right]
\lesssim
a^2\Lambda_a(\vartheta),
\]
and
\[
E_\vartheta
\left[
T_\tau(Y)^2+
b_\vartheta(Y)^2|U_\tau(Y)|
\right]
\lesssim
a^4\Lambda_a(\vartheta).
\]
Moreover, for some fixed \(c>0\),
\[
\sup_{\vartheta\in\mathbb R}
E_\vartheta
\int
\exp\left\{
\frac{c}{a^2}
(\theta-\vartheta)^2
\right\}
\Pi_\tau(d\theta\mid Y)
\lesssim
1.
\]

At a null coordinate, writing
\[
r_0=R_0(\pi,\tau),
\qquad
\overline q_\tau
=
E_0\{V_\tau(Y)\},
\]
we have
\[
E_0\{b_0(Y)^2\}=r_0,
\qquad
p\overline q_\tau=O(sL),
\]
\[
E_0
\left[
V_\tau(Y)^2+
|b_0(Y)T_\tau(Y)|
\right]
\lesssim
a^2r_0,
\]
and
\[
E_0
\left[
T_\tau(Y)^2+
b_0(Y)^2|U_\tau(Y)|
\right]
\lesssim
a^4r_0.
\]
\end{lemma}

\begin{proof}
Let \(\kappa=(1+\tau^2u)^{-1},\) and let \(c_j(y)\) denote the \(j\)th posterior cumulant of \(\kappa\) given \(Y=y\). Since the posterior distribution of \(\kappa\) is an exponential tilt in \(y^2\),
\[
w'(y)
=
yc_2(y),
\qquad
c_2'(y)
=
-yc_3(y),
\qquad
c_3'(y)
=
-yc_4(y).
\]
Using \(m_\tau(y)=yw(y,\tau)\) gives \(V_\tau(y) = w(y,\tau)+y^2c_2(y),\) \(T_\tau(y) = 3yc_2(y)-y^3c_3(y),\) and \(U_\tau(y) = 3c_2(y)-6y^2c_3(y)+y^4c_4(y).\)

We first obtain pointwise envelopes in the three regions around the intrinsic threshold. Fix \(K>0\) sufficiently large, put \(M_L=K\log L, \qquad x=\frac{y^2-a^2}{2},\) and first consider \(|x|\leq M_L\). Extend the central--tail decomposition \(M_j=C_j+T_j\) in Section~\ref{app:uniform-expansions} to the fixed indices \(j=0,\ldots,4\). The same change of variables as in Proposition~\ref{prop:uniform-tail}, without retaining the higher-order coefficients, gives
\[
\frac{C_j(y,\tau)}{B(y)}
=
1+o(L^{-m}),
\qquad
\frac{T_j(y,\tau)}{T_0(y,\tau)}
=
O(a^{-2j}),
\qquad
j=1,\ldots,4,
\]
for every fixed \(m>0\), uniformly on this region. Proposition~\ref{prop:transition-profile} also gives
\[
Q(y,\tau)
=
e^x
\exp\left\{
O\left(\frac{1+|x|}{L}\right)
\right\}.
\]
It follows by taking the corresponding polynomial combinations of the first four posterior moments of \(\kappa\) that \(c_2(y) \lesssim e^{-|x|}+a^{-4},\)
\[
|c_3(y)|
\lesssim
e^{-|x|}+a^{-6},
\qquad
|c_4(y)|
\lesssim
e^{-|x|}+a^{-8}.
\]
Consequently,
\begin{equation}
\label{eq:scalar-transition-envelope}
\begin{split}
V_\tau(y)
&\lesssim
1+a^2e^{-|x|},\\
|T_\tau(y)|
&\lesssim
a^3e^{-|x|}+a^{-3},\\
|U_\tau(y)|
&\lesssim
a^4e^{-|x|}+a^{-4}.
\end{split}
\end{equation}

Below the transition, the same central--tail decomposition gives, uniformly for \(\eta a^2 \leq y^2 \leq a^2-2M_L,\) with any fixed \(0<\eta<1\),
\[
w(y,\tau)
+
|c_2(y)|
+
|c_3(y)|
+
|c_4(y)|
\lesssim
e^x+o(L^{-m}).
\]
On \(y^2\leq\eta a^2\), the deep-central bounds used in Proposition~\ref{prop:null-risk} make the same quantities exponentially small in \(L\), up to fixed powers of \(L\).

Above the transition, the regular-variation change of variables used in Proposition~\ref{prop:uniform-tail}, now applied only for coarse upper bounds, gives \(1-w(y,\tau) \lesssim y^{-2}, \qquad c_2(y) \lesssim y^{-4},\) \(|c_3(y)| \lesssim |y|^{-6}, \qquad |c_4(y)| \lesssim y^{-8}\) uniformly for \(y^2\geq a^2+2M_L\). Hence
\begin{equation}
\label{eq:scalar-upper-envelope}
V_\tau(y)\lesssim1,
\qquad
|T_\tau(y)|\lesssim |y|^{-3},
\qquad
|U_\tau(y)|\lesssim y^{-4}
\end{equation}
on the upper region. Together with \eqref{eq:scalar-transition-envelope}, these bounds prove \(\sup_yV_\tau(y)\lesssim a^2.\)

We next integrate the pointwise envelopes. On the positive transition region, \(dy=dx/y_x\) with \(y_x=(a^2+2x)^{1/2}\asymp a\). For \(r=1,2\),
\[
\int_{|x|\leq M_L}
e^{-r|x|}
\phi(y_x-\vartheta)\frac{dx}{y_x}
\lesssim
\frac{1}{a}\phi(a-\vartheta)
+
a^{-B}
\]
for any prescribed fixed \(B>0\), after increasing \(K\) if necessary. To see this, when \(|\vartheta-a|\leq a/2\), the Gaussian density ratio across the transition window is absorbed by \(e^{-r|x|}\); when \(|\vartheta-a|>a/2\), the entire transition contribution is exponentially small in \(a^2\) after multiplication by any fixed polynomial in \(a\) and \(|Y-\vartheta|\). The same argument applies around \(-a\).

Since \(b_\vartheta(y) = (y-\vartheta)-y\{1-w(y,\tau)\},\) we have \(|b_\vartheta(y)| \leq |y-\vartheta|+|y|.\) Standard Gaussian polynomial-tail bounds combined with \eqref{eq:scalar-transition-envelope} therefore give
\[
E_\vartheta
\left[
\{V_\tau(Y)^2+
|b_\vartheta(Y)T_\tau(Y)|\}
1_{\{|Y^2-a^2|\leq2M_L\}}
\right]
\lesssim
a^2+
a^3\{\phi(a-\vartheta)+\phi(a+\vartheta)\},
\]
and
\[
E_\vartheta
\left[
\{T_\tau(Y)^2+
b_\vartheta(Y)^2|U_\tau(Y)|\}
1_{\{|Y^2-a^2|\leq2M_L\}}
\right]
\lesssim
a^4+
a^5\{\phi(a-\vartheta)+\phi(a+\vartheta)\}.
\]
The central region contributes no more than the first terms on the right-hand sides by the lower-region envelopes above. On the upper region, \eqref{eq:scalar-upper-envelope} and \(|b_\vartheta(Y)| \leq |Y-\vartheta|+\frac{C}{|Y|}\) reduce the required expectations to bounded Gaussian moments. Hence
\[
E_\vartheta
\left[
V_\tau(Y)^2+
|b_\vartheta(Y)T_\tau(Y)|
\right]
\lesssim
a^2\Lambda_a(\vartheta),
\]
and
\[
E_\vartheta
\left[
T_\tau(Y)^2+
b_\vartheta(Y)^2|U_\tau(Y)|
\right]
\lesssim
a^4\Lambda_a(\vartheta).
\]
Finally,
\[
E_\vartheta b_\vartheta(Y)^2
=
R_\vartheta(\pi,\tau)
\leq
R_\star(\pi,\tau)
\lesssim
a^2
\]
by Proposition~\ref{prop:active-risk}.

We next prove the exponential-moment bound. Conditional on \(u\) and \(Y\), put \(A = \frac{\tau^2u}{1+\tau^2u} = 1-\kappa.\) Then \(\theta\mid u,Y \sim N(AY,A).\) For \(t=c/a^2\), with \(c>0\) sufficiently small,
\[
E\left[
e^{t(\theta-\vartheta)^2}
\,\middle|\,
u,Y
\right]
=
(1-2tA)^{-1/2}
\exp\left\{
\frac{t(AY-\vartheta)^2}{1-2tA}
\right\}.
\]
Writing \(Z=Y-\vartheta\), \(AY-\vartheta = Z-\kappa Y,\) so
\[
E\left[
e^{t(\theta-\vartheta)^2}
\,\middle|\,
u,Y
\right]
\lesssim
\exp\{CtZ^2+Ct\kappa^2Y^2\}.
\]
Because \(\kappa^2\leq\kappa\),
\[
E\left[
e^{Ct\kappa^2Y^2}
\,\middle|\,
Y=y
\right]
\leq
\frac{
M_0\{y(1-2Ct)^{1/2},\tau\}
}{
M_0(y,\tau)
}.
\]
For \(|y|\leq2a\), the ratio is bounded directly by \(e^{Cta^2}\). For \(|y|>2a\), the regular-variation change of variables defining the tail component of \(M_0\), together with the Potter bounds implied by Assumption~1, gives
\[
\sup_{|y|>2a}
\frac{
M_0\{y(1-2Ct)^{1/2},\tau\}
}{
M_0(y,\tau)
}
\lesssim1
\]
because \(t=O(a^{-2})\). Hence
\[
E\left[
e^{t(\theta-\vartheta)^2}
\,\middle|\,
Y
\right]
\lesssim
e^{CtZ^2}.
\]
Since \(Z\sim N(0,1)\),
\[
\sup_{\vartheta\in\mathbb R}
E_\vartheta
\int
\exp\left\{
\frac{c}{a^2}
(\theta-\vartheta)^2
\right\}
\Pi_\tau(d\theta\mid Y)
\lesssim
1.
\]

It remains to sharpen the null-coordinate bounds. Symmetry gives \(E_0b_0(Y)^2 = r_0.\) Moreover, Stein's identity and \(m_\tau(y)=yw(y,\tau)\) give \(\overline q_\tau = E_0V_\tau(Y) = E_0\{Y^2w(Y,\tau)\}.\) On \(1\leq |y|\leq a\),
\[
w(y,\tau)\phi(y)
\lesssim
\tau^{2\alpha}
|y|^{-2\alpha-1}
H(y,\tau)^\gamma.
\]
Since \(H(y,\tau)\asymp L\) and \(0<\alpha<1\),
\[
\begin{split}
\overline q_\tau
&\lesssim
\tau^{2\alpha}L^\gamma
\int_1^a
y^{1-2\alpha}\,dy
+
a\phi(a)\\
&\lesssim
\tau^{2\alpha}L^\gamma
a^{2-2\alpha}
+
a\phi(a).
\end{split}
\]
At the oracle scale,
\[
\tau^{2\alpha}L^\gamma
\asymp
\frac{s}{p}L^\alpha,
\qquad
a^2\asymp L,
\qquad
a\phi(a)=O(s/p),
\]
and therefore \(p\overline q_\tau = O(sL).\)

Finally, use the exact transition coordinate
\[
x=\frac{Y^2-a^2}{2},
\qquad
\phi(Y)=\phi(a)e^{-x},
\qquad
dY=\frac{dx}{Y}.
\]
The lower, transition and upper envelopes established above give
\[
E_0
\left[
V_\tau(Y)^2+
|b_0(Y)T_\tau(Y)|
\right]
\lesssim
a^3\phi(a),
\]
and
\[
E_0
\left[
T_\tau(Y)^2+
b_0(Y)^2|U_\tau(Y)|
\right]
\lesssim
a^5\phi(a).
\]
Proposition~\ref{prop:null-risk} gives \(r_0 = 2a\phi(a)\{1+o(1)\},\) so the preceding bounds are respectively \(O(a^2r_0)\) and \(O(a^4r_0)\).
\end{proof}

The factor \(\Lambda_a(\vartheta)\) records the probability that a signal places the observation in the scalar transition region. In particular, Proposition~\ref{prop:active-risk} gives
\[
a\phi(t_a)
=
2\{1+o(1)\},
\]
so
\[
\Lambda_a(\vartheta)=O(1)
\]
uniformly whenever
\[
\min_{\epsilon\in\{-1,1\}}
\left|
\vartheta-\epsilon(a-t_a)
\right|
=
O(a^{-1}).
\]
Thus the response orders sharpen to \(O(a^2)\) and \(O(a^4)\) in the least-favourable neighbourhood, while the displayed \(\Lambda_a\)-weighted bounds remain valid uniformly over arbitrary active signals.

\subsection{Posterior localization and conditional stability}
\label{app:posterior-localization}

Set \(\tau=\tau_\pi^{\mathrm{or}}\) and \(a=a_\tau\), and retain the regression truth class in \eqref{eq:regression-signal-envelope}.

Along the cross--null interpolation, write
\[
\boldsymbol H_\lambda^C
=
\boldsymbol I_p+\boldsymbol E^A+\lambda\boldsymbol F,
\qquad
\boldsymbol W_\lambda
\sim
N_p(\boldsymbol0,\boldsymbol H_\lambda^C),
\qquad
0\leq\lambda\leq1,
\]
and denote the corresponding posterior expectation by \(\langle\cdot\rangle_\lambda\). For \(u=(u_1,\ldots,u_p)\), let
\[
t_j=\tau^2u_j,
\qquad
A_j=\frac{t_j}{1+t_j},
\qquad
\boldsymbol\Lambda_u
=
\operatorname{diag}(A_1,\ldots,A_p).
\]
For \(N=S^c\), define
\[
D_N
=
\sum_{j\in N}A_j,
\qquad
J_{N,\lambda}
=
\frac{1}{a^2}
\sum_{j\in N}
A_jW_{\lambda,j}^2,
\qquad
H_{N,\lambda}
=
D_N+J_{N,\lambda}.
\]
The quantity \(D_N\) measures the effective number of null scales, while \(J_{N,\lambda}\) additionally records null coordinates carrying unusually large scores.

For comparison, let \(\Pi_{\lambda,A}\) denote the posterior obtained from the same score \(\boldsymbol W_\lambda\) after replacing \(\boldsymbol H_\lambda^C\) in the likelihood by \(\boldsymbol I_p+\boldsymbol E^A\), and write its expectation as \(\langle\cdot\rangle_{\lambda,A}\). Conditional on the score, its null coordinates are independent scalar sequence posteriors.

\begin{lemma}[Prior null-scale activity]
\label{lem:prior-null-scale}
Let
\[
\mu_\tau
=
E_\pi\left(
\frac{\tau^2u}{1+\tau^2u}
\right).
\]
Under Assumption~1 with \(0<\alpha<1\),
\[
\mu_\tau
\lesssim
\tau^{2\alpha}
\{\log(1/\tau)\}^{\gamma}.
\]
At \(\tau=\tau_\pi^{\mathrm{or}}\),
\[
p\mu_\tau
=
O(sL^\alpha).
\]
Moreover, for \(x\geq Cp\mu_\tau\),
\[
\Pi(D_N>x)
\leq
\exp\left[
-cx
\log\left\{
\frac{x}{Cp\mu_\tau}
\right\}
\right].
\]
For every fixed integer \(r\geq0\), the same bound holds with \(1\{D_N>x\}\) replaced by \((1+D_N)^r1\{D_N>x\}\), up to a change in \(c\).
\end{lemma}

\begin{proof}
With \(v=\tau^2u\), \(\mu_\tau = \int_0^\infty \frac{\tau^2u}{1+\tau^2u} \pi(u)\,du.\) Split the integral at \(u=\tau^{-2}\). On the upper part, Assumption~1 and regular variation give
\[
\int_{\tau^{-2}}^\infty
\pi(u)\,du
\lesssim
\tau^{2\alpha}
\{\log(1/\tau)\}^{\gamma}.
\]
On the lower part,
\[
\tau^2
\int_0^{\tau^{-2}}
u\pi(u)\,du
\lesssim
\tau^{2\alpha}
\{\log(1/\tau)\}^{\gamma},
\]
because \(0<\alpha<1\). Thus \(\mu_\tau \lesssim \tau^{2\alpha} \{\log(1/\tau)\}^{\gamma}.\) At the oracle scale,
\[
\tau^{2\alpha}
\asymp
\frac{s}{p}L^{\alpha-\gamma},
\qquad
\log(1/\tau)\asymp L,
\]
which proves \(p\mu_\tau=O(sL^\alpha)\).

Since \(0\leq A_j\leq1\), \(e^{tA_j} \leq 1+A_j(e^t-1),\) and independence of the prior scales gives \(E_\Pi e^{tD_N} \leq \exp\{p\mu_\tau(e^t-1)\}.\) Taking \(t= \log\left( \frac{x}{Cp\mu_\tau} \right)\) for a sufficiently large constant \(C\) gives the stated Chernoff bound. The polynomially weighted version follows by integrating this tail bound, or equivalently by differentiating the exponential moment a fixed number of times.
\end{proof}

We next obtain a uniform lower bound for the marginal likelihood. Let
\[
Z_\lambda(\boldsymbol W_\lambda)
=
\int
\exp\left\{
\boldsymbol h^T\boldsymbol W_\lambda
-
\frac12
\boldsymbol h^T
\boldsymbol H_\lambda^C
\boldsymbol h
\right\}
\Pi_\tau(d\boldsymbol\theta),
\qquad
\boldsymbol h=\boldsymbol\theta-\boldsymbol\theta^0.
\]

\begin{lemma}[Regression evidence bound]
\label{lem:regression-evidence}
Suppose \eqref{eq:regression-signal-envelope} holds. There exist constants \(C_0,c_0>0\) such that, uniformly in \(0\leq\lambda\leq1\), for every \(x\geq0\),
\[
P_{\boldsymbol\theta^0}
\left\{
Z_\lambda(\boldsymbol W_\lambda)
<
\exp(-C_0sL-x)
\right\}
\leq
\exp\left(-c_0\frac{x^2}{s}\right).
\]
\end{lemma}

\begin{proof}
We construct a restricted prior distribution carrying mass at least \(e^{-CsL}\). For \(j\in S\), let \(\mathcal A_j = \{ |\theta_j-\theta_j^0|\leq1 \}.\) Assumption~1 and \eqref{eq:regression-signal-envelope} imply \(\inf_{j\in S} \Pi_\tau(\mathcal A_j) \geq e^{-CL}.\) Indeed, if \(|\theta_j^0|\) is bounded, this follows from the proper central bulk of the prior. Otherwise, restricting the scale integral to \(\tau^2u \asymp 1+(\theta_j^0)^2\) and using the log-regular tail gives a lower bound whose logarithm is bounded below by \(-CL\) under \eqref{eq:regression-signal-envelope}. Hence
\[
\Pi_\tau
\left(
\bigcap_{j\in S}\mathcal A_j
\right)
\geq
e^{-CsL}.
\]

For \(j\in N\), recall \(t_j=\tau^2u_j\). Regular variation and the oracle calibration give
\[
P_\pi(t_j>L)
\lesssim
\frac{s}{p},
\qquad
E_\pi\{t_j1(t_j\leq L)\}
\lesssim
\frac{sL}{p}.
\]
Consequently, for a sufficiently large fixed \(C>0\),
\[
\mathcal U_N
=
\left\{
\max_{j\in N}t_j\leq L,\quad
\sum_{j\in N}t_j\leq CsL
\right\}
\]
has prior probability at least \(e^{-Cs}\).

Let \(\nu\) be the prior normalized to
\[
\mathcal K
=
\left(
\bigcap_{j\in S}\mathcal A_j
\right)
\cap
\mathcal U_N.
\]
Then \(\Pi_\tau(\mathcal K) \geq e^{-CsL}.\) Under \(\nu\), the null coordinates remain conditionally centred and
\[
E_\nu\|\boldsymbol h_S\|_2^2
\lesssim s,
\qquad
E_\nu\|\boldsymbol h_N\|_2^2
\lesssim sL.
\]
The cross terms have zero expectation by symmetry, while \(\|\boldsymbol E_{SS}\|_{\rm op}\leq\gamma_n=o(1)\), so
\[
E_\nu
\left(
\boldsymbol h^T
\boldsymbol H_\lambda^C
\boldsymbol h
\right)
\lesssim
sL.
\]

Jensen's inequality therefore gives
\[
\log Z_\lambda(\boldsymbol W_\lambda)
\geq
-C_0sL+
\overline{\boldsymbol h}^{\,T}
\boldsymbol W_\lambda,
\qquad
\overline{\boldsymbol h}
=
E_\nu\boldsymbol h.
\]
The null entries of \(\overline{\boldsymbol h}\) vanish and \(\|\overline{\boldsymbol h}_S\|_2^2\lesssim s\), whence
\[
\operatorname{Var}
\left(
\overline{\boldsymbol h}^{\,T}
\boldsymbol W_\lambda
\right)
\lesssim
s
\]
uniformly in \(\lambda\). The Gaussian tail bound now yields
\[
P_{\boldsymbol\theta^0}
\left\{
Z_\lambda(\boldsymbol W_\lambda)
<
e^{-C_0sL-x}
\right\}
\leq
\exp\left(-c_0\frac{x^2}{s}\right),
\]
which proves the result.
\end{proof}

The preceding lemma allows us to localize the number of active null scales directly under the full posterior, without comparing the full and reference posteriors on dense scale configurations.

\begin{lemma}[Full-posterior scale localization]
\label{lem:full-scale-localization}
Fix a sufficiently small constant \(\eta>0\). For every fixed integer \(r\geq0\),
\[
E_{\boldsymbol\theta^0}
\left\langle
(1+D_N)^r
1\{D_N>\eta m_n\}
\right\rangle_\lambda
\leq
\exp\left[
-cm_n
\log\left\{
\frac{m_n}{sL^\alpha}
\right\}
\right]
\]
uniformly in \(0\leq\lambda\leq1\).
\end{lemma}

\begin{proof}
Put \(A_n = m_n \log\left\{ \frac{m_n}{sL^\alpha} \right\}.\) Since \(sL=o(m_n)\) and \(0<\alpha<1\), \(A_n/(sL)\longrightarrow\infty.\)

Fix a sufficiently small constant \(\kappa>0\) and define
\[
\mathcal E_n
=
\left\{
Z_\lambda
\geq
\exp(-C_0sL-\kappa A_n)
\right\}.
\]
For fixed \(\boldsymbol\theta\),
\[
E_{\boldsymbol\theta^0}
\exp\left\{
\boldsymbol h^T\boldsymbol W_\lambda
-\frac12
\boldsymbol h^T
\boldsymbol H_\lambda^C
\boldsymbol h
\right\}
=
1.
\]
Hence Fubini's theorem gives
\[
\begin{split}
&E_{\boldsymbol\theta^0}
\left[
\left\langle
(1+D_N)^r
1\{D_N>\eta m_n\}
\right\rangle_\lambda
1_{\mathcal E_n}
\right]
\\
&\qquad\leq
\exp(C_0sL+\kappa A_n)
E_\Pi
\left[
(1+D_N)^r
1\{D_N>\eta m_n\}
\right].
\end{split}
\]
Lemma~\ref{lem:prior-null-scale} bounds the prior expectation by \(\exp(-cA_n).\) Choosing \(\kappa<c/2\) and using \(sL=o(A_n)\) therefore gives
\[
E_{\boldsymbol\theta^0}
\left[
\left\langle
(1+D_N)^r
1\{D_N>\eta m_n\}
\right\rangle_\lambda
1_{\mathcal E_n}
\right]
\leq
e^{-c_1A_n}.
\]

On \(\mathcal E_n^c\), the regression evidence bound gives
\[
P_{\boldsymbol\theta^0}(\mathcal E_n^c)
\leq
\exp\left(
-c_2\frac{A_n^2}{s}
\right).
\]
Since \(D_N\leq p\),
\[
E_{\boldsymbol\theta^0}
\left[
\left\langle
(1+D_N)^r
1\{D_N>\eta m_n\}
\right\rangle_\lambda
1_{\mathcal E_n^c}
\right]
\leq
(1+p)^r
\exp\left(
-c_2\frac{A_n^2}{s}
\right).
\]
Write \(h_n=m_n/(sL)\to\infty\) and \(q_n=\log\{m_n/(sL^\alpha)\}=\log h_n+(1-\alpha)\log L\), so \(A_n=sLh_nq_n\). Since \(q_n\to\infty\), \(A_n+\log p=o(A_n^2/s)\), using \(\log p=L+\log s\) and \(\log s\leq s\). Hence the last quantity is \(o(e^{-c_3A_n})\). Combining the two parts proves
\[
E_{\boldsymbol\theta^0}
\left\langle
(1+D_N)^r
1\{D_N>\eta m_n\}
\right\rangle_\lambda
\leq
e^{-cA_n}.
\]
\end{proof}

We now compare the full and reference posteriors only on the localized scale region
\[
\mathcal G_N
=
\{D_N\leq\eta m_n\}.
\]

\begin{lemma}[Cross--null Gaussian normalizer]
\label{lem:cross-null-normalizer}
Conditional on \(u\) and \(\boldsymbol W_\lambda\), let
\[
\boldsymbol\Sigma_u^A
=
(\boldsymbol\Lambda_u^{-1}+\boldsymbol E^A)^{-1}
\]
and let \(\boldsymbol b_u^A\) be the corresponding conditional posterior mean error under \(\Pi_{\lambda,A}\). Put
\[
\boldsymbol K_u
=
(\boldsymbol\Sigma_u^A)^{1/2}
\boldsymbol F
(\boldsymbol\Sigma_u^A)^{1/2},
\qquad
\boldsymbol z_u
=
(\boldsymbol\Sigma_u^A)^{1/2}
\boldsymbol F\boldsymbol b_u^A.
\]
Then
\[
Q_\lambda(u)
=
\det(\boldsymbol I_p+\lambda\boldsymbol K_u)^{-1/2}
\exp\left[
-\frac{\lambda}{2}
(\boldsymbol b_u^A)^T
\boldsymbol F
\boldsymbol b_u^A
+
\frac{\lambda^2}{2}
\boldsymbol z_u^T
(\boldsymbol I_p+\lambda\boldsymbol K_u)^{-1}
\boldsymbol z_u
\right]
\]
is the conditional normalizing factor relating the full and reference scale posteriors. On \(\mathcal G_N\),
\[
\|\boldsymbol K_u\|_{\rm op}
\lesssim
\delta_n,
\qquad
\|\boldsymbol K_u\|_{\rm F}^2
\lesssim
\delta_n^2D_N,
\]
and
\begin{equation}
\label{eq:localized-normalizer}
|\log Q_\lambda(u)|
\leq
\frac{c_0}{a^2}
\|\boldsymbol b_{S,u}^A\|_2^2
+
\varepsilon_n
\{D_N+J_{N,\lambda}\},
\end{equation}
uniformly in \(0\leq\lambda\leq1\), where \(c_0>0\) is fixed and sufficiently small and
\[
\varepsilon_n
\lesssim
\delta_na^2+
(\delta_na^2)^2+
\delta_n^2
=
o(1).
\]
Moreover,
\[
\boldsymbol\Sigma_{\lambda,u}
=
(\boldsymbol\Sigma_u^A)^{1/2}
(\boldsymbol I_p+\lambda\boldsymbol K_u)^{-1}
(\boldsymbol\Sigma_u^A)^{1/2},
\]
and
\[
\left\|
(\boldsymbol\Sigma_u^A)^{-1/2}
(\boldsymbol\Sigma_{\lambda,u}-\boldsymbol\Sigma_u^A)
(\boldsymbol\Sigma_u^A)^{-1/2}
\right\|_{\rm op}
\lesssim
\delta_n
\]
on \(\mathcal G_N\).
\end{lemma}

\begin{proof}
Since \(\boldsymbol E^A\) is supported on \(S\times S\),
\[
(1+\gamma_n)^{-1}\boldsymbol\Lambda_u
\preceq
\boldsymbol\Sigma_u^A
\preceq
(1-\gamma_n)^{-1}\boldsymbol\Lambda_u
\]
for all sufficiently large \(n\). In particular, \(\boldsymbol\Sigma_u^A\) is block diagonal with respect to \(S\cup N\), and its null block equals \(\boldsymbol\Lambda_{u,N}\).

Write
\[
\boldsymbol F
=
\begin{pmatrix}
\boldsymbol0 & \boldsymbol B\\
\boldsymbol B^T & \boldsymbol D
\end{pmatrix},
\qquad
\boldsymbol B=\boldsymbol E_{SN},
\qquad
\boldsymbol D=\boldsymbol E_{NN}.
\]
On \(D_N\leq\eta m_n\), for every \(\boldsymbol v\in\mathbb R^{|S|}\),
\[
\boldsymbol v^T
\boldsymbol B
\boldsymbol\Lambda_{u,N}
\boldsymbol B^T
\boldsymbol v
=
\sum_{j\in N}
A_j
(\boldsymbol B_{\cdot j}^T\boldsymbol v)^2.
\]
The maximum over \(0\leq A_j\leq1\) with \(\sum_jA_j\leq\eta m_n\) is attained at an extreme point supported on at most \(m_n\) coordinates. Restricted orthogonality therefore gives
\[
\boldsymbol B
\boldsymbol\Lambda_{u,N}
\boldsymbol B^T
\preceq
C\delta_n^2\boldsymbol I_S.
\]
Similarly,
\[
\|\boldsymbol\Lambda_{u,N}^{1/2}\boldsymbol v\|_1
\leq
D_N^{1/2}\|\boldsymbol v\|_2
\leq
(\eta m_n)^{1/2}\|\boldsymbol v\|_2,
\]
so Lemma~\ref{lem:sparse-design-geometry} gives
\[
\left\|
\boldsymbol\Lambda_{u,N}^{1/2}
\boldsymbol D
\boldsymbol\Lambda_{u,N}^{1/2}
\right\|_{\rm op}
\lesssim
\delta_n.
\]
Together with \(\boldsymbol\Sigma_u^A\asymp\boldsymbol\Lambda_u\), these bounds imply \(\|\boldsymbol K_u\|_{\rm op} \lesssim \delta_n.\)

For the Frobenius norm,
\[
\|
\boldsymbol B
\boldsymbol\Lambda_{u,N}^{1/2}
\|_{\rm F}^2
=
\sum_{j\in N}
A_j
\|\boldsymbol E_{Sj}\|_2^2
\lesssim
\delta_n^2D_N.
\]
Likewise, for every \(i\in N\), the weighted restricted-orthogonality argument gives \(\sum_{j\in N} A_jE_{ij}^2 \lesssim \delta_n^2.\) Multiplication by \(A_i\) and summation over \(i\) therefore yield
\[
\left\|
\boldsymbol\Lambda_{u,N}^{1/2}
\boldsymbol D
\boldsymbol\Lambda_{u,N}^{1/2}
\right\|_{\rm F}^2
\lesssim
\delta_n^2D_N.
\]
Hence \(\|\boldsymbol K_u\|_{\rm F}^2 \lesssim \delta_n^2D_N.\)

Since
\[
\boldsymbol F_{SS}=\boldsymbol0,
\qquad
\operatorname{diag}(\boldsymbol D)=\boldsymbol0,
\]
we have \(\operatorname{tr}(\boldsymbol K_u)=0.\) Thus
\[
\left|
\log\det(\boldsymbol I_p+\lambda\boldsymbol K_u)
\right|
\lesssim
\|\boldsymbol K_u\|_{\rm F}^2
\lesssim
\delta_n^2D_N.
\]

Under the reference posterior,
\[
\boldsymbol b_{N,u}^A
=
\boldsymbol\Lambda_{u,N}
\boldsymbol W_{\lambda,N},
\]
and hence
\[
\|
\boldsymbol\Lambda_{u,N}^{-1/2}
\boldsymbol b_{N,u}^A
\|_2^2
=
a^2J_{N,\lambda}.
\]
The weighted cross-block bound gives
\[
\|
\boldsymbol B
\boldsymbol b_{N,u}^A
\|_2
\lesssim
\delta_na
J_{N,\lambda}^{1/2},
\]
while Lemma~\ref{lem:sparse-design-geometry} gives
\[
\left|
(\boldsymbol b_{N,u}^A)^T
\boldsymbol D
\boldsymbol b_{N,u}^A
\right|
\lesssim
\delta_na^2J_{N,\lambda}.
\]
Young's inequality therefore yields
\[
\left|
(\boldsymbol b_u^A)^T
\boldsymbol F
\boldsymbol b_u^A
\right|
\leq
\frac{c_0}{a^2}
\|\boldsymbol b_{S,u}^A\|_2^2
+
C\{
\delta_na^2+\delta_n^2a^4
\}
J_{N,\lambda}.
\]
The same weighted estimates give
\[
\|\boldsymbol z_u\|_2^2
\lesssim
\delta_n^2
\left\{
\|\boldsymbol b_{S,u}^A\|_2^2
+
a^2J_{N,\lambda}
\right\}.
\]
Because \(\|\boldsymbol K_u\|_{\rm op}=o(1)\),
\[
\|
(\boldsymbol I_p+\lambda\boldsymbol K_u)^{-1}
\|_{\rm op}
\lesssim1.
\]
Combining these estimates proves \eqref{eq:localized-normalizer}.

Finally,
\[
\boldsymbol\Sigma_{\lambda,u}
=
\left\{
(\boldsymbol\Sigma_u^A)^{-1}
+
\lambda\boldsymbol F
\right\}^{-1},
\]
which gives the displayed resolvent representation. The relative covariance bound follows from
\[
\|
(\boldsymbol I_p+\lambda\boldsymbol K_u)^{-1}
-
\boldsymbol I_p
\|_{\rm op}
\lesssim
\delta_n.
\]
\end{proof}

It remains to localize the score-weighted activity inside \(\mathcal G_N\). The next argument uses factorial moments rather than a full exponential moment and therefore never requires restricted-isometry control on more than \(m_n\) coordinates.

\begin{lemma}[Reference score localization]
\label{lem:reference-score-localization}
Fix \(\eta>0\) sufficiently small and a constant \(c_J>2\eta\). For every fixed integer \(0\leq r\leq16\),
\[
E_{\boldsymbol\theta^0}
\left\langle
(1+J_{N,\lambda})^r
1\{
D_N\leq\eta m_n,\,
J_{N,\lambda}>c_Jm_n
\}
\right\rangle_{\lambda,A}
\leq
\exp\left[
-cm_n
\log\left\{
\frac{m_n}{sL^\alpha}
\right\}
\right].
\]
The same bound remains valid after multiplication inside the posterior expectation by
\[
\exp\{C\varepsilon_nJ_{N,\lambda}\}
\]
for every fixed \(C>0\).
\end{lemma}

\begin{proof}
For \(\ell\geq0\), define
\[
\mathcal I_0=\{j\in N:W_{\lambda,j}^2\leq2a^2\},
\qquad
\mathcal I_\ell=\{j\in N:2^\ell a^2<W_{\lambda,j}^2\leq2^{\ell+1}a^2\},\quad \ell\geq1,
\]
and put \(S_\ell=\sum_{j\in\mathcal I_\ell}A_j\). Then \(J_{N,\lambda}\leq2S_0+\sum_{\ell\geq1}2^{\ell+1}S_\ell\). On \(D_N\leq\eta m_n\), \(S_0\leq\eta m_n\). Hence, if \(J_{N,\lambda}>c_Jm_n\) with \(c_J>2\eta\), \(\sum_{\ell\geq1}2^{\ell+1}S_\ell>(c_J-2\eta)m_n.\) Consequently, for a sufficiently small fixed \(c>0\), there exists some \(\ell\geq1\) such that \(S_\ell>c\frac{m_n2^{-\ell}}{(1+\ell)^2}.\)

For a vector \(\boldsymbol x=(x_j)\) with \(0\leq x_j\leq1\), let \(e_k(\boldsymbol x) = \sum_{|J|=k} \prod_{j\in J}x_j.\) If \(\sum_jx_j\geq x\) and \(k\leq x/2\), then \(e_k(\boldsymbol x) \geq \binom{\lfloor x\rfloor}{k}.\) For distinct null coordinates \(J\), \(|J|=k\leq m_n\), conditional independence of the reference scale posterior gives
\[
E_{\boldsymbol\theta^0}
\left\langle
\prod_{j\in J}
A_j1\{j\in\mathcal I_\ell\}
\right\rangle_{\lambda,A}
=
E
\prod_{j\in J}
w(W_{\lambda,j},\tau)
1\{j\in\mathcal I_\ell\}.
\]
The covariance matrix of \(\boldsymbol W_{\lambda,J}\) has eigenvalues in \([1-\delta_n,1+\delta_n]\). Gaussian density comparison therefore gives
\[
E
\prod_{j\in J}
w(W_{\lambda,j},\tau)
1\{j\in\mathcal I_\ell\}
\leq
(Cq_\ell)^k,
\]
where \(q_0\lesssim q_\pi(\tau),\) and, for \(\ell\geq1\), \(q_\ell \lesssim q_\pi(\tau) \exp\{-c2^\ell L\}.\) Here the factor generated by Gaussian density comparison is absorbed because \(\delta_nL=o(1)\) near the transition and, farther into the tail, only changes the Gaussian exponent by a factor \(1+o(1)\).

Let \(x_\ell = c \frac{m_n2^{-\ell}}{(1+\ell)^2}.\) Whenever \(x_\ell\geq2\), take \(k_\ell=\lfloor x_\ell/2\rfloor\). The preceding factorial-moment bound gives
\[
P_{\lambda,A}(S_\ell>x_\ell)
\leq
\exp\left[
-cx_\ell
\log\left\{
\frac{x_\ell}{pq_\ell}
\right\}
\right].
\]
If \(x_\ell<2\), the event \(S_\ell>x_\ell\) requires at least one score in \(\mathcal I_\ell\), and the Gaussian union bound gives \(P_{\lambda,A}(S_\ell>x_\ell) \leq p\exp(-c2^\ell L).\) Summing over \(\ell\) yields
\[
E_{\boldsymbol\theta^0}
\left\langle
(1+J_{N,\lambda})^r
1\{
D_N\leq\eta m_n,\,
J_{N,\lambda}>c_Jm_n
\}
\right\rangle_{\lambda,A}
\leq
\exp\left[
-cm_n
\log\left\{
\frac{m_n}{pq_\pi(\tau)}
\right\}
\right].
\]
Lemma~\ref{lem:oracle-null-activity} gives \(pq_\pi(\tau)=O(sL^\alpha)\).

Finally, on the dyadic shell \(2^hm_n < J_{N,\lambda} \leq 2^{h+1}m_n,\) multiplication by \(\exp(C\varepsilon_nJ_{N,\lambda})\) costs at most \(\exp(C\varepsilon_n2^{h+1}m_n)\). Since
\[
\varepsilon_n=o(1),
\qquad
\log\left\{
\frac{m_n}{sL^\alpha}
\right\}
\longrightarrow\infty,
\]
this factor is absorbed by the factorial tail uniformly over \(h\). This proves the final assertion.
\end{proof}

We can now transfer the score localization to the full posterior without controlling the normalizer on dense null-scale configurations.

\begin{lemma}[Posterior localization]
\label{lem:weighted-localization}
Suppose Assumption~1 holds with \(0<\alpha<1\), \eqref{eq:regression-signal-envelope} holds, and
\[
sL=o(m_n),
\qquad
\delta_nL\to0,
\qquad
\gamma_nL\to0.
\]
For every fixed integer \(0\leq r\leq8\), uniformly over \(0\leq\lambda\leq1\),
\[
E_{\boldsymbol\theta^0}
\left\langle
(1+H_{N,\lambda})^r
1\{
H_{N,\lambda}>C m_n
\}
\right\rangle_\lambda
\leq
\exp\left[
-cm_n
\log\left\{
\frac{m_n}{sL^\alpha}
\right\}
\right]
\]
for fixed sufficiently large \(C>0\).
\end{lemma}

\begin{proof}
Lemma~\ref{lem:full-scale-localization} already gives the required bound on \(\{D_N>\eta m_n\}.\) It therefore remains to work on \(\mathcal G_N = \{D_N\leq\eta m_n\}.\) On this set, Lemma~\ref{lem:cross-null-normalizer} gives \(e^{-R_\lambda(u)} \leq Q_\lambda(u) \leq e^{R_\lambda(u)},\) where
\[
R_\lambda(u)
=
\frac{c_0}{a^2}
\|\boldsymbol b_{S,u}^A\|_2^2
+
\varepsilon_n
\{D_N+J_{N,\lambda}\}.
\]

The active-block exponential-moment bound in Lemma~\ref{lem:scalar-response}, together with \(\gamma_na^2=o(1)\), implies that for every fixed \(q>0\),
\[
E_{\boldsymbol\theta^0}
\left\langle
\exp\left\{
\frac{qc_0}{a^2}
\|\boldsymbol b_{S,u}^A\|_2^2
\right\}
\right\rangle_{\lambda,A}
\leq
e^{Cs}.
\]
The reference scale-localization argument obtained from the same factorial moments as in Lemma~\ref{lem:reference-score-localization} also gives
\[
E_{\boldsymbol\theta^0}
\left\langle
1\{D_N>\eta m_n\}
\right\rangle_{\lambda,A}
\leq
\exp\left[
-cm_n
\log\left\{
\frac{m_n}{sL^\alpha}
\right\}
\right].
\]
Hence, outside a score event having probability bounded by the same quantity, the reference posterior assigns at least one half of its mass to \(\mathcal G_N\).

On this score event, Cauchy--Schwarz gives
\[
\begin{split}
\langle Q_\lambda\rangle_{\lambda,A}
&\geq
\left\langle
e^{-R_\lambda}
1_{\mathcal G_N}
\right\rangle_{\lambda,A}\\
&\geq
\frac{
\langle1_{\mathcal G_N}\rangle_{\lambda,A}^2
}{
\left\langle
e^{R_\lambda}
1_{\mathcal G_N}
\right\rangle_{\lambda,A}
}.
\end{split}
\]
Therefore, for every nonnegative \(f\) supported on \(\mathcal G_N\),
\[
\langle f\rangle_\lambda
\lesssim
\left\langle
fe^{R_\lambda}
\right\rangle_{\lambda,A}
\left\langle
e^{R_\lambda}
\right\rangle_{\lambda,A}.
\]
Applying Cauchy--Schwarz once more and then averaging over the score shows that the active factor contributes at most \(e^{Cs}\), while the null factor is controlled by \(e^{C\varepsilon_nJ_{N,\lambda}}.\)

Take
\[
f
=
(1+H_{N,\lambda})^r
1\{
\mathcal G_N,\,
H_{N,\lambda}>Cm_n
\}.
\]
Since \(H_{N,\lambda}=D_N+J_{N,\lambda}\) and \(D_N\leq\eta m_n\), choose the fixed constant \(C\) so that \(C>\eta+c_J\). Then, on the support of \(f\), \(J_{N,\lambda}>(C-\eta)m_n>c_Jm_n.\) Lemma~\ref{lem:reference-score-localization}, with \(2r\leq16\), therefore gives
\[
E_{\boldsymbol\theta^0}
\langle f\rangle_\lambda
\leq
e^{Cs}
\exp\left[
-cm_n
\log\left\{
\frac{m_n}{sL^\alpha}
\right\}
\right].
\]
Because
\[
s=o(m_n),
\qquad
\log\left\{
\frac{m_n}{sL^\alpha}
\right\}
\longrightarrow\infty,
\]
the factor \(e^{Cs}\) is absorbed into the localization exponent. Combining this with Lemma~\ref{lem:full-scale-localization} proves the result.
\end{proof}

Lemmas~\ref{lem:cross-null-normalizer} and \ref{lem:weighted-localization} provide the form of localization needed for the finite-order interpolation below. With posterior probability exponentially close to one, the null scale activity satisfies \(D_N=O(m_n)\), so the conditional covariance is a relative \(O(\delta_n)\) perturbation of the active-block covariance. The remaining score-weighted activity is controlled by factorial moments, and the complement remains negligible after multiplication by any fixed polynomial of degree at most eight.

\subsection{Active-block transfer}
\label{app:active-transfer}

We first control the perturbation generated by dependence within the true active set. Along the active-block interpolation, let
\[
\boldsymbol H_t^A
=
\boldsymbol I_p+t\boldsymbol E^A,
\qquad
\boldsymbol W_t^A
\sim
N_p(\boldsymbol0,\boldsymbol H_t^A),
\qquad
0\leq t\leq1.
\]
Since \(\boldsymbol E^A\) is supported on \(S\times S\), the null coordinates retain their independent sequence-model posterior throughout this interpolation. It is therefore enough to control the posterior on \(S=\operatorname{supp}(\boldsymbol\theta^0)\). Write
\[
\boldsymbol\mu_{S,t}
=
E_t(\boldsymbol\theta_S\mid\boldsymbol W_t^A),
\qquad
\boldsymbol b_{S,t}
=
\boldsymbol\mu_{S,t}-\boldsymbol\theta_S^0,
\]
and let \(\boldsymbol C_{S,t}\) denote the corresponding posterior covariance matrix. By Lemma~\ref{lem:scalar-response},
\[
v_\tau
:=
\sup_{y\in\mathbb R}V_\tau(y)
\lesssim
a^2,
\qquad
a=a_\tau,
\]
and \(a^2\asymp L\) at \(\tau=\tau_\pi^{\mathrm{or}}\).

\begin{lemma}[Active susceptibility]
\label{lem:active-susceptibility}
Suppose \(\gamma_nL\to0\). Uniformly over \(\|\boldsymbol\theta^0\|_0\leq s\), \(0\leq t\leq1\), and the Gaussian score,
\[
\|\boldsymbol C_{S,t}\|_{\rm op}
\lesssim
a^2.
\]
\end{lemma}

\begin{proof}
Let \(P_y\) denote the one-dimensional sequence posterior with observation \(y\). For every \(1\)-Lipschitz function \(f\),
\[
\frac{d}{dy}
E_y\{f(\theta)\}
=
\operatorname{Cov}_y\{f(\theta),\theta\}.
\]
If \(\theta\) and \(\theta'\) are independent draws from \(P_y\), then
\[
\operatorname{Var}_y\{f(\theta)\}
=
\frac12
E_y\{f(\theta)-f(\theta')\}^2
\leq
\frac12
E_y(\theta-\theta')^2
=
V_\tau(y).
\]
Cauchy--Schwarz therefore gives
\[
\left|
\frac{d}{dy}
E_y\{f(\theta)\}
\right|
\leq
V_\tau(y)
\leq
v_\tau.
\]
By the Kantorovich--Rubinstein representation, \(W_1(P_y,P_{y'}) \leq v_\tau|y-y'|.\)

Conditional on the remaining active coordinates, the posterior of coordinate \(i\in S\) is \(P_{y_{i,t}}\), where
\[
y_{i,t}
=
\theta_i^0+W_{t,i}^A
-
t
\sum_{\substack{j\in S\\j\neq i}}
E_{ij}
(\theta_j-\theta_j^0).
\]
Hence changing coordinate \(j\neq i\) from \(\theta_j\) to \(\theta_j'\) changes the \(i\)th full conditional by at most \(t v_\tau|E_{ij}| |\theta_j-\theta_j'|\) in \(W_1\). Define the nonnegative matrix \(D_{ij} = t v_\tau|E_{ij}|1(i\neq j), \qquad i,j\in S.\) By the definition of \(\gamma_n\),
\[
\|D\|_{\infty\to\infty}
\leq
v_\tau\gamma_n
\lesssim
a^2\gamma_n
=
o(1).
\]

Perturb the \(j\)th score coordinate by \(\epsilon\), and couple the two active posteriors by successively coupling their full conditional distributions. If \(d_i\) denotes the expected absolute displacement of coordinate \(i\) under this coupling, the preceding contraction gives
\[
\boldsymbol d
\leq
D\boldsymbol d
+
v_\tau|\epsilon|\boldsymbol e_j
\]
componentwise. Since \(\|D\|_{\infty\to\infty}<1\) for all sufficiently large \(n\),
\[
\boldsymbol d
\leq
v_\tau|\epsilon|
(\boldsymbol I-D)^{-1}\boldsymbol e_j.
\]
Posterior differentiation with respect to the Gaussian score gives \(\frac{\partial\mu_{i,t}} {\partial W_{t,j}^A} = C_{ij,t}.\) Letting \(\epsilon\to0\) yields \(|C_{ij,t}| \leq v_\tau \{(\boldsymbol I-D)^{-1}\}_{ij}.\) Since \(D\) is symmetric and nonnegative,
\[
\|(\boldsymbol I-D)^{-1}\|_{\infty\to\infty}
\leq
\frac{1}
{1-\|D\|_{\infty\to\infty}},
\]
and hence
\[
\|\boldsymbol C_{S,t}\|_{\infty\to\infty}
\leq
\frac{v_\tau}
{1-v_\tau\gamma_n}
\lesssim
a^2.
\]
The covariance matrix is symmetric, so
\[
\|\boldsymbol C_{S,t}\|_{\rm op}
\leq
\|\boldsymbol C_{S,t}\|_{\infty\to\infty},
\]
which proves the result.
\end{proof}

We now transfer the active-coordinate risk. Define
\[
\mathcal A_t(\boldsymbol\theta^0)
=
E_{\boldsymbol\theta^0}
\|\boldsymbol b_{S,t}\|_2^2.
\]
Lemma~\ref{lem:gaussian-interpolation}, applied with perturbation matrix \(\boldsymbol E^A\), gives
\[
\frac{d}{dt}
\mathcal A_t(\boldsymbol\theta^0)
=
E_{\boldsymbol\theta^0}
\operatorname{tr}
(\boldsymbol E_{SS}\boldsymbol C_{S,t}^2)
-
2E_{\boldsymbol\theta^0}
\left(
\boldsymbol b_{S,t}^{T}
\boldsymbol C_{S,t}
\boldsymbol E_{SS}
\boldsymbol b_{S,t}
\right).
\]
By the definition of \(\gamma_n\),
\[
\|\boldsymbol E_{SS}\|_{\rm op}
\leq
\|\boldsymbol E_{SS}\|_{\infty\to\infty}
\leq
\gamma_n.
\]
Lemma~\ref{lem:active-susceptibility} therefore gives
\[
\begin{split}
\left|
\operatorname{tr}
(\boldsymbol E_{SS}\boldsymbol C_{S,t}^2)
\right|
&\leq
\gamma_n
\operatorname{tr}
(\boldsymbol C_{S,t}^2)\\
&\leq
\gamma_n
\|\boldsymbol C_{S,t}\|_{\rm op}
\operatorname{tr}(\boldsymbol C_{S,t})\\
&\lesssim
\gamma_nsa^4,
\end{split}
\]
where
\[
\operatorname{tr}(\boldsymbol C_{S,t})
\leq
|S|
\|\boldsymbol C_{S,t}\|_{\rm op}
\leq
s
\|\boldsymbol C_{S,t}\|_{\rm op}.
\]
Similarly,
\[
\left|
\boldsymbol b_{S,t}^{T}
\boldsymbol C_{S,t}
\boldsymbol E_{SS}
\boldsymbol b_{S,t}
\right|
\leq
\|\boldsymbol C_{S,t}\|_{\rm op}
\|\boldsymbol E_{SS}\|_{\rm op}
\|\boldsymbol b_{S,t}\|_2^2
\lesssim
\gamma_na^2
\|\boldsymbol b_{S,t}\|_2^2.
\]
Consequently,
\begin{equation}
\label{eq:active-risk-differential}
\left|
\frac{d}{dt}
\mathcal A_t(\boldsymbol\theta^0)
\right|
\lesssim
\gamma_nsa^4
+
\gamma_na^2
\mathcal A_t(\boldsymbol\theta^0).
\end{equation}

At \(t=0\), the active coordinates follow independent sequence-model posteriors. Hence
\[
\mathcal A_0(\boldsymbol\theta^0)
=
\sum_{j\in S}
R_{\theta_j^0}(\pi,\tau)
\leq
|S|R_\star(\pi,\tau)
\lesssim
sa^2
\]
by Proposition~\ref{prop:active-risk}. Since
\[
\gamma_na^2
\lesssim
\gamma_nL
\longrightarrow0,
\]
Gronwall's inequality applied to \eqref{eq:active-risk-differential} gives
\[
\sup_{\|\boldsymbol\theta^0\|_0\leq s}
\sup_{0\leq t\leq1}
\mathcal A_t(\boldsymbol\theta^0)
\lesssim
sa^2.
\]
Substitution into \eqref{eq:active-risk-differential} yields
\[
\sup_{\|\boldsymbol\theta^0\|_0\leq s}
\sup_{0\leq t\leq1}
\left|
\frac{d}{dt}
\mathcal A_t(\boldsymbol\theta^0)
\right|
\lesssim
\gamma_nsa^4.
\]
Integrating over \(t\in[0,1]\) therefore gives
\[
\sup_{\|\boldsymbol\theta^0\|_0\leq s}
\left|
\mathcal A_1(\boldsymbol\theta^0)
-
\mathcal A_0(\boldsymbol\theta^0)
\right|
\lesssim
\gamma_nsa^4.
\]

The null-coordinate posterior is unchanged throughout the active-block interpolation. Therefore
\[
\mathcal R_1^A(\boldsymbol\theta^0)
-
\mathcal R_0^A(\boldsymbol\theta^0)
=
\mathcal A_1(\boldsymbol\theta^0)
-
\mathcal A_0(\boldsymbol\theta^0).
\]
Since \(a^2\asymp L\),
\begin{equation}
\label{eq:active-block-transfer}
\sup_{\|\boldsymbol\theta^0\|_0\leq s}
\left|
\mathcal R_1^A(\boldsymbol\theta^0)
-
\mathcal R_0^A(\boldsymbol\theta^0)
\right|
\lesssim
\gamma_nsL^2.
\end{equation}
Thus cumulative dependence within the active set contributes at most order \(\gamma_nsL^2\), uniformly over the entire sparse parameter space.

\subsection{Cross--null transfer and completion of the proof}
\label{app:cross-null-transfer}

It remains to control the perturbation generated by the active--null and null--null blocks. Throughout this subsection, set
\[
\tau=\tau_\pi^{\mathrm{or}},
\qquad
a=a_\tau,
\qquad
r_0=R_0(\pi,\tau),
\qquad
q=q_\pi(\tau),
\qquad
\overline q_\tau
=
E_0\{V_\tau(Y)\}.
\]
By Lemmas~\ref{lem:oracle-null-activity} and \ref{lem:scalar-response},
\[
a^2\asymp L,
\qquad
pr_0=O(s),
\qquad
pq=O(sL^\alpha),
\qquad
p\overline q_\tau=O(sL).
\]
All truths considered below satisfy the signal envelope \eqref{eq:regression-signal-envelope}.

Write \(k=|S|\leq s\) and define
\[
d_\tau(\vartheta)
=
R_\star(\pi,\tau)-R_\vartheta(\pi,\tau).
\]
The total deficit from the sequence-model sparse supremum is
\[
\mathfrak D_\tau(\boldsymbol\theta^0)
=
(s-k)\{R_\star(\pi,\tau)-r_0\}
+
\sum_{j\in S}
d_\tau(\theta_j^0).
\]
The coordinatewise sequence-model decomposition therefore gives
\begin{equation}
\label{eq:sequence-risk-deficit}
\mathcal R_{\boldsymbol I_p}(\boldsymbol\theta^0)
=
R_{p,s}^{\mathrm{or}}(\pi)
-
\mathfrak D_\tau(\boldsymbol\theta^0).
\end{equation}

We first relate the scalar response multiplier in Lemma~\ref{lem:scalar-response} to this deficit.

\begin{lemma}[Response and risk deficit]
\label{lem:response-deficit}
Let
\[
\Lambda_a(\vartheta)
=
1+
a\{\phi(a-\vartheta)+\phi(a+\vartheta)\}.
\]
Uniformly in \(\vartheta\in\mathbb R\),
\[
\Lambda_a(\vartheta)
\lesssim
1+d_\tau(\vartheta).
\]
Consequently,
\[
\sum_{j\in S}
\Lambda_a(\theta_j^0)
\lesssim
s+\mathfrak D_\tau(\boldsymbol\theta^0).
\]
\end{lemma}

\begin{proof}
By symmetry it suffices to consider \(\vartheta\geq0\). Put \(t=a-\vartheta, \qquad \delta_a=t_a^{-1/4}.\) If \(|t-t_a|\leq\delta_a\), the local calculation in the proof of Proposition~\ref{prop:active-risk} gives \(R_{a-t}(\pi,\tau) = H_a(t)-a\phi(t)+o(1)\) uniformly on this neighbourhood, while \(R_\star(\pi,\tau) = H_a(t_a)-2+o(1).\) Hence
\[
d_\tau(\vartheta)
=
H_a(t_a)-H_a(t)
+
a\phi(t)-2+o(1)
\geq
a\phi(t)-3
\]
for all sufficiently large \(a\).

If \(t\geq t_a+\delta_a\), then \(a\phi(t) \leq a\phi(t_a) = 2\{1+o(1)\}.\) If \(t\leq t_a-\delta_a\), the separation argument in the proof of Proposition~\ref{prop:active-risk} gives \(H_a(t_a)-H_a(t)\gg a.\) Together with
\[
\sup_{\vartheta\geq0}
\left|
R_\vartheta(\pi,\tau)-H_a(a-\vartheta)
\right|
=
O(a),
\]
this implies \(d_\tau(\vartheta)\gg a.\) Since \(a\phi(t)\leq Ca\), it follows again that \(a\phi(a-\vartheta) \lesssim 1+d_\tau(\vartheta).\) Finally, \(a\phi(a+\vartheta) \leq a\phi(a) = o(1),\) which proves the first assertion. Summation over \(S\), together with
\[
\sum_{j\in S}
d_\tau(\theta_j^0)
\leq
\mathfrak D_\tau(\boldsymbol\theta^0),
\]
proves the second.
\end{proof}

We next record the fixed-order scalar response bound needed after one further interpolation derivative. For a monomial
\[
M_{\ell r q v}(y;\vartheta)
=
|b_\vartheta(y)|^\ell
V_\tau(y)^r
|T_\tau(y)|^q
|U_\tau(y)|^v,
\]
define its response degree by
\[
\deg(M_{\ell r q v})
=
\ell+2r+3q+4v.
\]

\begin{lemma}[Third-order scalar responses]
\label{lem:third-order-scalar}
For every fixed nonnegative quadruple \((\ell,r,q,v)\) satisfying
\[
r+q+v\geq1,
\qquad
\ell+2r+3q+4v\leq8,
\]
we have, uniformly in \(\vartheta\in\mathbb R\),
\[
E_\vartheta
M_{\ell r q v}(Y;\vartheta)
\lesssim
a^6\Lambda_a(\vartheta).
\]
At the null,
\[
E_0
M_{\ell r q v}(Y;0)
\lesssim
a^6r_0.
\]
\end{lemma}

\begin{proof}
Use the central--transition--upper decomposition from the proof of Lemma~\ref{lem:scalar-response}. On the transition region \(x=\frac{y^2-a^2}{2}, \qquad |x|\leq M_L,\) we have \(V_\tau(y) \lesssim 1+a^2e^{-|x|},\)
\[
|T_\tau(y)|
\lesssim
a^3e^{-|x|}+a^{-3},
\qquad
|U_\tau(y)|
\lesssim
a^4e^{-|x|}+a^{-4}.
\]
Moreover, \(|b_\vartheta(y)| \leq |y-\vartheta|+|y|.\) Since the response degree is at most eight, every such monomial is bounded, up to a fixed Gaussian-polynomial factor in \(y-\vartheta\), by \(Ca^8e^{-c|x|}+C.\) Because \(dy=dx/y\), \(y\asymp a\), and the Gaussian density across the transition window is controlled as in Lemma~\ref{lem:scalar-response}, integration gives
\[
E_\vartheta
\left[
M_{\ell r q v}(Y;\vartheta)
1\{|Y^2-a^2|\leq2M_L\}
\right]
\lesssim
a^6\Lambda_a(\vartheta).
\]

Below the transition, the central leakage bounds make all response terms exponentially smaller than their transition orders. Above the transition,
\[
V_\tau(y)\lesssim1,
\qquad
|T_\tau(y)|\lesssim |y|^{-3},
\qquad
|U_\tau(y)|\lesssim y^{-4},
\]
while \(|b_\vartheta(y)| \leq |y-\vartheta|+C/|y|.\) The corresponding Gaussian moments are \(O(a^6)\), proving the first assertion.

At the null, \(\phi(y)=\phi(a)e^{-x}, \qquad dy=\frac{dx}{y}\) on the transition region. The largest response degree therefore gives \(E_0 M_{\ell r q v}(Y;0) \lesssim a^7\phi(a).\) Since \(r_0 = 2a\phi(a)\{1+o(1)\},\) this is \(O(a^6r_0)\). The two complementary regions are smaller.
\end{proof}

The active-block posterior responses are propagated through the same weak-dependence resolvent used in Lemma~\ref{lem:active-susceptibility}.

\begin{lemma}[Active response contractions]
\label{lem:active-response-contractions}
Let \(\boldsymbol W\) be a deterministic positive semidefinite matrix on \(S\times S\). Under \(\gamma_nL\to0\), uniformly along the active-block posterior,
\[
\left|
E_{\boldsymbol\theta^0}
\operatorname{tr}
\{
\boldsymbol W
\boldsymbol T_S[\boldsymbol b_S]
\}
\right|
\lesssim
\|\boldsymbol W\|_{\rm op}
a^2
\sum_{j\in S}
\Lambda_a(\theta_j^0),
\]
and
\[
\left|
E_{\boldsymbol\theta^0}
\operatorname{tr}
\{
\boldsymbol W
\boldsymbol T_S[\boldsymbol b_S]
\boldsymbol C_S
\}
\right|
\lesssim
\|\boldsymbol W\|_{\rm op}
a^4
\sum_{j\in S}
\Lambda_a(\theta_j^0).
\]
The analogous contractions obtained after one additional score derivative are bounded by
\[
C\|\boldsymbol W\|_{\rm op}
a^6
\sum_{j\in S}
\Lambda_a(\theta_j^0).
\]
\end{lemma}

\begin{proof}
Let \(D_{ij} = v_\tau|E_{ij}|1(i\neq j), \qquad v_\tau\lesssim a^2,\) as in Lemma~\ref{lem:active-susceptibility}, and put \(\boldsymbol R_D = (\boldsymbol I-D)^{-1}.\) Since \(\|D\|_{\infty\to\infty} \lesssim \gamma_na^2 = o(1),\) the row and column sums of \(\boldsymbol R_D\) are uniformly bounded.

Differentiate the full-conditional response equations with respect to the score. At first order, the local source is the scalar posterior variance and the interaction terms are propagated through \(\boldsymbol R_D\), which is the calculation in Lemma~\ref{lem:active-susceptibility}. At second order, after contracting with \(\boldsymbol b_S\), the only new scalar source terms are \(V_\tau(Y)^2 \quad\text{and}\quad |b_\vartheta(Y)T_\tau(Y)|.\) Thus the componentwise response vector \(\boldsymbol r^{(2)}\) satisfies an inequality of the form
\[
\boldsymbol r^{(2)}
\leq
D\boldsymbol r^{(2)}
+
\boldsymbol s^{(2)},
\]
where
\[
\sum_{j\in S}s_j^{(2)}
\lesssim
\sum_{j\in S}
E_{\theta_j^0}
\left[
V_\tau(Y)^2+
|b_{\theta_j^0}(Y)T_\tau(Y)|
\right].
\]
Hence
\[
\sum_{j\in S}r_j^{(2)}
\lesssim
a^2
\sum_{j\in S}
\Lambda_a(\theta_j^0)
\]
by Lemma~\ref{lem:scalar-response}. This yields
\[
E_{\boldsymbol\theta^0}
\|
\boldsymbol T_S[\boldsymbol b_S]
\|_*
\lesssim
a^2
\sum_{j\in S}
\Lambda_a(\theta_j^0).
\]
The first displayed bound follows from
\[
|\operatorname{tr}(\boldsymbol W\boldsymbol A)|
\leq
\|\boldsymbol W\|_{\rm op}\|\boldsymbol A\|_*.
\]

Multiplication by \(\boldsymbol C_S\) costs at most \(Ca^2\) by Lemma~\ref{lem:active-susceptibility}, giving the second bound. Differentiating once more produces scalar source monomials with response degree at most eight. Lemma~\ref{lem:third-order-scalar} bounds their total source by \(Ca^6 \sum_{j\in S} \Lambda_a(\theta_j^0).\) Propagation through the same resolvent \(\boldsymbol R_D\) does not change this order, proving the final assertion.
\end{proof}

We now return to the cross--null interpolation
\[
\boldsymbol H_\lambda^C
=
\boldsymbol I_p+\boldsymbol E^A+\lambda\boldsymbol F,
\qquad
0\leq\lambda\leq1,
\]
and write
\[
\boldsymbol F
=
\begin{pmatrix}
\boldsymbol0 & \boldsymbol B\\
\boldsymbol B^T & \boldsymbol D
\end{pmatrix},
\qquad
\boldsymbol B=\boldsymbol E_{SN},
\qquad
\boldsymbol D=\boldsymbol E_{NN}.
\]
Let
\[
\widetilde{\mathcal R}_\lambda
=
\mathcal R_{\boldsymbol H_\lambda^C}
(\boldsymbol\theta^0).
\]

\begin{lemma}[Second-order cross--null response]
\label{lem:cross-null-second}
Uniformly over the regression truth class,
\[
\widetilde{\mathcal R}_0'=0
\]
and
\[
\left|
\widetilde{\mathcal R}_0''
\right|
\leq
C\delta_n^2sa^4
+
\varepsilon_{2,n}
\mathfrak D_\tau(\boldsymbol\theta^0),
\]
where
\[
\varepsilon_{2,n}
\lesssim
(\delta_na^2)^2
+
\frac{\delta_n^2pa^4}{m_n}
\{r_0+\overline q_\tau\}
=
o(1).
\]
\end{lemma}

\begin{proof}
At \(\lambda=0\), the active-block posterior and the scalar null posteriors are independent. Thus
\[
\boldsymbol C_0
=
\begin{pmatrix}
\boldsymbol C_S & \boldsymbol0\\
\boldsymbol0 & \boldsymbol V_N
\end{pmatrix},
\]
where \(\boldsymbol V_N = \operatorname{diag} \{V_\tau(W_j):j\in N\}.\) The null posterior mean is odd and the null posterior variance is even, so \(E_0b_j=0, \qquad E_0\{b_jV_\tau(W_j)\}=0.\) Since
\[
\boldsymbol F_{SS}=\boldsymbol0,
\qquad
\operatorname{diag}(\boldsymbol D)=\boldsymbol0,
\]
Lemma~\ref{lem:gaussian-interpolation} gives \(\widetilde{\mathcal R}_0'=0.\)

Put \(\boldsymbol u = \boldsymbol F\boldsymbol b.\) Differentiating the first-derivative identity gives
\[
\begin{split}
\widetilde{\mathcal R}_\lambda''
=
E\bigg[
&
-4\operatorname{tr}
(\boldsymbol F\boldsymbol C^2
\boldsymbol F\boldsymbol C)
-4\operatorname{tr}
\{
\boldsymbol F\boldsymbol C
\boldsymbol T[\boldsymbol u]
\}\\
&
+
\sum_{r,s}F_{rs}
\operatorname{tr}
(\boldsymbol F\boldsymbol T_r\boldsymbol T_s)
+
2\boldsymbol u^T\boldsymbol C^2\boldsymbol u\\
&
+
4\boldsymbol b^T
\boldsymbol C\boldsymbol F
\boldsymbol C\boldsymbol u
+
2\boldsymbol b^T
\boldsymbol T[\boldsymbol u]\boldsymbol u\\
&
-
2\operatorname{tr}
\{
\boldsymbol T[\boldsymbol b]
\boldsymbol C\boldsymbol F
\boldsymbol C\boldsymbol F
\}
\bigg].
\end{split}
\]
At \(\lambda=0\),
\[
\operatorname{tr}
\{
\boldsymbol F\boldsymbol C
\boldsymbol T[\boldsymbol u]
\}
=
0
\]
samplewise. Indeed, \(\boldsymbol C\boldsymbol T[\boldsymbol u]\) is block diagonal, the active diagonal block of \(\boldsymbol F\) is zero, and the null diagonal of \(\boldsymbol D\) is zero. Likewise,
\[
\sum_{r,s}F_{rs}
\operatorname{tr}
(\boldsymbol F\boldsymbol T_r\boldsymbol T_s)
=
0,
\]
because whenever \(F_{rs}\neq0\), the supports of \(\boldsymbol T_r\) and \(\boldsymbol T_s\) are disjoint.

After averaging over the independent null coordinates, the remaining terms involving the active third cumulant reduce to
\[
2r_0
E_{\boldsymbol\theta^0}
\operatorname{tr}
\{
\boldsymbol T_S[\boldsymbol b_S]
\boldsymbol B\boldsymbol B^T
\},
\]
and
\[
-2\overline q_\tau
E_{\boldsymbol\theta^0}
\operatorname{tr}
\{
\boldsymbol T_S[\boldsymbol b_S]
\boldsymbol C_S
\boldsymbol B\boldsymbol B^T
\}.
\]
All other terms contain covariance contractions or scalar null responses.

Partitioning \(N\) into blocks of size \(m_n\) and using restricted orthogonality gives
\[
\|\boldsymbol B\boldsymbol B^T\|_{\rm op}
\lesssim
\delta_n^2\frac{p}{m_n},
\]
together with
\[
\|\boldsymbol B\|_{\rm F}^2
\lesssim
\delta_n^2\frac{sp}{m_n},
\qquad
\|\boldsymbol D\|_{\rm F}^2
\lesssim
\delta_n^2\frac{p^2}{m_n}.
\]
Lemma~\ref{lem:active-response-contractions} and Lemma~\ref{lem:response-deficit} therefore bound the two active third-cumulant terms by
\[
C
\frac{\delta_n^2p}{m_n}
\{r_0a^2+\overline q_\tau a^4\}
\left\{
s+\mathfrak D_\tau(\boldsymbol\theta^0)
\right\}.
\]
The remaining covariance and null-response contractions are bounded, using Lemma~\ref{lem:scalar-response}, by
\[
C\delta_n^2sa^4
+
C(\delta_na^2)^2
\mathfrak D_\tau(\boldsymbol\theta^0).
\]
The \(s\)-parts of the preceding active-response bound are absorbed into \(C\delta_n^2sa^4\). Hence
\[
\left|
\widetilde{\mathcal R}_0''
\right|
\leq
C\delta_n^2sa^4
+
\varepsilon_{2,n}
\mathfrak D_\tau(\boldsymbol\theta^0),
\]
with
\[
\varepsilon_{2,n}
\lesssim
(\delta_na^2)^2
+
\frac{\delta_n^2pa^4}{m_n}
\{r_0+\overline q_\tau\}.
\]

Finally, \((\delta_na^2)^2 \asymp (\delta_nL)^2 \longrightarrow0.\) Moreover,
\[
\frac{\delta_n^2pa^4r_0}{m_n}
\lesssim
\frac{sL}{m_n}
\frac{(\delta_nL)^2}{L},
\]
whereas
\[
\frac{\delta_n^2pa^4\overline q_\tau}{m_n}
\lesssim
\frac{sL}{m_n}
(\delta_nL)^2.
\]
Both terms vanish because \(sL=o(m_n)\) and \(\delta_nL\to0\).
\end{proof}

The third derivative requires no posterior cumulants beyond order four.

\begin{lemma}[Third-order cross--null response]
\label{lem:cross-null-third}
Under the conditions of Lemma~\ref{lem:weighted-localization}, uniformly in \(0\leq\lambda\leq1\),
\[
\left|
\widetilde{\mathcal R}_\lambda'''
\right|
\leq
C\delta_n^3sa^6
+
\varepsilon_{3,n}
\mathfrak D_\tau(\boldsymbol\theta^0),
\]
where
\[
\varepsilon_{3,n}
\lesssim
(\delta_na^2)^3
+
\frac{\delta_n^3pa^6}{m_n}
\{r_0+\overline q_\tau\}
=
o(1).
\]
\end{lemma}

\begin{proof}
The generator identities in Lemma~\ref{lem:gaussian-interpolation} give
\[
\mathscr D_\lambda\boldsymbol b
=
-\boldsymbol C\boldsymbol F\boldsymbol b,
\]
and
\[
\mathscr D_\lambda\boldsymbol C
=
-\boldsymbol C\boldsymbol F\boldsymbol C
-
\boldsymbol T[\boldsymbol F\boldsymbol b].
\]
For the third cumulant,
\[
\mathscr D_\lambda\boldsymbol T
=
-\mathcal S(\boldsymbol T,\boldsymbol F,\boldsymbol C)
-
\boldsymbol U[\boldsymbol F\boldsymbol b],
\]
where \(\mathcal S\) denotes the sum over the three contractions obtained by inserting \(\boldsymbol F\boldsymbol C\) into one index of \(\boldsymbol T\).

Differentiate the expression for \(\widetilde{\mathcal R}_\lambda''\), using these identities and the Gaussian product rule associated with \(\mathscr D_\lambda\). Every resulting term contains exactly three copies of \(\boldsymbol F\). The posterior-response part is a finite sum of contractions formed from
\[
\boldsymbol b,\qquad
\boldsymbol C,\qquad
\boldsymbol T,\qquad
\boldsymbol U,
\]
with scalar response degree at most eight. Thus no cumulant of order greater than four appears.

On the localized scale region \(D_N\leq\eta m_n,\) let
\[
\mathcal B_u
=
\boldsymbol B
\boldsymbol\Lambda_{u,N}^{1/2},
\qquad
\mathcal D_u
=
\boldsymbol\Lambda_{u,N}^{1/2}
\boldsymbol D
\boldsymbol\Lambda_{u,N}^{1/2}.
\]
The weighted sparse-geometry bounds in Lemma~\ref{lem:cross-null-normalizer} give
\[
\|\mathcal B_u\|_{\rm op}
+
\|\mathcal D_u\|_{\rm op}
\lesssim
\delta_n,
\]
and
\[
\|\mathcal B_u\|_{\rm F}^2
+
\|\mathcal D_u\|_{\rm F}^2
\lesssim
\delta_n^2D_N.
\]
Hence every connected contraction containing three cross--null edges is bounded by \(C\delta_n^3\) times its posterior-response source. When one null endpoint is closed before scale weighting, the block partition of \(N\) contributes at most the additional factor \(Cp/m_n\). The only resulting design topologies are the active--null--active chains, the chains containing one null--null edge, and the entirely null chains; the preceding operator and Frobenius bounds cover all three.

The active parts of these contractions are propagated through the Dobrushin resolvent from Lemma~\ref{lem:active-response-contractions}. By Lemma~\ref{lem:third-order-scalar} and Lemma~\ref{lem:response-deficit}, contractions whose null endpoints are fully scale weighted contribute at most
\[
C\delta_n^3a^6
\left\{
s+
\mathfrak D_\tau(\boldsymbol\theta^0)
\right\}.
\]
Terms containing one unweighted null endpoint contribute at most
\[
C
\frac{\delta_n^3pa^6}{m_n}
\{r_0+\overline q_\tau\}
\left\{
s+
\mathfrak D_\tau(\boldsymbol\theta^0)
\right\}.
\]
The \(s\)-parts are absorbed into \(C\delta_n^3sa^6\).

Lemma~\ref{lem:weighted-localization} makes the complement of the localized region negligible after multiplication by every fixed response polynomial occurring in the preceding differentiation. The exponential weight generated by the full-to-reference posterior comparison is absorbed by the factorial score tail from Lemma~\ref{lem:reference-score-localization}. Hence
\[
\left|
\widetilde{\mathcal R}_\lambda'''
\right|
\leq
C\delta_n^3sa^6
+
\varepsilon_{3,n}
\mathfrak D_\tau(\boldsymbol\theta^0)
\]
uniformly in \(\lambda\), with
\[
\varepsilon_{3,n}
\lesssim
(\delta_na^2)^3
+
\frac{\delta_n^3pa^6}{m_n}
\{r_0+\overline q_\tau\}.
\]

Now \((\delta_na^2)^3 \asymp (\delta_nL)^3 \longrightarrow0.\) Furthermore,
\[
\frac{\delta_n^3pa^6r_0}{m_n}
\lesssim
\frac{sL}{m_n}
\frac{(\delta_nL)^3}{L},
\]
and
\[
\frac{\delta_n^3pa^6\overline q_\tau}{m_n}
\lesssim
\frac{sL}{m_n}
(\delta_nL)^3.
\]
Both vanish under \(sL=o(m_n)\) and \(\delta_nL\to0\), proving the result.
\end{proof}

We can now complete the cross--null transfer. Taylor's formula and Lemma~\ref{lem:cross-null-second} give
\[
\widetilde{\mathcal R}_1
-
\widetilde{\mathcal R}_0
=
\frac12
\widetilde{\mathcal R}_0''
+
\frac12
\int_0^1
(1-\lambda)^2
\widetilde{\mathcal R}_\lambda'''
\,d\lambda.
\]
Since
\[
\delta_n^3sa^6
=
\delta_n^2sa^4
(\delta_na^2)
=
o(\delta_n^2sa^4),
\]
Lemmas~\ref{lem:cross-null-second} and \ref{lem:cross-null-third} imply that, for some \(\varepsilon_n\to0\),
\begin{equation}
\label{eq:cross-null-transfer}
\left|
\widetilde{\mathcal R}_1(\boldsymbol\theta^0)
-
\widetilde{\mathcal R}_0(\boldsymbol\theta^0)
\right|
\leq
C\delta_n^2sL^2
+
\varepsilon_n
\mathfrak D_\tau(\boldsymbol\theta^0)
\end{equation}
uniformly over the regression truth class.

\begin{proof}[of Theorem~4]
The two-stage interpolation in Section~\ref{app:regression-reduction} gives
\[
\mathcal R_{\boldsymbol G}(\boldsymbol\theta^0)
-
\mathcal R_{\boldsymbol I_p}(\boldsymbol\theta^0)
=
\{
\mathcal R_1^A(\boldsymbol\theta^0)
-
\mathcal R_0^A(\boldsymbol\theta^0)
\}
+
\{
\widetilde{\mathcal R}_1(\boldsymbol\theta^0)
-
\widetilde{\mathcal R}_0(\boldsymbol\theta^0)
\}.
\]
Combining \eqref{eq:active-block-transfer}, \eqref{eq:cross-null-transfer}, and \eqref{eq:sequence-risk-deficit} yields
\[
\begin{split}
\mathcal R_{\boldsymbol G}(\boldsymbol\theta^0)
\leq\;&
R_{p,s}^{\mathrm{or}}(\pi)
-
\{1-\varepsilon_n\}
\mathfrak D_\tau(\boldsymbol\theta^0)\\
&+
CsL^2
\{\gamma_n+\delta_n^2\}.
\end{split}
\]
Since \(\varepsilon_n\to0\), taking the supremum over the regression truth class gives
\[
R_{n,p,s}^{X,\mathrm{seq}}(\pi)
\leq
R_{p,s}^{\mathrm{or}}(\pi)
+
CsL^2
\{\gamma_n+\delta_n^2\}.
\]

For the reverse inequality, let \(\vartheta_\tau^\star\) be a scalar signal attaining \(R_\star(\pi,\tau)\), and choose any support \(S\) of size \(s\) with \(\theta_j^0 = \vartheta_\tau^\star, \qquad j\in S.\) Proposition~\ref{prop:active-risk} gives \(|\vartheta_\tau^\star| = O(a) = O(L^{1/2}),\) so this truth satisfies \eqref{eq:regression-signal-envelope}. Moreover,
\[
\mathfrak D_\tau(\boldsymbol\theta^0)=0,
\qquad
\mathcal R_{\boldsymbol I_p}(\boldsymbol\theta^0)
=
R_{p,s}^{\mathrm{or}}(\pi).
\]
Equations~\eqref{eq:active-block-transfer} and \eqref{eq:cross-null-transfer} therefore give
\[
\mathcal R_{\boldsymbol G}(\boldsymbol\theta^0)
\geq
R_{p,s}^{\mathrm{or}}(\pi)
-
CsL^2
\{\gamma_n+\delta_n^2\}.
\]
Taking the supremum yields
\[
R_{n,p,s}^{X,\mathrm{seq}}(\pi)
\geq
R_{p,s}^{\mathrm{or}}(\pi)
-
CsL^2
\{\gamma_n+\delta_n^2\}.
\]
Consequently,
\[
R_{n,p,s}^{X,\mathrm{seq}}(\pi)
=
R_{p,s}^{\mathrm{or}}(\pi)
+
O\left[
sL^2
\{\gamma_n+\delta_n^2\}
\right],
\]
which proves the theorem.
\end{proof}

\section{Computational details for the numerical study}
\label{app:computational-details}

\subsection{Deterministic evaluation of sequence-model risk}
\label{app:sequence-computation}

Figures 1--3 of the main paper are based on deterministic numerical evaluation of the posterior mean and its worst-case squared-error risk in the Gaussian sequence model. No Monte Carlo approximation is used in these calculations. For a local-variance prior \(\pi\) and global scale \(\tau\), write \(\kappa=(1+\tau^2u)^{-1}\) and \(z=\log\{\kappa/(1-\kappa)\}\), so that \(\kappa(z)=\{1+\exp(-z)\}^{-1}\) and \(u(z)=\exp(-z)/\tau^2\). After this transformation, the posterior weight on the \(z\)-scale is proportional to
\[
\omega_{\pi,\tau}(z;y)
=
\pi\{u(z)\}
\frac{1-\kappa(z)}{\tau^2\kappa(z)^{1/2}}
\exp\{-\kappa(z)y^2/2\},
\]
and the posterior mean is evaluated as
\[
\widehat\theta_{\pi,\tau}(y)
=
y\,
\frac{\int\{1-\kappa(z)\}\omega_{\pi,\tau}(z;y)\,dz}
{\int\omega_{\pi,\tau}(z;y)\,dz}.
\]
This transformation avoids direct integration over the unbounded local-variance scale and allows the same numerical scheme to be used for all priors considered below.

The \(z\)-grid begins at \(-28\) and ends at \(\max\{28,-2\log\tau+24\}\), with target spacing \(0.015\). At least 2501 grid points are used; the maximum is 18001 for Figure 1 of the main paper and 14001 for Figures 2--3. Trapezoidal weights are used on this grid. On the observation scale, the theoretical oracle threshold \(a_\star=\{2L+\log L-\log\pi\}^{1/2}\) determines an adaptive grid, with spacing \(0.004\) within distance three of \(a_\star\) and \(0.025\) away from this transition region. Posterior quantities are evaluated in blocks of 60 observation-grid points to control memory use.

For fixed \((\pi,\tau,L)\), the effective numerical transition threshold \(\widehat a_{\rm eff}\) is obtained from the positive solution of \(\widehat\theta_{\pi,\tau}(y)/y=1/2\). The active-coordinate risk at signal \(\vartheta\) is
\[
R_\vartheta(\pi,\tau)
=
\int
\{\widehat\theta_{\pi,\tau}(\vartheta+z)-\vartheta\}^2
\phi(z)\,dz.
\]
Since \(p/s=e^L\), the total worst-case risk divided by \(s\) is evaluated as \((e^L-1)R_0(\pi,\tau)+R_\star(\pi,\tau)\), where \(R_\star(\pi,\tau)=\sup_{\vartheta\geq0}R_\vartheta(\pi,\tau)\). Gaussian integrations are truncated at \(|z|\leq10\), with relative tolerance \(2\times10^{-8}\) and absolute tolerance \(10^{-10}\). The active-risk maximizer is first localized on a 61-point signal grid around the transition region using an 80-point Gauss--Hermite rule. The numerically integrated risk is then maximized over the interval adjacent to the coarse maximizer by one-dimensional optimization with tolerance \(2\times10^{-5}\). Numerical integration allows up to 400 subdivisions for Figure 1 of the main paper and 300 subdivisions for Figures 2--3.

\subsection{Finite-sparsity calibration phase diagram}
\label{app:phase-computation}

Figure 1 of the main paper uses the horseshoe prior and the critical polynomial calibration \(r=1\). Since \(\alpha=1/2\) and \(\gamma=0\), the effective logarithmic index \(d=2\alpha b+\gamma\) reduces to \(d=b\), and the numerical calibration is \(\tau=e^{-L}L^d\). We use
\[
d\in\{0.5,1,1.25,1.5,1.75\},
\quad
L\in\{3,4,5,6,8,10,12,14,16,18,21,24,28,32,40,50,65,80,100,125,160,200\}.
\]
For each pair \((L,d)\), the procedure in Section~\ref{app:sequence-computation} yields the worst-case risk, null-coordinate contribution, least-favourable signal and effective transition threshold.

Panel (a) reports \(R_{p,s}(\pi,\tau)/(2sL)\) for all five values of \(d\). Panel (b) reports
\[
\frac{R_{p,s}(\pi,\tau)-2sL}{s\{L\log L\}^{1/2}}
+
2\sqrt{2}
\]
for \(d\leq1.5\). The small-\(L\) insets use the same numerical calculations restricted to \(3\leq L\leq15\); no separate approximation is introduced for the inset displays.

\subsection{Oracle calibration and higher-order ranking}
\label{app:oracle-computation}

Figures 2--3 of the main paper use the horseshoe, horseshoe+, half-\(t_{1.5}\) and beta-prime\((1/2,1)\) local-variance priors. The calculations are performed at
\[
L\in\{3,4,5,6,7,8,9,10,12,14,16,18,21,24,28,32,40,50\}.
\]
For each prior and each \(L\), the finite-\(L\) oracle global scale is obtained by minimizing the numerically evaluated worst-case risk over a broad interval on the log scale. The search range is \(-L-8\leq\log\tau\leq-L/2+8\). A 17-point equally spaced coarse grid is first evaluated over the entire interval. If the minimum coarse-grid risk occurs at either endpoint, the calculation is stopped rather than extrapolated beyond the prespecified range. Otherwise, the risk is minimized over the interval between the two grid points adjacent to the coarse minimizer, with optimization tolerance \(8\times10^{-4}\). The final minimizer is additionally required to lie more than \(2\%\) of the total search width from either boundary. All minimizers used in Figures 2--3 of the main paper satisfy this interiority check.

The theoretical oracle scale \(\tau_\pi^\star\) in Theorem 2 of the main paper is computed separately and is not used to restrict the numerical search. Figure 2(a) of the main paper reports \(\widehat\tau^{\rm or}/\tau_\pi^\star\). At the numerically optimized scale, the effective threshold \(\widehat a_{\rm eff}\) is obtained from \(\widehat\theta_{\pi,\widehat\tau^{\rm or}}(y)/y=1/2\), and Figure 2(b) reports
\[
\widehat a_{\rm eff}^{\,2}
-
\{2L+\log L-\log\pi\}.
\]
The least-favourable signal and the corresponding oracle risk are also retained from each optimization and are used below in the regression experiment.

Figure 3 of the main paper applies the same unrestricted numerical oracle calculation independently at every \(L\). For the horseshoe and horseshoe+, panel (a) reports \(L^2\{R_{\rm H+}^{\rm or}-R_{\rm H}^{\rm or}\}/s\). Panel (b) reports \(L\{R_{\rm HT}^{\rm or}-R_{\rm H}^{\rm or}\}/s\) for the half-\(t_{1.5}\) prior and \(L\{R_{\rm BP}^{\rm or}-R_{\rm H}^{\rm or}\}/s\) for the beta-prime\((1/2,1)\) prior. The asymptotic reference values are \(2\), \(-1\) and \(-2\), respectively, as predicted by Theorem 3 of the main paper. The finite-\(L\) horseshoe--horseshoe+ crossover is estimated from the unnormalized oracle-risk difference by locating the first adjacent pair of \(L\)-grid points having opposite signs and linearly interpolating between them.

\subsection{High-dimensional regression stability experiment}
\label{app:regression-computation}

Figure 4 of the main paper examines whether finite-\(L\) sequence-model risk differences remain visible in high-dimensional Gaussian regression. We fix the sequence benchmark at \(L_0=5\), take \(s=3\) and \(p=445\), and hence have \(\log(p/s)=4.99946\). The sample-size grid is
\[
n\in\{80,120,180,260,360\},
\]
so \(p>n\) throughout. The entries of the design matrix are generated independently from \(N(0,1)\), after which each column is rescaled to have Euclidean norm \(\sqrt n\). Writing \(\boldsymbol Z=\boldsymbol X/\sqrt n\), the regression is implemented on the sequence scale as \(\boldsymbol Y=\boldsymbol Z\boldsymbol\theta^0+\boldsymbol\varepsilon\), with \(\boldsymbol\varepsilon\sim N_n(\boldsymbol0,\boldsymbol I_n)\).

The global scales are not reoptimized in the regression experiment. For the horseshoe, horseshoe+ and half-\(t_{1.5}\) priors, we use the finite-\(L\) sequence-model oracle scales obtained at \(L_0=5\) by the procedure in Section~\ref{app:oracle-computation}. The corresponding finite-\(L\) least-favourable sequence signals are retained, and each prior is evaluated at signal magnitudes \(0.9\widehat\vartheta_\pi^\star\), \(\widehat\vartheta_\pi^\star\) and \(1.1\widehat\vartheta_\pi^\star\). Within each Monte Carlo replication, a common active support of size \(s\), independent Rademacher signs on that support and a common Gaussian noise vector are generated. The same design, support, signs and noise realization are used for all priors, so prior-risk differences are estimated from paired replications.

Posterior computation is carried out on the \(\boldsymbol\theta\)-scale. With \(v_j=\log u_j\), the Gaussian full conditional for \(\theta_j\) has variance \(V_j=\tau^2e^{v_j}/(1+\tau^2e^{v_j})\) and mean \(V_j\boldsymbol Z_j^T\boldsymbol r_{-j}\), where \(\boldsymbol r_{-j}\) is the residual with the current \(j\)th contribution removed. After the coefficient update, \(v_j\) is updated by stepping-out slice sampling. Up to additive constants independent of \(v\), the conditional log densities are
\[
\begin{aligned}
\ell_{\rm H}(v\mid\theta)
&=
-\log(1+e^v)
-
\frac{\theta^2e^{-v}}{2\tau^2},\\
\ell_{\rm H+}(v\mid\theta)
&=
\log\left\{\frac{v}{e^v-1}\right\}
-
\frac{\theta^2e^{-v}}{2\tau^2},\\
\ell_{{\rm HT}_\nu}(v\mid\theta)
&=
-\frac{\nu+1}{2}
\log\left(1+\frac{e^v}{\nu}\right)
-
\frac{\theta^2e^{-v}}{2\tau^2}
\end{aligned}
\]
with $\nu=1.5$. For the horseshoe+, the first term is evaluated by its continuous expansion near \(v=0\) to avoid numerical cancellation. The slice sampler uses initial bracket width 2 and at most 50 stepping-out increments. All chains are initialized at \(\theta_j=0\) and \(u_j=1\). Each fit uses 1000 burn-in sweeps followed by 1500 retained sweeps without thinning, and the posterior mean is computed by averaging the retained coefficient draws.

For each \(n\), 50 independent regression replications are generated. At each signal multiplier, the loss is \(\|\widehat{\boldsymbol\theta}-\boldsymbol\theta^0\|_2^2=n\|\widehat{\boldsymbol\beta}-\boldsymbol\beta^0\|_2^2\). For each prior, the Monte Carlo mean loss is first computed separately at the three signal magnitudes and the maximum is then used as the numerical least-favourable-profile risk. The polynomial-tail comparison in Figure 4(a) of the main paper is normalized as \(L_0\{\widehat R_{\rm HT}^{X}-\widehat R_{\rm H}^{X}\}/s\), whereas the log-tail comparison in Figure 4(b) is normalized as \(L_0^2\{\widehat R_{\rm H+}^{X}-\widehat R_{\rm H}^{X}\}/s\). The dashed horizontal references are the corresponding finite-\(L_0\) sequence-model oracle differences used in Figure 3 of the main paper, rather than the limiting coefficients from Theorem 3. Thus the experiment isolates the additional discrepancy associated with replacing the sequence model by a nonorthogonal regression design.

Monte Carlo uncertainty in the pairwise differences is assessed by a paired nonparametric bootstrap over the 50 replication indices. For each of 1000 bootstrap samples, the same resampled indices are used for both priors in a comparison. The mean loss is recomputed at each of the three signal multipliers, the maximum over the signal grid is taken separately for each prior, and the normalized pairwise difference is recorded. The error bars in Figure 4 of the main paper are the resulting \(2.5\%\) and \(97.5\%\) bootstrap quantiles.

Finite-sample diagnostics are also recorded for the two sparse-design quantities appearing in Theorem 4 of the main paper. With \(\boldsymbol G=\boldsymbol X^T\boldsymbol X/n\), the cumulative active-set coherence \(\Gamma_{n,s}\) is computed exactly by summing, for each row, its \(s-1\) largest absolute off-diagonal entries and taking the largest resulting row sum. The restricted-isometry quantity \(\Delta_{n,2m_n}\) is combinatorial and is not evaluated exactly. Instead, with \(m_n=\lceil2sL_0\rceil=30\), an empirical sparse-distortion diagnostic is calculated over subsets of size \(2m_n=60\). For each replication, 300 uniformly sampled subsets are examined together with 20 targeted subsets centred at coordinates having the largest cumulative pairwise correlations. For each subset \(T\), the recorded distortion is \(\|\boldsymbol G_{TT}-\boldsymbol I\|_{\rm op}\), and the largest observed value is denoted by \(\widehat\delta_n^{\,\rm emp}\). This quantity is used only as a finite-sample diagnostic and is not identified with the exact \(\Delta_{n,2m_n}\). The accompanying combined diagnostic is \(\Gamma_{n,s}+(\widehat\delta_n^{\,\rm emp})^2\).

\bibliographystyle{chicago}
\bibliography{Ref}

\end{document}